\documentclass[11pt]{memoir}

\setlrmarginsandblock{1.5in}{*}{*}
\checkandfixthelayout

\counterwithout{section}{chapter}
\counterwithin{equation}{section}
\counterwithin{figure}{section}

\makeatletter
\renewcommand{\@memmain@floats}{%
  \counterwithin{figure}{section}
  \counterwithin{table}{section}}
\makeatother

\firmlists

\usepackage{hyperref}
\usepackage{memhfixc} 
\usepackage[numbered]{bookmark} 
\usepackage[all]{hypcap} 
\usepackage[algo2e,algosection,tworuled,noend,noline]{algorithm2e}
\usepackage[charter]{gacs}
\usepackage{gacs-algo} 

\theoremstyle{definition} 

\renewcommand{\le}{\leq}
\renewcommand{\ge}{\geq}

\renewcommand{\vek}[1]{\mathbf{#1}}
\newcommand{\fld}[1]{\ensuremath{\textit{#1\/}}}
\newcommand{\rul}[1]{\ensuremath{\texttt{#1}}}

\newcommand{\instr}{\item[]\hspace{-0.6em}}

\def\B{B}  
\def\U{U}
\def\V{V}

\newcommand{\va}{\vek{a}} 
\newcommand{\blank}{\text{\textvisiblespace}}
\newcommand{\Configs}{\mathrm{Configs}}

\newcommand{\E}{E} 

\makeatletter
\if@charter\else\fi
\makeatother
\newcommand{\escno}{q}
\newcommand{\F}{F}

\def\G{G} 
\newcommand{\h}{h} 
\newcommand{\hc}{\hat h}
\newcommand{\vhc}{\vek{\hat h}}
\newcommand{\Int}{\mathrm{Int}} 
\newcommand{\Noise}{\mathit{Noise}}
\newcommand{\passno}{\pi}
\newcommand{\pos}{\mathrm{pos}}
\newcommand{\curcell}{\textrm{cur-cell}}
\newcommand{\Q}{Q} 
\newcommand{\Rg}{R} 
\newcommand{\s}{s} 
\renewcommand{\S}{S} 
\newcommand{\Tu}{T} 
\newcommand{\Tus}{T^{*}}
\newcommand{\Z}{Z} 
\newcommand{\z}{z} 

\newcommand{\tape}{\mathrm{tape}}
\newcommand{\Interpr}{\mathrm{Interpr}} 
\newcommand{\Decode}{\mathrm{Decode}}
\newcommand{\Encode}{\mathrm{Encode}}

\newcommand{\Histories}{\mathrm{Histories}}

\newcommand{\PadLen}{\mathit{PadLen}} 

\newcommand{\Left}{\text{left}}
\newcommand{\Right}{\text{right}}

\renewcommand{\r}{\vek{r}} 
\newcommand{\x}{\vek{x}} 
\newcommand{\y}{\vek{y}} 

\newcommand{\Addr}{\fld{Addr}}
\newcommand{\Age}{\fld{Age}} 
\newcommand{\Base}{\fld{Base}}
\newcommand{\BigDigression}{\fld{BigDigression}}
\newcommand{\Core}{\fld{Core}}
\newcommand{\Drift}{\fld{Drift}}
\newcommand{\FrontAddr}{\fld{FrontAddr}}
\newcommand{\Half}{\fld{Half}} 
\newcommand{\Heal}{\fld{Heal}} 
\newcommand{\Hold}{\fld{Hold}}
\newcommand{\Index}{\fld{Index}}
\newcommand{\Info}{\fld{Info}}
\newcommand{\Kind}{\fld{Kind}}
\newcommand{\Output}{\fld{Output}}
\newcommand{\Payload}{\fld{Payload}}
\newcommand{\Pass}{\fld{Pass}} 
\newcommand{\Pos}{\fld{Pos}}
\newcommand{\Rebuild}{\fld{Rebuild}} 
\newcommand{\Replace}{\fld{Replace}} 
\newcommand{\Sweep}{\fld{Sweep}} 
\newcommand{\Tape}{\fld{Tape}}
\newcommand{\Work}{\fld{Work}} 

\newcommand{\Bad}{\mathrm{Bad}}
\newcommand{\Vacant}{\mathrm{Vac}}
\newcommand{\New}{\mathrm{New}}
\newcommand{\new}{\mathrm{new}}

\newcommand{\Stem}{\mathrm{Stem}}
\newcommand{\Booting}{\mathrm{Booting}}
\newcommand{\Bridge}{\mathrm{Bridge}}

\newcommand{\Member}{\mathrm{Member}}
\newcommand{\Outer}{\mathrm{Outer}}

\newcommand{\Compute}{\rul{Compute}}
\newcommand{\rHeal}{\rul{Heal}}
\newcommand{\MoveFront}{\rul{MoveFront}}
\newcommand{\ProcessPayload}{\rul{ProcessPayload}}
\newcommand{\rBoot}{\rul{Boot}}
\newcommand{\rRebuild}{\rul{Rebuild}}
\newcommand{\rRebuildHeal}{\rul{RebuildHeal}}
\newcommand{\rSimulate}{\rul{Simulate}}

\newcommand{\WriteProgramBit}{\rul{WriteProgramBit}}

\newcommand{\True}{\mathrm{True}}
\newcommand{\False}{\mathrm{False}}
\newcommand{\cns}[1]{c_{\textrm{\upshape #1}}}
\newcommand{\CEsc}{\cns{esc}}
\newcommand{\CBoundaries}{\cns{bndr}}
\newcommand{\CMarg}{\cns{marg}}

\newcommand{\CRebuild}{\cns{Rebuild}}

\newcommand{\CRelief}{\cns{relief}}

\newcommand{\CSpill}{\cns{spill}}
\newcommand{\CStain}{\cns{stain}}

\begin{document}

\title{A reliable Turing machine}

\author{Ilir \c{C}apuni 
\\ University of Montenegro
\\ ilir@bu.edu
\and
Peter G\'acs
\\ Boston University
\\ gacs@bu.edu
}
\maketitle

\begin{abstract}
  We consider computations of a Turing machine subjected to noise.
  In every step, the action (the new state and the new content of the observed
  cell, the direction of the head movement) can differ from that prescribed by
  the transition function with a small probability (independently of previous
  such events).  We construct a universal 1-tape Turing machine that for a low enough
  (constant) noise probability performs arbitrarily large computations.  For
  this unavoidably, the input needs to be encoded---by a simple code depending
  on its size.  The work uses a technique familiar from reliable cellular
  automata, complemented by some new ones.
\end{abstract}

\newpage


\section{Introduction}

This work addresses a question from the area of ``reliable computation with unreliable components''.
A certain class of machines is chosen (like a Boolean circuit, cellular automaton, Turing machine).
It is specified what kind of faults (local in space and time)
are allowed, and a machine of the given
kind is built that---paying some price in performance---carries out essentially the same
task as any given machine of the same kind without any faults would.

We confine attention to \emph{transient}, \emph{probabilistic} faults:
the fault occurs at a given time but no component is damaged permanently,
and faults occur independently of each other, with a bound on their probability.
This is in contrast to bounding the \emph{number} of faults, allowing them to be set by an adversary.
Historically the first result of this kind is~\cite{VonNeum56}, which for each Boolean
circuit \( C \) of size \( n \) constructs a new circuit \( C' \)
of size \( O(n\log n) \) that performs the same task as \( C \) with a (constant)
high probability, even though each gate of
\( C' \) is allowed to fail with some (constant) small probability.

Cellular automata as a model have several theoretical advantages over
Boolean circuits, and results concerning
reliable computation with them have interest also from a purely mathematical or 
physical point of view (non-ergodicity).
The simple construction in~\cite{GacsReif3dim88} gives, for any 1-dimensional
cellular automaton \( A \) a 3-dimensional cellular automaton \( A' \) 
that performs the same task as \( A \) with high probability, even though each cell of
\( A' \) is allowed to fail in each time step with some constant small probability.
(A drawback of this construction is the requirement of synchronization: all cells of \( A' \) must
update simultaneously.)
Reliable cellular automata in less than 3 dimensions can also be constructed (even
without the synchrony requirement), but so far only at a
steep increase in complexity (both of the construction and the proof).
The first such result was~\cite{Gacs1dim86}, relying on some ideas proposed in~\cite{Kurd78}.

Here, the reliability question will be considered for a \emph{serial}
computation model---a 1-tape Turing machine---as opposed to parallel ones like
Boolean circuits or cellular automata.
There is a single elementary processing unit (the \df{active} unit)
interacting with a memory of unlimited size.
The error model needs to be relaxed.
Allowing each memory component to fail in each time step with constant probability
makes, in the absence of parallelism, reliable computation seem impossible.
Indeed, while in every step some constant fraction of the memory gets corrupted,
the active unit can only correct a constant \emph{number} of them per step.

\begin{remark}
  The choice a single tape for the Turing machine seems unnecessarily restricting,
  given how time-consuming it is to even compare two strings on such machines.
  However, having two tapes would add to the physical implausibility, assuming
  some kind of unlimited-length safe connection from the heads to a 
  processing unit (or at least between each other).
\end{remark}

In the relaxed model considered here, faults can affect only the operation of the active unit.
More precisely at any given time the allowed operations of the machine are the usual ones:
changing its state, writing to the observed tape cell, moving the head by a step left or right (or not at all).
The transition table of the machine prescribes which action to take.
So our fault model is the following.

\begin{definition}
  Let \( \Noise \) be a random subset of some set \( U \).
  We will say that the distribution of \( \Noise \) is \( \eps \)-\df{bounded} if for every finite subset \( A \)
  we have
  \begin{align*}
   \Pbof{A\subseteq\Noise} \le \eps^{|A|}.
  \end{align*}
\end{definition}
See~\cite{Toom80} for an earlier use of this kind of restriction.

\begin{definition}\label{def:faults-eps-bounded}
  Let \( \cC = (C_{1},C_{2}\dots) \) be  a random sequence of configurations of a Turing machine \( T \)
  with a given fixed
  transition table, with the property that for each time \( t \), the \( C_{t+1} \) is obtained from \( C_{t} \) by
  one of the allowed operations.
  We say that a \df{fault} occurred at time \( t \) if the operation giving \( C_{t+1} \) from \( C_{t} \) is
  not obtained by the transition function.
  Let \( \Noise\subseteq\bbZ_{+} \) be the (random) set of faults in the sequence.
  We say that faults of the sequence \( \cC \) are \( \eps \)-\df{bounded}, if the set \( \Noise \) is.  
\end{definition}

The challenge for Turing machines
is still significant, since even if only with small probability, occasionally a group
of faults can put the head into the middle of a large segment of the tape rewritten
in an arbitrarily ``malicious'' way.
A method must be found to recover from all these situations.

Here we define a Turing machine that is reliable---under this fault model--- in
the same sense as the other models above.
The construction and proof are similar in complexity to the ones for 1-di\-men\-sion\-al cellular automata;
however, we did not find a reduction to these earlier results.
A natural idea is to let the Turing machine simulate a 1-dimensional cellular automaton, by
having the head make large sweeps, and update the tape as the simulated cellular automaton would.
But apart from the issue of excessive delay, we did not find any simple
way to guarantee the large sweeps in the presence of faults (where ``simple'' means not building some new
hierarchy), even for some price paid in efficiency.
So here we proceed ``from scratch''.

Many ideas used here are taken from~\cite{Gacs1dim86} and~\cite{GacsSorg01},
but hopefully in a somewhat simpler and more intuitive conceptual framework.
Like~\cite{Gacs1dim86} it confines all probability reasoning to a single lemma,
and deals on each level only with a numerical restriction on faults (of that level).
On the other hand, like~\cite{GacsSorg01} it defines a series of generalized objects
(generalized Turing machines here rather than generalized cellular automata),
each one simulating the next in the series.
In~\cite{GacsSorg01} a ``trajectory'' (central for defining the notion of simulation) was
a random history whose distribution satisfies certain constraints (most of which are combinatorial).
Here it is a single history satisfying only some combinatorial constraints.

The work~\cite{AsarinCollins2005} seems related by its title 
but is actually on another topic.
The work~\cite{DurandRomashShenTiling12} applies the self-simulation and
hierarchical robustness technique developed for cellular automata in an interesting, but
simpler setting.
Several attempts at the chemical or
biological implementation of universal computation have to deal with
the issue of error-correction right away.
In these cases generally the first issue is the faults occurring at the active site (the head).
See~\cite{BennettThermodynComp1982,QianSoloveichikWinfree2011}.

Our result can use any standard definition of 1-tape Turing machines whose tape alphabet
contains some fixed ``input-output alphabet'' \( \Sigma \); we will
introduce one formally in Section~\ref{sec:TM}.
We will generally view a tape symbol as a tuple consisting of several \df{fields}.
The notation
\begin{align*}
   a.\Output, a.\Info
 \end{align*}
 shows \( \Output \) and \( \Info \) as fields of tape cell state \( a \).
Combining the same field of all tape cells, we can talk about a \df{track}
(say the \( \Output \) track and \( \Info \) track).
For ease of spelling out a result, we consider only computations whose outcome
is a single symbol, written into the \( \Output \) field of tape position 0.
It normally holds a special value---say \( * \) ---meaning \df{undefined}.

Block codes (as defined in Section~\ref{sec:codes} below) are specified by a
pair \( \pair{\psi_{*}}{\psi^{*}} \) of encoding and decoding functions.
In the theorem below, the input of the computation, of some length \( n \), is broken up into blocks
that are encoded by a block code that depends in some simple way on \( n \).
Its redundancy depends on the size of the input as a log power.
The main result in the theorem below shows a Turing machine simulating a fault-free
Turing machine computation in a fault-tolerant way.
It is best to think of the simulated machine \( G \) as some universal Turing machine.

\begin{theorem}\label{thm:main}
  For any Turing machine \( \G \) there are constants \( \alpha_{1},\alpha_{2}>0 \),
  for each input size \( n \) a
  block code \( (\varphi_{*}, \varphi^{*}) \) of block size \( O((\log n)^{\alpha_{1}}) \),
a fault bound  \( 0\le\eps <1 \) and a Turing machine \( M_{1} \) with a 
function \( a\mapsto a.\Output \) defined on its alphabet,
such that the following holds.

Let \( M_{1} \) start its work from the initial tape configuration \( \varphi_{*}(x) \) with the head
in position 0,
running through a random sequence of configurations whose faults are \( \eps \)-bounded in the sense
of Definition~\ref{def:faults-eps-bounded}.
Suppose that at time \( t \) the machine \( \G \) writes a value \( y\ne * \) 
into the \( \Output \) field of the cell at position 0.
Then at any time greater than \(    t(\log t)^{\alpha_{2}\log\log\log t} \),
the tape symbol \( a \) of machine \( M_{1} \) at position 0
 will have \( a.\Output= y \) with probability at least \( 1 - O(\eps) \).
\end{theorem}

\section{Overview}

The overview, but even the main text, is not separated completely
into two parts, namely the definition of the Turing machine, followed by the proof of its reliability.
The definition of the machine is certainly translatable (with a lot of tedious work) into just a Turing
machine transition table (or ``program''), but its complexity requires first to develop a conceptual
apparatus behind it which is also used in the proof of reliability.
We will try to indicate below at the beginning of each
section whether it is devoted more to the program or more to the conceptual apparatus.

\subsection{Isolated bursts of faults}\label{sec:bursts}

Let us introduce some basic elements of the \emph{program}.
In~\cite{burstyTuring13} we defined a Turing machine \( M_{1} \) that simulates ``reliably'' any other
Turing machine even when it is subjected to isolated ``bursts'' of faults (that is a group
of faults occurring in consecutive time steps) of constant size.
We will use some of the ideas of~\cite{burstyTuring13}, without relying directly
on any of its details, and will add several new ideas.
Here is a brief overview of this machine \( M_{1} \).

Each tape cell of the simulated machine \( M_{2} \) will be represented by a block of
some size \( \Q \) called a \df{colony}, of the simulating machine \( M_{1} \).
Each step of \( M_{2} \) will be simulated by a computation of \( M_{1} \) called
a \df{work period}.
During this time, the head of \( M_{1} \) moves around over the
current colony-pair, decodes the represented cell symbols,
then computes and encodes the new symbols, and finally moves the head 
to the new position of the head of \( M_{2} \).
The major processing steps will be 
carried out on a working track three times within one work period,
recording the result onto separate tracks.
The information track is changed only in a final majority vote.

The organization is controlled by a few key fields, for example a field
called \( \Addr \) showing the position of each cell in the colony, and a field
\( \Age \), the number of the last step of the computation
that has been performed already.
The most technical part is to protect this control information from faults.
To discover such structural disruptions locally
before the head would go far in the wrong direction,
the head will make frequent short zigzags.
Any local inconsistency will be detected this way, triggering the healing procedure.


\subsection{Hierarchy}\label{sec:hier}

Here, we start the development of the \emph{conceptual apparatus}.
In order to build a machine resisting faults 
occurring independently in each step with some small probability,
we take the approach used for one-dimensional cellular automata.
We aim at building a \df{hierarchy of simulations}:
machine \( M_{1} \) simulates machine \( M_{2} \) which simulates machine \( M_{3} \), and so on.
Machine \( M_{k} \) has alphabet
\begin{align}\label{eq:Sigma_k}
   \Sigma_{k}=\{0,1\}^{s_{k}},
\end{align}
that is its tape cells have some ``capacity'' \( s_{k} \).
All these machines should be implementable on a universal Turing machine with
the same program (with an extra input, the number \( k \) denoting the level).
For ease of analysis, we introduce the notion of \df{cell size}:
level \( k \) has its own cell size \( \B_{k} \) and block (colony) size \( \Q_{k} \)
with \( B_{1}=1 \), \( \B_{k+1}=\B_{k}\Q_{k} \).
This allows
locating each tape cell of \( M_{k} \) on the same interval where the cells of \( M_{1} \) simulate it.
One cell of machine \( M_{k+1} \) is simulated by a colony of machine \( M_{k} \);
so one cell of \( M_{3} \) is simulated by \( \Q_{1}\Q_{2} \) cells of \( M_{1} \).
Further, one step of, say, machine \( M_{3} \) is simulated by one
work period of \( M_{2} \) of, say, \( O(\Q_{2}^{2}) \) steps.

Per construction, machine \( M_{1} \) can withstand
bursts of faults with size  \( \le \beta \) for some constant parameter \( \beta \),
separated by at least some constant number \( \gamma \) of work periods.
It would be natural now to expect that machine
\( M_{1} \) can withstand also some \emph{additional}, larger bursts
of size \( \le \beta \Q_{1} \) if those are separated
by at least \( \gamma \) work periods of \( M_{2} \).
However, a new obstacle arises.
Damage caused by a big burst of faults spans several colonies.
The repair mechanism of machine \( M_{1} \) outlined in Section~\ref{sec:bursts} 
is too local to recover from such extensive damage, leaving
the whole hierarchy endangered.
So we add a new mechanism to \( M_{1} \) that
will just try to restore the colony \emph{structure} of a large enough portion of the
tape (of the extent of several colonies).
The task of restoring the original \emph{information} is left to higher levels (whose simulation
now can continue).

All machines above \( M_{1} \) in the hierarchy live only in simulation: the hardware is \( M_{1} \).
Moreover, the \( M_{k} \) with \( k>1 \)
will not be ordinary Turing machines, but \df{generalized} ones,
with some new features seeming necessary in a simulated Turing machine:
allowing for some ``disordered'' areas of the tape not obeying the transition function,
and occasionally positive distance between neighboring tape cells.

A tricky issue is ``forced self-simulation''.
Each machine \( M_{k} \) can be implemented on a universal machine using as inputs
the pair \( (p,k) \) where \( p \) is the common program and \( k \) is the level.
Eventually, \( p \) will just be hard-wired into the definition of \( M_{1} \),
and therefore faults cannot corrupt it.
While creating \( p \) for machine \( M_{1} \),
we want to make it simulate a machine \( M_{2} \) that has the same program \( p \).
The method to achieve this has been
applied already in some of the cellular automata and tiling papers cited, 
and is related to the proof of Kleene's fixed-point theorem (also called the recursion theorem).

Forced self-simulation can give rise to an infinite sequence of simulations, achieving
the needed robustness.
Let us point out that fixing the program of self-simulation does not prevent universality.
A track  (which we will call \( \Payload \)) will be set aside for simulating the machine
\( \G \) of Theorem~\ref{thm:main}.
If this simulation of \( \G \) does not finish in a certain number of steps,
a built-in mechanism will \df{lift} its tape content to the \( \Payload \)
field of the simulated cell-pair, allowing it to be continued in a colony-pair of the next
level (with the corresponding higher reliability).

\subsection{Structuring the noise}\label{sec:sparsity}

From the probabilistic assumptions about the noise, one can draw some combinatorial conclusions.
This part of the work is rather simple and self-contained, and is similar to
some earlier publications on these topics.

The set of faults in the noise model of the theorem is a set of points in time.
It turns out more convenient to use an equivalent model:
an \( \eps \)-bounded \emph{space-time} set of points.
Let us make this statement more formal.

\begin{lemma}\label{lem:space-time-noise}
  Let \( \cC=(C_{1},C_{2},\dots) \) be the random sequence of configurations of a Turing machine
  with an \( \eps \)-bounded set of faults \( \Noise_{1}\subseteq\bbZ_{+} \),
  as in Definition~\ref{def:faults-eps-bounded}.
  Let \( h(t) \) be the (random) position of the head at time \( t \).
  Then the random set \( \Noise_{2}=\setOf{(h(t),t)}{t\in\Noise_{1}} \).
is an \( \eps \)-bounded subset of \( \bbZ\times\bbZ_{+} \).
\end{lemma}
\begin{proof}
  Let \( A \) be a finite subset of  \( \bbZ\times\bbZ_{+} \), and
  \( A' = \setOf{t}{(p,t)\in A} \).
  If \( A\subseteq\Noise_{2} \) then \( A'\subseteq\Noise_{1} \) and \( |A'|=|A| \).
  Hence 
\begin{align*}
 \Pbof{A\subseteq\Noise_{2}}\le\Pbof{A'\subseteq\Noise_{1}}\le \eps^{|A'|}=\eps^{|A|} .
\end{align*}
\end{proof}

The construction outlined above counts with \emph{bursts} (rectangles of space-time
containing  \( \Noise \)) increasing in size and decreasing
in frequency---which is a combinatorial set of constraints.
To derive such constraints from the above probabilistic model
the  we stratify \( \Noise \) as follows.
We will have two series of parameters:  \( \B_{1}<\B_{2}<\dotsm \) and
\( \S_{1}<\S_{2}<\dotsm \).
Here \( \B_{k} \) is the size of cells of \( M_{k} \) as represented on the tape of \( M_{1} \),
and \( \S_{k} \) is a (somewhat increased)
bound on the time needed to simulate one step of  \( M_{k} \).

Here are some informal definitions.
For some constants \( \beta,\gamma>1 \),
a \df{burst} of noise of type \( \pair{a}{b} \)
is a space-time set that is coverable by a rectangle of size  \( a\times b \).
For an integer \( k>0 \) it is of \df{level} \( k \) when it is of type \( \beta(\B_{k},\S_{k}) \).
It is \df{isolated} if it is 
essentially alone in a rectangle of size \( \gamma(\B_{k+1}\times\S_{k+1}) \) 
First we remove such isolated bursts of level 1, then of level 2
from the remaining set, and so on.
It will be shown that with not too fast increasing sequences \( \B_{k},\S_{k} \), with probability 1,
this infinite sequence of operations completely erases \( \Noise \): thus each fault belongs to
a burst of ``level'' \( k \) for some \( k \).

Machine \( M_{k} \) will concentrate only on correcting isolated bursts of level \( k \) and on restoring
the framework allowing \( M_{k+1} \) to do its job.
It can ignore the lower-level bursts and will need to work correctly
only in the absence of higher-level bursts.

If we modeled noise as a set of time points then a burst of faults would be a time
interval of size \( \beta\S_{k} \) and might affect a space interval as large as \( \beta\S_{k} \),
covering many times more simulated cells of level \( k \).
Therefore we model noise as a set of space-time points; by Lemma~\ref{lem:space-time-noise},
this does not change the independence assumption of the main theorem.

\begin{definition}\label{def:isolation}
Let \( \r=\pair{r_{1}}{r_{2}} \), \( r_{1}, r_{2}> 0 \)
be a two-dimensional nonnegative vector.
A \df{rectangle of ``radius''} \( \r  \) \df{centered} at point  \( \x \) is
\begin{align}\label{eq:ball1}
  \bB(\x,\r) = \setOf{\y}{\abs{y_{i} - x_{i}} < r_{i}, i=1,2}.
\end{align}  
Let \( E\subseteq \bbZ\times\bbZ_{\ge 0} \) be a space-time set (to be considered our noise set).
A point \( \x \) of \( E \) is \df{\( \pair{\r}{\r^{*}} \)-isolated} if
\(  E \cap \bB(\x,\r^{*})\subseteq \bB(\x, \r)  \),
that is all points of \( E \) that are \( \r^{*} \)-close to \( x \) are also \( \r \)-close.
A set \( E \) is called \( \pair{\r}{\r^{*}} \)-\df{sparse} if each of its points is \( \pair{\r}{\r^{*}} \)-isolated.
\end{definition}

The following lemma will justify talking about bursts of faults.

\begin{lemma}[Bursts]\label{lem:bursts}
  Suppose that the set \( E \) is \( \pair{\r}{\r^{*}} \)-sparse with \( \r^{*}>2\r \).
  For an element \( \x\in E \) let \( E_{\x} = \bB(\x,\r)\cap E \).
  Each set \( E_{\x} \) is contained in a rectangle of size \( r_{1}\times r_{2} \).
  Every rectangle of size \( (r^{*}_{1}-r_{1})\times (r^{*}_{2}-r_{2}) \) intersects
  with at most one of the sets \( E_{\x} \).
\end{lemma}
\begin{proof}
  We introduce a relation \( \x\sim\y \eqv \y\in E_{\x} \) between elements of the set \( E \).
  The relation is clearly symmetric and reflexive, but we claim that it is also transitive.
  Indeed, suppose that \( \y\sim \x \) and \( \y'\sim\x \),
  Given that \( \r\le 2\r^{*} \), we have \( \y'\in \bB(\y,\r^{*}) \),
  and then by sparsity, \( \y'\in \bB(\y,\r) \).

  By the relation \( \sim \) we can partition the set \( E \)
  into subsets of the form \( E_{\x} \).
  We claim that each of these sets is coverable by a rectangle of size \( r_{1}\times r_{2} \).
  Indeed, suppose this is not so: then either the horizontal projection of \( E_{\x} \)
  is \( \ge r_{1} \) or the vertical one is \( \ge r_{2} \): suppose the former.
  Then there are elements \( \y,\y'\in E_{x} \) whose horizontal distance is \( \ge r_{1} \)
  contrary to \( \y'\sim\y \).

  Suppose that some rectangle of size \( (r^{*}_{1}-r_{1})\times (r^{*}_{2}-r_{2}) \) intersects
  \( E_{\x} \) and \( E_{\y} \).
  Then \( \bB(\x,\r^{*}) \) contains both \( \x \) and \( \y \), hence by sparsity \( \y\in \bB(\x,\r^{*}) \)
and \( E_{\x}=E_{\y} \).
\end{proof}

\begin{definition}\label{def:sparsity}
Let \( \gamma> 1 \), \( \beta \ge 3\gamma \) be parameters, and let
  \begin{align*}
  1 &=\B_{1}<\B_{2}<\dotsm,\quad
  1=\S_{1}<\S_{2}<\dotsm,
\\ &\S_{k+1}/\S_{k},\; \B_{k+1}/\B_{k}\ge 2\beta
\end{align*}
be sequences of integers to be fixed later.
For a space-time set \( E\subseteq\bbZ\times\bbZ_{\ge 0} \), let \( E^{(1)} = E \).
For \( k>1 \) let \( E^{(k+1)} \) be obtained by deleting from \( E^{(k)} \) the
\( \pair{\beta\pair{\B_{k}}{\S_{k}}}{\gamma\pair{\B_{k+1}}{\S_{k+1}}} \)-isolated points.
Set \( E \) is \( k \)-\df{sparse} if \( E^{(k+1)} \) is empty.
It is simply \df{sparse} if \( \bigcap_{k}E^{(k)}=\emptyset \).
When \( E=E^{(k)} \) and \( k \) is known
then we will denote \( E^{(k+1)} \) simply by \( E^{*} \).
\end{definition}

\begin{definition}[Burst]\label{def:burst}
  Suppose that the set \( E \) is \( \pair{\r}{\r^{*}} \)-sparse with \( \r^{*}>2\r \).
  By the above lemma, it is partitioned into subsets of the form \( E_{\x} \).
  In what follows we will call these sets \df{bursts}.
  The definition depends on the parameter \( \r \) which will always be
  clear from the context.
  Typically, it will be \( \beta(\B,\S) \) where \( \beta \) is a constant,
  \( \B=\B_{k} \) and \( \S=\S_{k} \) for some \( k \) called the \df{level}.
\end{definition}

The following lemma connects the above defined sparsity notions to the requirement
of small fault probability.

\begin{lemma}[Sparsity]\label{lem:sparsity}
Let \( \Q_{k} = \B_{k+1}/\B_{k} \),  \( \V_{k} = \S_{k+1}/\S_{k} \), and
\begin{align}\label{eq:growth-assumption}
  \lim_{k\rightarrow\infty}\frac{\log \Q_{k}\V_{k}}{1.5^k}=0.
\end{align}
For sufficiently small \( \eps \), for every \( k\ge 1 \) the following holds.
Let \( E\subseteq \bbZ\times\bbZ_{\ge 0} \) 
be a random set that is \( \eps \)-bounded as in Definition~\ref{def:faults-eps-bounded}.
Then for each point \( \x \)  and each \( k \),
 \begin{align*}
   \Pbof{\bB(\x,\pair{\B_{k}}{\S_{k}})\cap E^{(k)}\neq\emptyset} < \eps \cdot 2^{-1.5^{k-1}}.
 \end{align*}
As a consequence, the set \( E \) is sparse with probability 1.
\end{lemma}

This lemma allows a doubly exponentially increasing sequence \( \U_{k} \), resulting
in relatively few simulation levels as a function of the computation time

\subsection{Difficulties}\label{sec:novelties}

We list here some of the main problems that the paper deals with, 
and some general ways in which they will be solved or avoided.
Some more specific problems will be pointed out later, along with their solution.

\begin{description}

\item[Non-aligned colonies] A large burst of faults in \( M_{1} \) can modify the order of
entire colonies or create new ones with gaps between them.
To deal with this problem, machines \( M_{k} \) for \( k>1 \)
will be \df{generalized Turing machines}, allowing for non-adjacent cells.

\item[Clean areas]
  On the tape of a generalized Turing machine, based on its content, some areas will 
  be called \df{clean}, the rest \df{disordered}.
  In clean areas, the analysis can count on an existing underlying simulation,
  and therefore the transition function is applicable.
  Noise can disorder the areas where it occurs.

\item[Extending cleanness]
  The predictability of the machine is decreased when the head enters into disorder.
  But the model still provides some ``magical'' properties
  helping to restore cleanness (in the absence of new noise):
  \begin{Alphenum}
  \item escaping from any area in a bounded amount of time;
 \item shrinking disorder, as the head passes in and out of it;
 \item\label{i:many-slides}
   cleaning an interval when passed over a certain number of times.
 \end{Alphenum}
While an area is cleaned, it will also be re-populated with cells.
Their content is not important, what matters is the restoration of predictability.

\item[Rebuilding]
The need to reproduce these cleaning properties in simulation is the
main burden of the construction.
The part of the program devoted to this is the rebuilding procedure,
invoked when local repair fails.
  It reorganizes a part of the tape having the size of a few colonies.
\end{description}

\subsection{Generalized Turing machines}\label{sec:TM}

This section, together with Section~\ref{sec:traj},
introduces the key concepts used in the proof.
Let us recall that a one-tape Turing machine is defined by a
finite set \( \Gamma \) of \df{internal states},
a finite alphabet \( \Sigma \) of \df{tape symbols}, a transition function \( \delta \),
and possibly some distinguished states and tape symbols.
At any time, the head is at some integer position \( \h \), and is observing the tape
symbol \( A(\h) \).
The meaning of \( \delta(a,q)=(a',q',d) \) is that if \( A(\h)=a \) and the state is \( q \) then
the \( A(\h) \) will be rewritten as \( a' \) and \( \h \) will change to \( \h+d \).

We will use a model that is slightly different, but has clearly the same expressing power.
(Its advantage is that it is a little more convenient to describe its simulations.)
There are no internal states, but the head observes and modifies a \emph{pair} of
neighboring tape cells at a time; in fact, we imagine it to be positioned between these
two cells called the \df{current cell-pair}.
The \df{current cell} is the left element of this pair.
Thus, a Turing machine is defined as \(    (\Sigma,\tau) \) where
the tape alphabet \( \Sigma \) contains at least the distinguished
symbols \( \blank,0,1 \) where \( \blank \) is called the \df{blank symbol}.
The \df{transition function} is
\(  \tau\colon\Sigma^{2}\to \Sigma^{2}\times\{-1,1\} \).
A \df{configuration} is a pair \( \xi = (A,\h) = (\xi.\tape,\xi.\pos) \)
where  \( \h\in\bbZ \) is the \df{current} (or \df{observed)}
head position, (between cells \( h \) and \( h+1 \)),
and \( A\in\Sigma^{\bbZ} \) is the \df{tape content}, or \df{tape configuration}:
in cell \( p \), the tape contains the symbol \( A(p) \).
Though the tape alphabet may contain
non-binary symbols, we will restrict input and output to binary.
The tape is blank at all but finitely many positions.

As the head observes the pair of tape cells
with content \( \va=(a_{0},a_{1}) \) at positions \( \h \), \( \h+1 \) denote \(  (\va',d)=\tau(\va)  \).
The transition \( \tau \) will
change the tape content at positions \( \h \), \( \h+1 \) to \( a'_{0} \), \( a'_{1} \),
and move the head to tape position to \( \h+d \).
A \df{fault} occurs at time \( t \) if the output \( (\va',d) \)
of the    transition function at this time is replaced with some other value
(which then defines the next configuration).


\begin{remark}\label{rem:internal-state}
  The informal description of the simulation program below is in some places
  written as if there was such a thing as an internal state.
  But this can clearly be implemented for example as follows.
  In the current cell-pair, one element will always be marked as the one carrying an ``internal state'',
  and a field of this cell can be used to represent this state.
  Anticipating the move of the head to left or right,
  the transition may move this mark (along with the changed state), if needed,
  to the other element of the cell-pair.
\end{remark}

The machines that occur in simulation will be a generalized version of the above model,
allowing non-adjacent cells and areas called ``disordered''
in which the transition function is non-applicable.
As a convenience feature, two integer parameters are added:
the cell body size \( \B\ge 1 \) and and an upper bound \( \Tu\ge 1 \) on the transition time.
These allow placing all the different Turing
machines in a hierarchy of simulations onto the same space line and the same time line.

\begin{definition}[Generalized Turing machine]\label{def:gen-TM}
  A \df{generalized Turing machine} \( M \) is defined by a tuple
  \begin{align}\label{eq:gen-TM}
    (\Sigma, \tau, \Vacant, \New, \Bad,\B, \Tu, \passno, \escno),
  \end{align}
  where \( \Sigma \) is the \df{alphabet}, and
\begin{align*}
  \tau: \Sigma^{2}\times\{\True,False\}\to \Sigma^{2}\times\{-1,1\}  .
\end{align*}
is the \df{transition function}.
In  \( \tau(a,b,\alpha) \) the argument \( \alpha \) is \( \True \) if the pair of observed cells is
adjacent (no gap between them), and \( \False \) otherwise.
Among the elements of the tape alphabet we distinguish the input-output element \( 0,1 \),
a special symbols \( \Vacant \), \( \Bad \) and a subset \( \New \).
\begin{itemize}
\item \( \Vacant \) plays the role of a blank symbol (the absence of a cell).
 \item The symbol \( \Bad \) marks \df{disordered} areas of the tape.

   \item The state of newly created cells is in \( \New \).

    \item The parameters \( \B,\Tu \) were discussed above.
The integer \( \passno \) will play the role of the number of passes needed to clean an area (see below).
The positive real \( \escno \) will help upper-bound the escape time from a disordered area.
The parameter \( \S_{k} \) used in structuring the noise will be specified in
Definition~\ref{def:hier-params}.

\end{itemize}
The transition function \( \tau \) has no inputs or outputs that are \( \Bad \) or \( \Vacant \).

\end{definition}

The effect of the transition function on configurations will be explained in
Definition~\ref{def:dictated}.

\begin{remark}\label{rem:New}
Let \( \tau(a,b,\alpha) = (a',b',d) \).
 We will have \( a'\in\New \) or \( b'\in\New \) only if \( a,b\not\in\New \) and \( \alpha=\False \),
that is the observed cells are not adjacent.
Even then we can have \( a'\in\New \) only if \( d=-1 \) and \( a'\in\New \) only if \( d=1 \).   
\end{remark}

A formal definition of a configuration of a generalized Turing machine is given in
Section~\ref{sec:gen-TM}, though it is essentially defined by the tape content \( A \) and the head
position.
A point \( p \) is \df{clean} if  \( A(p)\ne\Bad \).
A set of points is \df{clean} if it consists of clean points.
We say that there is a \df{cell} at a position \( p\in\bbZ \) if the interval
\( p+\lint{0}{\B} \) is clean, \( A(p)\ne \Vacant \) and all other elements of this interval are vacant.
In this case, we call the interval \( p+\lint{0}{\B} \) the \df{body} of this cell.
Thus, cell bodies must not intersect.
If their bodies are at a distance \( <\B \) from each
other, with a clean interval containing both, then they are called \df{neighbors}.
They are called \df{adjacent} if this distance is \( 0 \).

A sequence of configurations conceivable as a computation will be called a ``history''.
For standard Turing machines, 
the histories that obey the transition function could be called ``trajectories''.
For generalized Turing machines the definition of trajectories is more complex; it
allows some limited violations of the transition function, while providing the mechanisms
for eliminating disorder.
Let \(    \Noise\subseteq \bbZ\times\bbZ_{\ge 0} \)
denote the set of space-time points at which faults occur.
Section~\ref{sec:traj} below will define a certain subset of possible histories
called \df{trajectories}.
In order to motivate their choice, we first introduce the notion of simulation.

 \subsection{Simulation}\label{sec:sim}

The notion of simulation used in the proof and introduced here,
relies on a certain concept of trajectories.
On the other hand, the simulation concept helps motivate Section~\ref{sec:traj} where
trajectories will be defined.

Until this moment, we used the term ``simulation'' informally, to denote
a correspondence between configurations of
two machines which remains preserved during the computation.
In the formal definition, this correspondence will essentially be a code
\( \varphi=(\varphi_{*},\varphi^{*}) \).
The \emph{decoding} part of the code is the more important one.
We want to say that machine \( M_{1} \) simulates machine \( M_{2} \) via
simulation \( \varphi \) if whenever \( (\eta, \Noise) \) is a trajectory of \( M_{1} \) 
then \( (\eta^{*},\Noise^{*}) \),
defined by \( \eta^{*}(\cdot,t)=\varphi^{*}(\eta(\cdot,t)) \), is a
trajectory of \( M_{1} \).
Here, \( \Noise^{*} \) is computed by the residue operation (deleting isolated elements)
as in Definition~\ref{def:sparsity}.
We will make, however, two refinements.
First, we require the above condition only for
those \( \eta \) for which the initial configuration
 \( \eta(\cdot,0) \) has been obtained by encoding, that is it has the form 
\( \eta(\cdot,0)=\varphi_{*}(\xi) \).
Second, to avoid the transitional ambiguities in a history,
we define the simulation decoding as a mapping \( \Phi^{*} \)
between \emph{histories}, not just configurations:
\( \Phi^{*}(\eta,\Noise)=(\eta^{*},\Noise^{*}) \).

\begin{sloppypar}
\begin{definition}[Simulation] \label{def:simulation-central}
  Let \( M_{1},M_{2} \) be two generalized Turing machines, and let
  \( 
    \varphi_{*}:\Configs_{M_{2}} \to \Configs_{M_{1}} \)
be a mapping from configurations of \( M_{2} \) to those of \( M_{1} \), such that it maps
starting configurations into starting configurations.
Let \(    \Phi^{*}:\Histories_{M_{1}} \to \Histories_{M_{2}} \) be a mapping.
The pair \( (\varphi_{*}, \Phi^{*})  \)
is called a \df{simulation} (of \(  M_{2}  \) by \(  M_{1}  \)) if for every
trajectory \(  (\eta, \Noise)  \) of \( M_{1} \) with initial
configuration \(  \eta(\cdot,0)=\varphi_{*}(\xi)  \),
the history \(  (\eta^{*},\Noise^{*})=\Phi^{*}(\eta,\Noise)  \) is
a trajectory of machine \(  M_{2}  \).
\end{definition}
  \end{sloppypar}

In the noise-free case it is easy to find examples of simulations.
However, in the cases with noise, finding any nontrivial example 
is a challenge, and depends on a careful definition of trajectories for generalized Turing machines.

\subsection{Trajectories}\label{sec:traj}

This section completes the definition of the central concept of the
proof---modulo the natural definitions spelled out in Section~\ref{sec:gen-TM}.
A history of a generalized Turing machine \( M \) is a trajectory
if it obeys certain constraints on its fault-free parts.
We discuss these properties first informally.

\begin{description}
\item[Transition Function] This property says---in more precise terms---that
  in a clean area, the transition function is obeyed.

\item[Spill Bound] limits the extent to which a disordered interval can spread.

\item[Escape] limits the time for which the head can be trapped in a small area.

\item[Attack Cleaning] erodes disorder as the head repeatedly enters and leaves it.

\item[Pass Cleaning] cleans the interior of an
  interval if the head passes over it enough times.
  
\end{description}

The definition below depends on the notions of current cell-pair, switch and dwell period given in
 in Section~\ref{sec:gen-TM}, but should be understandable as it is.

\begin{definition}[Transition]\label{def:dictated}
Suppose that at times \( t' \) before a switching time \( t \) but after 
any previous switch, the current cell-pair \( (x,y) \) has state \( \va = (a,b)\).
Let \( (a',b',d) =\tau(a,b,\alpha) \), where \( \alpha=\True \) if the cell-pair is adjacent and \( \False \)
otherwise.
Let \( u,v \) be the states of the cells \( x,y \) after the transition, and
let \( x',y' \) be the new current cell pair.
We say that the switch is \df{dictated by the transition function} if the following holds.
We state the conditions for \( d=1 \), the case \( d=-1 \) is analogous.
\begin{itemize}
\item \( u=a' \).
  
  \item Suppose \( b'\not\in\New \); then \( v=b' \), \( x'=y \).
    If cell \( y \) has a neighbor \( z \) on the right then \( y'=z \).
    Else a new adjacent neighbor \( z \) is created on the right of \( y \)
    with a state in \( \New \), and again \( y'=z \).

  \item Suppose \( b'\in\New \) (in which case \( \alpha=\False \) and \( x,y \) are not adjacent,
    see Remark~\ref{rem:New}).
    Then \( v=\Vacant \) (cell \( y \) is erased), \( x'=x \), and a cell \( y' \) adjacent to \( x \) on the right
    is created with state \( b' \).
    We will say that cell \( y \) is \df{replaced} with the new cell \( y' \), and call this a \df{replacement
      situation}.
\end{itemize}
\end{definition}

As a consequence of this definition, new cells are created
automatically when the head would step onto a vacant area,
and whenever a cell is ``killed'' another one is created automatically in a place overlapping with its body.

We will use the following constants:  
\begin{align}\label{eq:cns.traj}
  \CRebuild = 8,\;
  \CSpill = \CMarg = \CRebuild,\;
\end{align}
and we will assume
\begin{align}\label{eq:beta-lb}
  \gamma> 4(\CMarg + \CSpill),  
\end{align}
where \( \gamma \) was used in Definition~\ref{def:sparsity} (sparsity).
(Even though we set all these constants to the same value, it helps clarity 
to give them separate names.)


 \begin{definition}\label{def:interior}
  For a set \( K \) on the line,  and some real \( c \) let us define its \( c \)-\df{interior}
  \(  \Int(K,c) \) as the set of those points of \( K \) that are at a distance \( \ge c \) from
  its complement.
  For an interval \( K = \lint{a}{b} \), this is \( \lint{a+c}{b-c} \).
  In this case, will use it also  with negative \( c \); then ``interior'' is really an extended
  neighborhood of \( I \).
 \end{definition}

  In the following definition, it is important to keep in mind the difference between noise and disorder.
  A noise-free space-time rectangle can very well contain disordered areas on the tape.
  
  \begin{definition}[Trajectory]\label{def:traj}
\begin{sloppypar}
   A history  \( (\eta, \Noise) \) of a generalized Turing 
machine~\eqref{eq:gen-TM} with \(\eta(t) =\)
\( (A(t), \h(t), \vhc(t)) \)
is called a \df{trajectory} of \( M \) if the following conditions hold, in any 
noise-free space-time interval \( I\times J \).
  \end{sloppypar}
\begin{description}

\item[Transition Function]\label{i:def.traj.transition}
Consider a switch, where the current cell-pair \( \vhc \)
is inside a clean area, by a distance of at least \( 2.5\B \). 
Then the new state of the current cell-pair and
the direction towards the new current head position
are dictated by the transition function.
The only change on the tape occurs on the interval enclosing the new and old current cells.
Further, the length of the dwell period before the switch is bounded by \( \Tu \).

\item[Spill Bound]\label{i:spill-bound}
  A clean interval can shrink by at most \( \CSpill \B \).

\item[Escape] \label{i:def.traj.escape}
  The head will leave any interval of size \( \le \gamma\B \) 
  within time \( \escno\Tu \).

\begin{sloppypar}
\item[Attack Cleaning] \label{i:def.traj.attack-cleaning}
  Suppose that the current cell-pair \( (x,x') \)
  is at the right end of a clean interval \( \lint{a}{b} \) of size
\( \ge (\CSpill+2)\B \), with the head at position \( x \).
Suppose further that the transition function directs the head right,
and not in a replacement situation of Definition~\ref{def:dictated}.
Then by the time the head comes back to \( x-(\CSpill+1)\B \), the clean area is extended
to the right by at least \( \B \).
Similarly when ``left'' and ``right'' are interchanged.
 \end{sloppypar}

\item[Pass Cleaning]
  Suppose that the head makes at least \( \passno \) (left-right, right-left) pairs
  of passes over an interval \( I \).
  Then at some time during this, the interior \( \Int(I, \CMarg\B) \) of \( I \) becomes clean.

\end{description}
\end{definition}

Recall that we will have a hierarchy of simulations \( M_{1}\to M_{2}\to \dotsm \) where
machine \( M_{k} \) simulates machine \( M_{k+1} \).
Our construction will set \( \passno=8k + O(1) \) for \( M_{k} \).
This can be interpreted as saying that each 8 passes raise the ``organization level'':
\( 8 k \) passes achieve cleanness on the level of \( M_{k} \).

\subsection{Scale-up}

Above, we have set up the conceptual structure of the construction and the proof.
Here are some of the parameters: 

\begin{definition}\label{def:hier-params}
  Let
  \begin{align*}
  \Q_{k}&=c_{\Q}\cdot 2^{1.2^{k}},
    \\   \passno_{k}&=5 k + c_{\passno}, 
    \\   \escno_{k}&=\CEsc\Q_{k-1}\passno_{k-1},
    \\   \U_{k} &= c_{\U}\Q_{k}\passno_{k}^{9},
 \end{align*}
for appropriate constants \( c_{i}>0 \).
 These sequences clearly satisfy~\eqref{eq:growth-assumption}, and
 define \( \B_{k}=\B_{1}\prod_{i<k}\Q_{i} \), 
 \( \Tu_{k}=\Tu_{1}\prod_{i<k}\U_{i} \), \( \S_{k} = \Tu_{k}\escno_{k} \),
 so
 \begin{align*}
   \V_{k} = \S_{k+1}/\S_{k} = \U_{k}\escno_{k}/\escno_{k-1},
 \end{align*}
 satisfying~\eqref{eq:growth-assumption}.
 When we write for example \( \Q=\Q_{k} \) then we can write \( \Q^{*}=\Q_{k+1} \).
\end{definition}

 The main remaining part is the definition of the simulation program and the decoding \( \Phi^{*} \),
and the proof that with this
program, the properties of a trajectory \( (\eta,\Noise) \) of machine \( M=M_{k} \) imply
that the history \( \Phi^{*}(\eta,\Noise)=(\eta^{*},\Noise^{*}) \) obeys the same trajectory
requirements on the next level.
The program is described in Sections~\ref{sec:sim-struc}-\ref{sec:healing}.
The most combinatorially complex part of the proof of trajectory properties
is in Section~\ref{sec:cleaning}, bounding and eliminating disorder on the next level.
Section~\ref{sec:computation} wraps up the proof of the main theorem.


\section{Some formal details}

We give here some details that were postponed from the overview section.

\subsection{Examples}\label{sec:examples}

Examples given in Section~\ref{sec:appendix}
motivate various complexities of the construction and proof.
Some of them may not be completely understandable without the details of the following program:
refer to them when wondering about the necessity for some feature.
In all examples where a ``burst'' is mentioned, it is understood in the sense of
Definition~\ref{def:burst}, as a space-time set of faults covered by a rectangle of a
certain (``small'') size.
A burst of ``level'' \( k \) is also referred to as a burst (of faults) of machine \( M_{k} \).
If \( k \) is fixed and we just talk about a machine \( M = M_{k} \) and its
history \( \eta \) then we will also refer to a burst of \( \eta \).


A good-sized neighborhood of the head will contain enough information
to prevent a burst from pushing the computation from one phase to another, wrong one.
There will also be some reasons for which even a non-constant size neighborhood will need
to be checked repeatedly.
For this, the head will proceed in zigzags: every
step advancing the head in the simulation is followed by some \( \Z \) steps 
of going backward and forward again (with parameter \( \Z \) chosen appropriately below),
checking consistency (and starting a healing process if necessary).
This will also enable the head to progress into a large
disordered area, without being easily fooled into going away.



Example~\ref{xpl:need-feather} raises a problem.
The device by which we will mitigate the effect of this kind of capturing is another property of the
movement of the head which will call \df{feathering}:
if the head turns back from a tape cell then next time it must go beyond.
This requires a number of adjustments to the program (see later).

Examples~\ref{xpl:two-slides} and~\ref{xpl:unbounded} show that disorder may not be
eliminated by a bounded number of noise-free slides over it.
Our construction will ensure that, on the other hand,
\( O(k) \) passes (free of \( k \)-level noise) will restore organization to level \( k \).
This property of the construction will be incorporated into our definition of a generalized
Turing machines as the ``magical'' property~\eqref{i:many-slides} above.

\subsection{Codes}\label{sec:codes}

The input of our computation will be encoded by some error-correcting code,
to defend against the possibility of losing information even at the first reading.

\begin{definition}[Codes]\label{def:codes}
    Let \( \Sigma_{1},\Sigma_{2} \) be two finite alphabets.
    A \df{block code} is given by a positive integer \( \Q \)---called
    the \df{block size}---and a pair of functions
    \begin{align*}
            \psi_{*} :\Sigma_{2}\to\Sigma_{1}^{\Q},
            \quad
            \psi^{*}:\Sigma_{1}^{\Q}\to\Sigma_{2}
    \end{align*}
    with the property \( \psi^{*}(\psi_{*}(x))=x \).
    Here \( \psi_{*} \) is the encoding function (possibly introducing redundancy)
    and \( \psi^{*} \) is the decoding function (possibly correcting errors).
    The code is extended to (finite or infinite) strings by encoding each letter individually:
\begin{align*}
 \psi_{*}(x_{1},\dots,x_{n})=\psi_{*}(x_{1})\dotsm\psi_{*}(x_{n}) .
\end{align*}

\end{definition}

\subsection{Proof of the sparsity lemma}

\begin{proof}[Proof of Lemma~\ref{lem:sparsity}]
  The proof uses slightly more notation than it would if we simply assumed independence
  of faults at different space-time sites, but it is essentially the same.

  Let \( \cE_{k}(\x) \) be the event  \( \bB(\x,\pair{\B_{k}}{\S_{k}})\cap E^{(k)}\neq\emptyset \).
  Let \( \cM=\cM_{k}(\x) \) be the set of minimal sets \( A\subseteq \bbZ\times\bbZ_{+} \)
with  \( A\subseteq E\imp \cE_{k}(\x) \).

  \begin{claim}\label{claim:minimal-set}
    Each set in \( \cM_{k}(\x) \) is contained in \( \bB(\x, 1.5\gamma(\B_{k},\S_{k})) \).
  \end{claim}
  \begin{proof}
    The statement is clearly true for \( k=1 \).
    Suppose it is true for \( k \), let us prove it for \( k+1 \).
    The event \( \cE_{k+1}(\x) \) holds if and only if for
    some \( \x'\in\bB(\x,\pair{\B_{k+1}}{\S_{k+1}})\cap\E \)
    there is some point \( \y \)  in 
\begin{align*}
  E^{(k)}\cap \bB(\x',\gamma(\B_{k+},\S_{k+1}))\setminus \bB(\x',\beta(\B_{k},\S_{k})) .
\end{align*}
    Then there is some minimal set \( A' \) 
    with the property \( A'\subseteq E\imp \cE_{k}(\y) \).
    By the inductive assumption these sets are contained in
    \(  \bB(\y,2\gamma(\B_{k},\S_{k})) \).
    All the minimal sets \( A \) with \( A\subseteq E\imp \cE_{k+1}(\x) \)
have the form \( A'\cup\{\x'\} \) with some such \( \x' \) and \( A' \).
    Also 
\begin{align*}
     A &\subseteq \bB(\x',(\gamma\B_{k+1}+2\gamma\B_{k},\gamma\S_{k+1}+2\gamma\S_{k})) 
\\     &\subseteq \bB(\x,(\gamma\B_{k+1}+(2\gamma+1)\B_{k},\gamma\S_{k+1}+(2\gamma+1)\S_{k})) 
\subseteq \bB(\x, 1.5\gamma(\B_{k+1},\S_{k+1})),
\end{align*}
assuming \( \S_{k+1}/\S_{k}> 6 \), \( \B_{k+1}/\B_{k}> 6 \).
\end{proof}

Let \(  f_{k}(\x)= \sum_{A\in\cM}\eps^{|A|} \).
By the union bound we have \( \Prob(\cE_{k}(x))\le f_{k}(\x) \).

Let \( p_{k}= \eps\cdot 2^{-1.5^{k-1}} \).
We will prove \(   f_{k}(\x) < p_{k} \) by induction.
For \( k=1 \), rectangles \( \bB(\x_{i},\pair{\B_{1}}{\S_{1}}) \) have size \( 1 \), so
by the \( \eps \)-boundedness, \( f_{1}(\x)<\eps \).
 Assume that the statement holds for \( k \), we will prove it for \( k+1 \).

 Suppose \( \y \in E^{(k)}\cap \bB(\x,\pair{\B_{k+1}}{\S_{k+1}}) \).
According to the definition of \( E^{(k)} \),  there is a point
\begin{align}\label{eq:sparse-as}
 \vek{z} \in
 \bB(\y,\gamma\pair{\B_{k+1}}{\S_{k+1}})\cap E^{(k)}\setminus \bB(\y,\beta\pair{\B_{k}}{\S_{k}}).
 \end{align}
Consider a standard partition of space-time into rectangles \(  K_{p}=\bB(\vek{c}_{p},(\B_{k},\S_{k})) \).
Let 
\begin{align*}
       I&=\setOf{p}{K_{p}\cap \bB(\x,\gamma\pair{\B_{k+1}}{\S_{k+1}})\ne\emptyset}.
 \end{align*}
 We are only interested in rectangles \( K_{p} \) with \( p\in I \).
 Let
\begin{align*}
  K'_{p}=\bB(\vek{c}_{p},\pair{2\gamma\B_{k}}{1.5\gamma\S_{k}}) .
\end{align*}
If \( K_{i},K_{j} \) are the rectangles in this partition containing \( \y \) and \( \vek{z} \), then
\( K'_{i}\cap K'_{j}=\emptyset \).
This follows from the fact that \( \abs{y_{1} - z_{1}}>\beta\B_{k} \),
\( \abs{y_{2} - z_{2}}>\beta\S_{k} \), and \( \beta\ge 3\gamma \) 
in Definition~\ref{def:sparsity}.
The event \( \y\in E^{(k)} \) can be written as
\( \bigcup_{A\in\cM_{k}(\y)}\{A\subseteq E\} \), and
by Claim~\ref{claim:minimal-set} we have \( A\subseteq \bB(\y,1.5\gamma(B_{k},\S_{k}))\subseteq K'_{i} \)
for each \( A\in\cM_{k}(\y) \).
Similarly for \( \vek{z} \) and \( K'_{j} \).
Let \( \cM(i)=\bigcup_{\y\in K_{i}}\cM_{k}(\y) \), then each set \( A\in\cM(i) \) is in \( K'_{i} \).
The disjointness of \( K'_{i} \) and \( K'_{j} \) and the inductive assumption implies
\begin{align}\nonumber
  f_{k+1}(\x) &\le \sum_{i,j\in I, K'_{i}\cap K'_{j}=\emptyset}\sum_{A\in\cM(i),A'\in\cM(j)}\eps^{|A|+|A'|}
   =\sum_{i,j\in I, K'_{i}\cap K'_{j}=\emptyset}f_{k}(\vek{c}_{i})f_{k}(\vek{c}_{j})
  \\\label{eq:sparsity.last-line}&\le |I|^{2}p_{k}^{2} = |I|^{2}\eps^{2} 2^{-1.5^{k}}\cdot 2^{-0.5\cdot 1.5^{k-1}}
       = p_{k+1}\eps |I|^{2}2^{-0.5\cdot 1.5^{k-1}}.
\end{align}
We have  \( |I|\le (2\gamma \Q_{k}+1)(2\gamma\V_{k}+1) \).
Since \( \lim_{k}\frac{\log\V_{k}\Q_{k}}{1.5^k}=0 \), 
the multiplier of \( p_{k+1} \) in~\eqref{eq:sparsity.last-line} is \( \le 1 \) for sufficiently small  \( \eps \).
\end{proof}

\subsection{Configuration, history}\label{sec:gen-TM}

A configuration, as defined below, contains a pair of positions
\( \vhc = (\hc_{0},\hc_{1}) \) called the \df{current cell-pair}:
In difference to the Turing machines of Section~\ref{sec:TM},
the position of the head may not be exactly between the current cells: this allows the model
to fit into the framework where 
a generalized Turing machine \( M^{*} \) is simulated by some (possibly
generalized) Turing machine \( M \).
The head \( h \) of \( M \)---made equal to that of \( M^{*} \)---may
oscillate inside and around the current cell-pair of \( M^{*} \).

\begin{definition}[Configuration]\label{def:config}
    A \df{configuration} \( \xi \) of a generalized Turing machine~\eqref{eq:gen-TM} is a tuple
    \begin{align*}
      (A,\h,\vhc) = (\xi.\tape,\xi.\pos, (\xi.\curcell_{0},\xi.\curcell_{1}))
    \end{align*}
    where \( A:\bbZ\to\Sigma \) is the tape, \( \h \in\bbZ \) is the head position, \( \vhc\in\bbZ^{2} \)
    is the current cell-pair.
    We have \( A(p)=\Vacant \) in all but finitely many positions \( p \).
Whenever the interval \( \h+\lint{-4\B}{4\B} \) is clean the current cell-pair
must be within it. 
Let
    \begin{align*}
         \Configs_{M}
    \end{align*}
    denote the set of all possible configurations of a Turing machine \( M \).
\end{definition}

The above definitions can be localized to define a configuration 
over a space interval \( I \) containing the head.

\begin{definition}[History]
For a generalized Turing machine~\eqref{eq:gen-TM}, consider
a sequence \( \eta = \) \( ( \eta(0,\cdot),\) \( \eta(1,\cdot),\) \( \dots) \),
of configurations along with a noise set \( \Noise \).
Let \( \h(t)= \eta(t,\cdot).\pos \) be the head position.

A \df{switching time} is a noise-free time when any part of \( \eta \) other than \( \h(t) \)
changes (\( \h(t) \) is also allowed to change at non-switching times).
A \df{dwell period} is the interval between any consecutive pair of
switching times with the property that the
space-time rectangle containing them and the head is clean and noiseless.

The pair \(  (\eta,\Noise) \) will be called a \df{history}
of machine \( M \) if the following conditions hold.
\begin{itemize}
\item \( \abs{\h(t) - \h(t')} \le \abs{t' - t} \).
  
\item In two consecutive configurations, the tape content \( A(p,t) \) of the positions \( p \)
  not in \( \h(t) + \lint{-2\B}{2\B} \) remains the same.
\item At each noise-free switching time the head is on the new current cell-pair:
  \( \hc_{0}(t)=\h(t) \).
(In particular, when at a switching time a current cell becomes
\( \Vacant \), the head must already be elsewhere.)

\item The length of dwell periods is at most \( \Tu \).

\end{itemize}
The above definition can be localized to define a history \( I\times J \) containing the head.
Let
\begin{align*}
  \Histories_{M}
\end{align*}
denote the set of all possible histories of \( M \).
\end{definition}

\subsection{Hierarchical codes}\label{sec:hier-codes}

Recall the notion of a code in Definition~\ref{def:codes} and of configuration
in Definition~\ref{def:config}.

\begin{definition}[Code on configurations]\label{def:configuration-code}
\begin{sloppypar}
 Consider two generalized Turing machines \( M_{1},M_{2} \) with the corresponding
alphabets and transition functions, where \( \B_{2}/\B_{1} \) is an integer denoted \( \Q=\Q_{1} \).
Assume that a block code
\(
   \psi_{*}:\Sigma_{2}\to\Sigma_{1}^{\Q}
\)
is given, with an appropriate decoding function, \( \psi^{*} \).
Symbol \( a\in\Sigma_{2} \) is interpreted as the content of some tape square.
This block code gives rise to a \df{code on configurations}, that is a pair of functions
    \begin{align*}
        \varphi_{*} :\Configs_{M_{2}} \to \Configs_{M_{1}},
        \quad
        \varphi^{*}:\Configs_{M_{1}} \to \Configs_{M_{2}}
    \end{align*}
    that encodes some (initial) configurations \( \xi \) of \( M_{2} \) into configurations of \( M_{1} \):
    each cell of \( M_{2} \) is encoded into a colony of \( M_{1} \) occupying the same interval.
    Formally, assuming \( \xi.\curcell_{j}=\xi.\pos+(j-1)\B_{2} \), \( j=0,1 \) 
    we set \( \varphi_{*}(\xi).\pos = \xi.\pos \),
    \( \varphi_{*}(\xi).\curcell_{j}= \varphi_{*}(\xi).\pos+(j-1)\B_{1} \),
  and for all \( i\in\bbZ \),
\begin{align*}
 \varphi_{*}(\xi).\tape(i\B_{2},i\B_{2}+\B_{1}, \dots, (i+1)\B_{2} - \B_{1}) = \psi_{*}(\xi.\tape(i)).
 \end{align*}
\end{sloppypar}
 \end{definition}

\begin{definition}[Hierarchical code]\label{def:hierarchical-code}
For \( k\ge 1 \), let \( \Sigma_{k} \) be an alphabet, of a generalized Turing machine \( M_{k} \).
Let \( \Q_{k}=\B_{k+1}/\B_{k} \) be an integer (viewed as colony size), let \( \varphi_{k} \)
be a code on configurations defined by a block code
  \begin{align*}
       \psi_{k}: \Sigma_{k+1}\to \Sigma_{k}^{\Q_{k}}
  \end{align*}
as in Definition~\ref{def:configuration-code}.
The sequence \( (\Sigma_{k},\varphi_{k}) \), (\( k\ge 1) \),  is
called a \df{hierarchical code}.
For this hierarchical code, configuration \( \xi^{1} \) of \( M_{1} \)
is called a \df{hierarchical code configuration} of height \( k \) if a sequence
of configurations \( \xi^{2},\xi^{3},\dots,\xi^{k} \) of \( M_{2},M_{3},\dots,M^{k} \) exists with
\begin{align*}
 \xi^{i}=\varphi_{*i}(\xi^{i+1})
 \end{align*} 
for all \( i \).
If we are also given a sequence of mappings \( \Phi^{*}_{1} \), \( \Phi^{*}_{2} \), \( \dots \) 
such that for each \( i \), the pair \( (\varphi_{i*},\Phi_{i}^{*}) \),
is a simulation of \( M_{i+1} \) by \( M_{i} \) 
then we have a \df{hierarchy of simulations} of height \( k \).
\end{definition}

We will construct a hierarchy of simulations whose height grows during the
computation---by a mechanism to be described later.

\section{Simulation structure}\label{sec:sim-struc}

In what follows we will describe the program of the reliable Turing machine:
a hierarchical simulation in which simultaneously each \( M_{k+1} \) is simulated
by \( M_{k} \), with an added mechanism to
raise the height of the hierarchy when needed.
Most of the time, we will write \( M=M_{k} \),  \( M^{*}=M_{k+1} \).
Ideally, cells will be grouped into colonies of size \( \Q=\B^{*}/\B \).
Simulating one step of \( M^{*} \) takes a sequence of steps of \( M \)
constituting a \df{work period}.
Machine \( M \) will perform the simulation as long as the noise
in which it operates is \( (\beta\pair{\B}{\Tu}, \gamma\pair{\B^{*}}{\Tus}) \)-sparse
(as in Definition~\ref{def:isolation}).
This means, by Lemma~\ref{lem:bursts}, that a burst affects at most \( \beta \) consecutive tape cells,  
and there is at most one burst in any \( \gamma \)  neighboring work periods.
A design goal for the program is to
``correct'' a burst within a space much smaller than a colony.

To see whether consistency, that is the basic tape pattern
supporting simulation, is broken somewhere, a very local precaution will be taken in each step:
each step will check whether the current cell-pair is allowed in a healthy configuration.
If not then a \df{healing} procedure will be called;
we will also say that \df{alarm} will be called.
On the other hand, the \df{rebuilding} procedure will be called on some indications that
healing fails.

\subsection{Error-correcting code}\label{sec:coding}

Let us add error-correcting features to the block codes introduced in
Definition~\ref{def:codes}.

\begin{sloppypar}
\begin{definition}[Error-correcting code]\label{def:err-code}
A block code is \( (\beta,t) \)-\df{burst-error-correcting},
if for all \( x\in\Sigma_{2} \), \( y\in\Sigma_{1}^{\Q} \) we
have \( \psi^{*}(y)=x \) whenever \( y \) differs from
\( \psi_{*}(x) \) in at most \( t \) intervals of size \( \le\beta \).
For such a code, we will say that a word \( y\in\Sigma_{1}^{\Q} \) is \( r \)-\df{compliant}
if it differs from a codeword of the code by at most \( r \) intervals of size \( \le\beta \).
\end{definition}
  \end{sloppypar}

\begin{example}[Repetition code]\label{xpl:tripling}
  Suppose that \( \Q\ge 3\beta \) is divisible by 3,
  \( \Sigma_{2}=\Sigma_{1}^{\Q/3} \), \( \psi_{*}(x)=xxx \).
  If \( y=y(1)\dots y(\Q) \), then \( x=\psi^{*}(y) \) is defined by
    \( x(i)=\maj(y(i),y(i+\Q/3),y(i+2\Q/3)) \).
    For all \( \beta\le \Q/3 \), this is a
    \( (\beta,1) \)-burst-error-correcting code.
    Repeating 5 times instead of 3 gives a \( (\beta,2) \)-burst-error-correcting
    code.
  \end{example}

  \begin{example}[Reed-Solomon code]\label{xpl:Reed-Solomon}
    There are much more efficient such codes than just repetition.
    One, based on the Reed-Solomon code, is outlined in Example 4.6
    of~\cite{GacsSorg01}.
    If each symbol of the code has \( l \) bits then the code can be up to \( 2^{l} \) symbols long.
    Only \( 2 t\beta \) of its symbols need to be redundant in order
    to correct \( t \) faults of length \( \beta \).
  \end{example}

Consider a (generalized) Turing machine 
\( (\Sigma,\tau) \) simulating some Turing machine \( (\Sigma^{*},\tau^{*}) \).
We will assume that the alphabet \( \Sigma^{*} \) is a subset of the set of  binary strings
\( \{0,1\}^{l} \) for some \( l<\Q \).
We will store the coded information in the interior of the colony, since it is more exposed 
to errors near the boundaries.

\begin{definition}\label{def:colony-interior}
Let 
\begin{align*}
  \PadLen 
\end{align*}
be a parameter to be defined later (in Definition~\ref{def:PadLen}).
A cell belongs to the \df{interior} of a colony spanning an interval \( I \)
if it is in \( \Int(I,\PadLen\B) \) (with the interior as in Definition~\ref{def:interior}).
The area within \( \PadLen\B \) of a colony end is called the \df{turn region}.
\end{definition}

Let \( (\upsilon_{*}, \upsilon^{*}) \) be a \( (\beta,2) \)-burst-error-correcting block code
with
\begin{align*}
  \upsilon_{*}: \{0,1\}^{l} \cup \set{\emptyset}
   \to\{0,1\}^{(\Q-2\cdot\PadLen)\B}.
\end{align*}
We could use, one of the above example codes, 
but we require that there are some fixed Turing machines
\( \Encode \) and \( \Decode \) computing them:
 \begin{align*}
   \upsilon_{*}(x)=\Encode(x),\quad 
   \upsilon^{*}(y)=\Decode( y).
 \end{align*}
 In Section~\ref{sec:length-work-period} we will be more specific about the choice of code
 and its space and time requirement.

 Recall that our Turing machine has some special states, among others: \( 0,1,\new_{0},\new_{1} \).
We require that at least some of these, namely \( \new_{0} \) and \( \new_{1} \) have encodings
that are especially simple: so \( \upsilon_{*}(\new_{0}) \) and \( \upsilon_{*}(\new_{1}) \) can be
written down in a single pass of the Turing machine \( M \).

Let us now define the block code \( (\psi_{*}, \psi^{*}) \) used below in the
definition of the configuration code \( (\varphi_{*}, \varphi^{*}) \)  
outlined in Section~\ref{sec:hier-codes}:
\begin{equation}\label{eq:psi}
   \psi_{*}(a)  = 0^{\PadLen}\upsilon_{*}(a)0^{\PadLen}.
\end{equation}
The decoded value \( \psi^{*}(x) \) is obtained by first removing \( \PadLen \)
symbols from both ends of \( x \) to get \( x' \), and then computing \(
\upsilon^{*}(x') \).
 It will be easy to compute the configuration code from \( \psi_{*} \),
once we know what tracks of the tape need initialization.







\( 
 \)

\subsection{Rule language}\label{sec:language}

The generalized Turing machines \( M_{k} \) to be defined
differ only in the parameter \( k \).
We will denote therefore \( M_{k} \) frequently simply by \( M \),
and \( M_{k+1} \), simulated by \( M_{k} \),  by \( M^{*} \).
Similarly we will denote the colony size \( \Q_{k} \) by \( \Q \).

We will describe the transition function
\( \tau_{k}=\tau \)  mostly in an informal way, as procedures of a program;
these descriptions are readily translatable into a set of \df{rules}.
Each rule consists of some (nested) conditional statements,
similar to the ones seen in an ordinary program:
 ``\textbf{if} \textit{condition} \textbf{then} \textit{instruction}
\textbf{else} \textit{instruction}'', 
where the condition is testing values of some fields of the observed cell-pair, and
the instruction can either be elementary, or itself a conditional statement. 
The elementary instructions are an \df{assignment} of a value to a field
of a cell symbol, or a command to move the head.
It will then be possible to write one fixed \emph{interpreter} Turing machine that carries
out these rules, assuming that the whole set of rules is a string and each field is also represented
as a string.

Assignment of value \( x \) to a field \( y \) of the state or cell symbol will
be denoted by \( y \gets x \).

Our description of rules is informal, making them sometimes look more like a procedure with
many steps.
This will just mean that the rule has some of its own dedicated fields to which it can refer and which
it can also set: with the help of these, indeed a sequence of actions can take place.
One of these fields may indicate which procedure is being performed.
Typically, only an element of the current cell-pair would carry the field indicating the
procedure being performed, but in some cases it would be a track on a larger area.
When a procedure ``calls'' another one, it will always be clear which one is being performed
and which one is just waiting for the called one to finish---there is no recursion.
Rules can also have parameters, like \( \MoveFront(d) \).
This parameter can also be seen as just referring to some field.
Similarly, when we say that a procedure \df{returns} a value, it just sets a certain field.

We may refer to two procedures performed one after the other, even 
when the first one does not move the head, like
\( d\gets\ProcessPayload(j) \) followed by \( \MoveFront(d) \).
The translation of this high-level description into
nested if-then-else instructions would combine the two procedures into one.

\subsection{Fields}\label{sec:fields}



A properly formatted configuration of \( M \) splits the tape into blocks of \( \Q \)
consecutive cells called \df{colonies}.
One colony of the tape of the simulating
machine represents one tape cell of the simulated machine.
The two colonies that correspond to the current cell-pair of the
simulated machine is scanning are called the \df{base colony-pair}.
A colony-pair can also be formally defined,
for the program, based on some field values in cells.
Sometimes the left base colony will just be called the \df{base colony}.
Most of the computation proceeds over the base colony-pair.
The direction of the simulated head movement, once figured out by the computation,
is called the \df{drift}.
The neighbor colonies of the base colony-pair may not be adjacent, in which case there will be
a \df{bridge} between them formed by neighboring (not necessarily adjacent) tape cells.
The possible space between neighbor colonies other than
the base colony-pair will be filled by stem cells (see below).

Let us describe some of the most important fields we will use in the tape cells;
others will be introduced later.
\begin{description}
\item[Procedures]
  Some fields will just indicate which procedure is currently active:
  for example
  \begin{align*}
    \rBoot,\rSimulate,\rHeal,\rRebuild,\rRebuildHeal.
 \end{align*}
The basic simulation activity is called the procedure \( \rSimulate \): when it is active,
 we may also say that the computation is in \df{normal mode}.
The \df{booting} procedure is used on the highest level of the simulation hierarchy; this level will be
raised if the length of computation time of the simulated machine \( G \) makes it necessary.
When this procedure is active,
we may also say that the computation is in \df{booting mode}.
The \df{healing} procedure tries to correct some local fault in simulation due to a couple of neighboring
bursts, while the \df{rebuilding} procedure attempts to restore the colony structure
on the scale of a couple of colonies.
When it is active, we may say that the computation is in \df{rebuilding mode}.
The rebuilding procedure may also need some healing; this is handled by \( \rRebuildHeal \).

\item[Info] The  \( \Info \) track of a colony of \( M \)
  contains the string that encodes the content of the simulated cell of \( M^{*} \).

\item[\( \Payload \)] The \( \Payload \)
  field of each cell contains a tape segment of the simulated machine \( G \).
  See Remark~\ref{rem:duplication} below on duplication.  

\item[Address] The field \( \Addr \)
of the cell shows the position of the cell in its colony:
it takes values in \( \lint{0}{\Q} \).

\item[Drift] The direction in \( \{-1,1\} \) in which the simulated head moves will be recorded on the track
 \( \Drift \).

\item[Age, Sweep] The \( \Age \)
  field keeps track of the step number of the computation within the work period of a colony pair.
  The work period will consist of some consecutive \df{stages}, whose beginning is marked by a certain
  value of \( \Age \).
  When a stage is finished (as seen from other indicators), the \( \Age \) may jump to the starting value of
  the next stage:
  the number of actual steps in a stage may not always be the same, but upper bounds are established.

  In some parts of the program (like the transfer phase),
  new cells may be inserted, causing the \( \Age \) field to experience a---harmless---jump.
  On the other hand, in these parts, the front will make sweeps over the whole colony pair;
  So instead of the \( \Age \) a field called \( \Sweep \) will be used that just counts
  which sweep is being performed.

\item[Kind] Cells will be designated as belonging to a number of
  possible \df{kinds}, signaled by the field \( \Kind \)
with values \(  \New, \Booting, \Stem, \Member_{0},\Member_{1},\Bridge, \Outer \).
Here is a description of their role.
\begin{itemize}
\item The kind \( \New \) has been discussed before.
\item A cell is of the \( \Booting \) kind if it is on the top level of simulation (see Section~\ref{sec:hier}).

\item Cells of the base colony-pair are of type \( \Member_{0} \) and \( \Member_{1} \) respectively.
  Members of other colonies have the kind \( \Outer \).

\item If the two base colonies are close but not adjacent then there will be \( <\Q \)
  adjacent cells of type \( \Bridge \) between them, extending the left
  base colony towards the right one.

\item \( \Stem \) is the kind of cells filling the space in between colonies other than the two base colonies.
\end{itemize}

\item[Heal, Rebuild] During healing and rebuilding, some special fields of the state and cell are used,
  treated as subfields of the field \( \Heal \) or \( \Rebuild \).
  In particular, instead of numbering their individual steps by an \( \Age \) field,
  these procedures make left-right and right-left \df{sweeps} over their interval of operation.
  There will be a \( \Heal.\Sweep \) field (or, respectively, \( \Rebuild.\Sweep \))
  showing the number of the current sweep.
  A cell will be said to be \df{marked for rebuilding} if any part of the \( \Rebuild \) track is defined.

\end{description}

\begin{remark}
  The \( \Outer \) kind is redundant:
  whether a cell is outer can be computed from its \( \Drift \) and \( \Age \) fields.
  But we use it for clarity.
\end{remark}

\begin{remark}[Duplication]\label{rem:duplication}
  The \( \Info \) track of a colony encodes the content \( s^{*} \) of the simulated cell of \( M^{*} \).
  In particular, the \( s^{*}.\Payload \) field contains the tape segment of the simulated machine \( G \)
  represented by the simulated cell.
  The \( \Payload \) track of the colony will represent the same tape segment (cut up
  into pieces in its individual cells).
  This duplication will not create too much space redundancy---due to
  the small number of levels relative to the size of the computation.
  (It could also be avoided, since at any one time the \( \Payload \) track only needs to represent
  some small part of \( s^{*}.\Payload \), the part currently worked on, therefore its ``bandwidth''
  could be kept small.
  In~\cite{GacsSorg01} for cellular automata this approach led to a constant
  factor space redundancy.)
  
\end{remark}

\subsection{Head movement}\label{sec:sweep}

The global structure of a work period is this:
\begin{description}

\item[Simulation phase]
Compute the new state of the simulated cell-pair, and the simulated direction (called the drift).
Then check the ``meaningfulness'' of the result.

\item[Transfer phase]
  The head moves into the neighbor colony-pair
  in the simulated head direction called drift (creating and destroying bridges if needed).
  In this phase, the number the current sweep is shown on the \( \Sweep \) track.
\end{description}

As the head leaves behind a cell, this remembers the last \( \Age \) and \( \Sweep \) value.
In all important cases as we will see, the simulation
will recognize from the neighborhood of the head if some age jump happened by error.


\begin{definition}[Front]\label{def:front}
  The position towards which the \( \Age \) values of the cells increase both
  from left and right will be called the \df{front}.
  The direction of the front is towards the smaller step values.
\end{definition}

Globally in a configuration, due to earlier faults, there may be more than one front, but locally
we can talk about ``the'' front without fear of confusion.

 Bridges between colonies present some extra complication---let us address it.

\begin{definition}[Gaps]\label{def:gaps}
If the bodies of two cells are not adjacent, but are at a distance \( <\B \) then the space
between them is called a \df{small gap}.
We also call a small gap such a space between the bodies of two colonies.
On the other hand, if the distance of the bodies of two colonies is \( >\B \) 
but \( <\Q\B \) then the space between them is called a \df{large gap}.
\end{definition}

A large gap between two colonies will be filled by a bridge when they become a base colony-pair.
A bridge is always extending the left member colony, except possibly during transfer while
the colony pair is moved.
Building a bridge or making repairs may involve
``killing'' some cells that are in the way and replacing them with new ones, via
the replacement action of Definition~\ref{def:dictated}.

Due to the zigging and feathering requirements mentioned in Section~\ref{sec:novelties},
moving the front will actually be a complex procedure itself, described below at~\eqref{eq:MoveFront}.
This procedure combines zigging and feathering as described below, and
uses the parameters
\begin{align}\label{eq:FDef}
 \Z = \passno^{2+\rho},\; \F = \Z\passno^{2+\rho}
\end{align}
with \( 0<\rho<1/4 \).
The choices will be motivated in Sections~\ref{sec:feathering} and~\ref{sec:pass-cleaning}.

\subsubsection{Zigging}\label{sec:zigging}

A \df{zigzag} movement will check the consistency of a few cells around the front.
The process creates a \df{frontier zone} of about \( \Z/2 \) cells around the front,
where \( \Z \) was defined in~\eqref{eq:FDef}.
In normal mode, this interval is recognizable just from
the \( \Age \) track, but during rebuilding will be marked on a special track
(see Section~\ref{sec:rebuilding}).
In every second step of moving in any direction (say at every even value of \( \Age \)),
the head will perform a forward-backward-forward
zigzag: going \( \Z \) steps ahead of the center of the frontier zone, then \( \Z \) steps behind it.

The turns while doing this are \df{small turns} defined in Section~\ref{sec:feathering},
meaning that a few more steps (normally at most 2) may be needed to find a turning point.

The step counting for zigging can be done locally, in the current cell-pair.
 
\begin{remark}\label{rem:zigging-choices}
  The size of the parameter \( \Z \) in~\eqref{eq:FDef}
  and with it the size of the frontier zone
  is motivated, among others, in the proof of Lemma~\ref{lem:keep-directed} where
  it has to withstand a large number of bursts.
  The forward-zigging property  allows, among other uses, to recognize a pair of opposing
  frontier zones.
\end{remark}

\subsubsection{Feathering}\label{sec:feathering}

Example~\ref{xpl:need-feather} suggests that our Turing machine
should have the property that between two turns
on the same point \( x \), it should pass \( x \) at least once.
We will call this property \df{feathering},
referring to the picture of the path of the head in a space-time diagram, as
in Example~\ref{xpl:feathering}.
In fact in some cases we will require more:

\begin{definition}\label{def:feathering}
  A Turing machine has the \( c \)-\df{feathering} property for a \( c>0 \), if  after
a right-to-left turn on point \( x \), the next right-to-left turn at a point \( \ge x \)
must be at a point \( \ge x + c\B \), and similarly for left-to-right turns.  
\end{definition}

 The following example suggests that any computation can be reorganized to accomodate feathering,
 at the price of at most a logarithmic delay.

\begin{example}\label{xpl:feathering}
Suppose that, repeatedly, arriving from the left at position 1, the head decides to turn left.
It can then make its turns at the following sequence of positions:
\begin{align*}
 1, 2, 1, 3, 1, 2, 1, 4, 1, 2, 1, 3, 1, 2, 1, 5, 1, 2, 1, 3, 1, 2, 1, 4, 1, 2, 1, 3, 1, 2, \dots
 \end{align*}
\end{example}

\begin{definition}[Digression]\label{def:digression}
Whenever a turn is postponed since the head is not allowed to turn due to feathering,
the simulation to be carried out by the head is suspended until the head returns.
This is called a \df{digression}.
\end{definition}

If in the original execution the head turned back \( t \) consecutive
times to the left from position \( p \), then now it will 
turn back from somewhere in a zone of size \( O(\log t) \) to the right of \( p \) in 
each of these times.
Computing the exact turning point is not necessary, but the following lemma
will be useful.

\begin{lemma}\label{lem:feathering-lb}
  Suppose that the Turing machine has the feathering property.
  For each \( n\ge 0 \), if during some time interval the head passes \( 2^{n} \) times to
  the right from point \( x \), then during the same interval, it must reach \( x+n\B \).
  The analogous statement holds for passing to the left.
\end{lemma}
The proof is easy by induction.

\begin{remark}\label{rem:amortized-delay}
  The amortized delay caused by \( c \)-feathering is, in fact, only by a factor of \( c \) times the
  number of turns.
  Indeed, given a Turing machine \( T \) we can build a Turing machine \( T' \)
  with the \( c \)-feathering property simulating \( T \) as follows.
  Machine \( T' \) is almost like machine \( T \), except that it keeps a marker at the place
  of the simulated head of \( T \), and a track \( \Pass \) with values \( 0,\pm 1 \)
  marking the place of each past turn that has not yet been passed over.
  It is set to 0 in each cell that the head passes over, 1 when the head turns left (since next time
  it must be passed to right), and \( -1 \) if it turns right.
  (The sign is only for exposition purposes, as it always follows from the context.)
  Suppose that the head of \( T \) must make a right-to-left turn.
  If \( \Pass=0 \) then \( T' \) just turns (moving the marker with it).
  Otherwise the head leaves the simulated head in place, and moves right by \( c \) cells.
  Every time it encounters another cell with \( \Pass=1 \) it resets the counter,
  moving \( c \) cells again.
  If it encounters no such cell then it turns.
  As it returns to the simulated head, it continues the simulation.

  We claim that when simulating a machine making \( n \) turns,
  the total delay due to these digressions is at most \( 2 c n \).
  Indeed, every rightward step during a right-to-left digression of \( T' \)
  is within distance \( c \) of a past right-to-left turn, which on the other hand
  can be attributed to a right-to-left turn of the original computation of \( T \).
\end{remark}

Our simulating Turing machine will have two different feathering properties: it will 
obey 1-feathering for all its turns, but on certain kinds of turn called \df{big turn}
will obey \( \F \)-feathering for the parameter \( \F \) defined in~\eqref{eq:FDef}.
The \( \F \)-feathering property will sometimes force the head during the work period
to go beyond the boundaries of the colony-pair, but only to a limited extent:

\begin{definition}\label{def:PadLen}
Let 
\begin{align}\label{eq:PadLen}
  \PadLen &= 4\F\log\Q.   
\end{align}
For an interval \( I \) spanned by a colony pair we call the 
the \df{turn region} the set \( \Int(I,-\PadLen)\setminus C \), that is the close
outside neighborhood of \( I \).
\end{definition}

\begin{description}
\item[Small turn]
  Whenever the head needs to turn (for example during zigging), the event will be called a \df{small turn}.
\begin{align*}
 \Pass,
\end{align*}
whose default value is 0, will be set to \( \pm 1 \).
Consider right-to-left turns, the left-to-right turns are analogous.
The head then arrives at a cell-pair \( (x,y) \) from the left.
If \( y \) has \( \Pass=1 \) then the head is not allowed to turn left: it continues right.
If \( y \) has \( \Pass=0 \) and the head turns left then \( y \) gets \( \Pass\gets 1 \).
In both cases, \( x \) gets \( \Pass\gets 0 \).
The event when a cell with \( \Pass=0 \)
is not found within \( 3\Delta \) steps
(with \( \Delta \) defined in~\eqref{eq:Delta} below) will be called \df{small turn starvation}.
Then start or restart healing, (see below) but still don't turn.

Suppose that the current cell-pair is \( (x,y) \), and
the next step replaces \( y \) by some \( y' \) as in Definition~\ref{def:dictated}. 
Then \( y' \) inherits the \( \Pass \) value of \( y \).

The set \( \New \) of states will only have two elements, \( \new_{0} \) or \( \new_{1} \), 
where \( \new_{i} \) is a state with field \( \Pass=i \).
When a replacement operation takes place as in Definition~\ref{def:directed}, then
the \( \Pass \) field of the cell being replaced will be inherited by the cell replacing it.

\begin{remark}\label{rem:no-kill-on-turn}
As a consequence of the above rule, a cell will \emph{never be killed} just when
the head turned back from it.
\end{remark}

  In normal and rebuilding mode, a zigging move is done only after every two steps of moving the front:
  this leaves every second cell with \( \Pass=0 \) in the wake of this movement.

\item[Big turn]
  The turns of the front will be called \df{big turns}, carried out by 
the procedure
\begin{align}\label{eq:MoveFront}
   \MoveFront(d), d\in\{-1,0,1\}.
\end{align}
The possibility of \( d=0 \) will be used in Section~\ref{sec:payload}; let it simply mean
calling first \( \MoveFront(1) \) and then \( \MoveFront(-1) \).
A big left turn will be performed by calling \( \MoveFront(-1) \).
The procedure is governed with the help of the field:
\begin{align*}
     \BigDigression\in\{ -1,0,1,\dots,\F \}\cup\{*,\omega,\delta_{-1},\delta_{1} \}
\end{align*}
where \( * \) means ``undefined'', this is the default value.
We will have \( \Z/2 \) consecutive cells traveling with the head, storing one and the same value
of \( \BigDigression \) (except while it is updated): 
we will call this the \df{digression marking zone}, or \df{D-zone}.

Let us describe the actions performed on account of a big right-to-left turn of the front in normal mode;
left-to-right turns are similar.
The somewhat simpler case of big turns of the rebuilding procedure
will be described In Section~\ref{sec:rebuilding}.
So assume that the front has been moving right, calling \( \MoveFront(1) \) repeatedly.
It carries the D-zone, and increases \( \BigDigression \) in it
by 1 at every step until \( \BigDigression=\F \) or until it encounters \( \Z/2 \)
consecutive cells marked \( \BigDigression=\omega \), called a 
the \df{footprint} of an earlier right-to-left turn.
In this case, it resets \( \BigDigression\gets 0 \) in the D-zone as
it continues right (erasing the old footprint in the process).

When \( \MoveFront(-1) \) is called, the right movement still continues (shifting the D-zone, and
making zigs of size \( \Z \) ahead and behind its center) until
\( \BigDigression=\F \).
The area between the front and the D-zone is filled with \( \delta_{1} \) (it would be \( \delta_{-1} \)
in case of a digression towards the left).
The old footprint is erased as it is passed over.
Once a big left turn is allowed since \( \BigDigression=\F \),
we set \( \BigDigression\gets \omega \) in the D-zone, leaving a new footprint.
The head moves back to the front, carrying \( \BigDigression=-1 \)
in the D-zone.
Once the D-zone is at the front,
the front is moved one step left, and the D-zone is moving with it,
setting \( \BigDigression\gets 0 \).


In normal mode, when the search for a place of big turn takes longer than \( 3\F\log\Q \) steps
this will be called a \df{big turn starvation}.
To recognize this, the D-zone can also maintain a field \( \FrontAddr \) to keep track of the
address difference between it and the front.
This will be recognizable from the distance of the D-zone from the front as seen from the address
and age tracks, and then rebuilding will be called.
\end{description}

\begin{remark}\label{rem:MoveFront}
  The structure of the program will be such that \( \MoveFront(d) \) or \( \ProcessPayload(j) \) (see later)
  will only be called when a previous execution of \( \MoveFront(\cdot) \) finished.
\end{remark}

Machine \( M \) will have the property that after a fault-free
path passed over a clean interval, both small turns and big turns can happen without too
long digressions.
We give here only an informal argument; formal proof must wait until 
a (more) complete definition of the simulation.
Zigs are by the definition spaced by \( \ge 2 \) cells apart, making sure that the points
with \( \Pass \ne 0 \) are in general at a distance of \( \ge 2 \) apart.
Healing can create only a constant number of segments of \( \Pass\ne 0 \) of size \( \le 3\Delta \)
with \( \Delta \) defined in~\eqref{eq:Delta}.
As for big turns, Example~\ref{xpl:feathering} (for \( \F=1 \)) shows that
a big turn attempt will be delayed by at most \( \F \) times the logarithm
of the total number of big turns inside a colony or a rebuild area.

The simulated Turing machine will also have the feathering property,
therefore the simulation will not turn back 
from one and the same colony repeatedly, without having passed it in the meantime.

\begin{sloppypar}
\begin{remark}\label{rem:big-turns}
  The size of the parameter \( \F \) is motivated by the proof of Lemma~\ref{lem:pass-clean}.
  Here is a sketch of the argument (it can be safely skipped now).
   At some time \( t_{0} \) in some interval \( I \)
  we will have clean subintervals \( J_{k}(t_{0}) \), \( k=1,2,\dots \)
  of size \( \ge 6\Z\B \) in which no fault will appear, and which are separated from each other by areas of
  size \( O(\passno^{2}\Z\B)\ll\F\B \).
  For times \( t>t_{0} \) we will track the maximal clean intervals \( J_{k}(t) \) containing
  the middle of \( J_{k}(t_{0}) \).

  Assume that the head passes over \( I \) noiselessly left to right and later
  also noiselessly from right to left.
  If the head moves in a zigging way to the
  right then the Attack Cleaning property will clean out the area between \( J_{i}(t) \) and \( J_{i+1}(t) \),
  joining them.
  This does not happen only in case of a big turn from the right end of \( J_{i}(t) \).
  But then in the next pair of passes over \( I \), the \( \F \)-feathering property implies that
  the big turn from the end of \( J_{i}(t) \) is at least a distance \( \F\B \) to the right.
  Our choice of \( \F \) implies that then \( J_{i}(t) \) will be joined to \( J_{i+1}(t) \).
  So two noise-free passes would join all the intervals \( J_{i}(t) \) into a clean area.  
\end{remark}  
\end{sloppypar}

\subsection{Simulation phase}\label{sec:simulation-phase}

As mentioned in Section~\ref{sec:fields} the work period will be divided into \df{stages}
starting at specific values of the \( \Age \) field.
These stages are grouped into two big phases: the \df{simulation} phase and the \df{transfer}
phase.
The simulation phase, governed by the \( \Compute \) rule,
computes new values for the current cell-pair of the
simulated machine \( M^{*} \) represented by the current (base) colony-pair,
and the move direction of the head of  \( M^{*} \).
The cell state of \( M^{*} \) will be stored on the track \( \Info \) of the
representing colony.
The move direction of \( M^{*} \) 
will be written into the \( \Drift \) field of \emph{each} cell of the base colony-pair
(filling the whole track with the same symbol \( d\in\{-1,1\} \)).

Let
\begin{align}\label{eq:stain}
  \beta' &= \beta+2\CSpill,
\\   \CStain &= 2\beta'+ 1.
\end{align}
The rule \( \Compute \) relies on some fixed \( (\CStain, 2) \) burst-error-correcting
code, moreover it expects each of the words found on the \( \Info \) track to be
2-compliant (Definition~\ref{def:err-code}).

The rule starts with checking that the input colonies
are compliant using a rule \( \rul{ComplianceCheck} \).
Then essentially repeats for \( j=1,2,3 \)
the following \df{stages}: decoding, applying the transition, encoding.
It uses some additional tracks like \( \Work \) for most of the computation, and outputs
its result onto the \( \Hold[j] \) track.
The \( \Info \) track will not be modified before all the \( \Hold[j] \) tracks are written.

In more detail:
\begin{enumerate}
\item At the start, the current cell-pair is the left pair of cells of the left
  member of the base colony-pair.
  
  Everywhere outside the base colony-pair, the \( \Drift \) values (as well as
  the increasing \( \Age \) values) are pointing towards it.

\item\label{i:turn-region}
  During the execution, the big turns may occur outside the colony-pair,
  in the turn region as in Definition~\ref{def:PadLen}.
  If there is no neighboring colony then outer adjacent stem cells will be used, or added as needed.  

\item For \( j=1,2,3 \), call \( \rul{ComplianceCheck} \) on the \( \Info \) track of
  both colonies of the pair, and
  write the resulting bit into the \( \fld{Compliant}_{j} \) track of each.
  
  Then pass through each colony of the pair and for each address \( i \),
  if in the cell with this address, the majority
  of \( \fld{Compliant}_{j} \), \( j=1,2,3 \) is false, then turn this cell
  into the \( i \)th cell of the colony representing the state \( \new_{0} \).
  Recall that in Section~\ref{sec:coding}, we required that the codes of the
  states \( \new_{i} \), \( i=0,1 \) are simple enough so that they can be written in a single pass.

  (We could have used any other state instead of \( \new_{0} \) here: just some simple default
  state is needed.)

\item For \( j=1,2,3 \)       
  do the following, using some work tracks:
  \begin{enumerate}
    
  \item Calling by \( \va \) the pair of strings found on the \( \Info \) track of
    the interiors \( \Int(C,\PadLen) \) of the base colonies \( C \),
    decode it into the pair of strings 
\begin{align}\label{eq:decoded}
 (\tilde\va_{0},\tilde\va_{1}) = \tilde\va=\upsilon^{*}(\va)  
\end{align}
    (the current state of the simulated cell-pair), and
    store it on some auxiliary track in the base colony-pair.
    Do this by computing \( \tilde\va = \Decode(\va) \)
    on some simulated Turing machine.
    (The time complexity of this procedure will be discussed in Section~\ref{sec:redundancy}.)

  \item \label{i:comp.trans}
    Compute \( (\va',d)=\tau^{*}(\tilde\va,\alpha) \),
    where \( \alpha=\True \) if the pair of colonies is adjacent, else \( \False \).
    This step needs elaboration for two reasons.
    First, part of this transition is the processing of track \( \Info.\Payload \), which represents
    the tape of the Turing machine \( G \) simulated by the machine \( M^{*} \).
    Second, the program of the transition function \( \tau^{*} \) is not written explicitly anywhere
    (this is a ``self-simulation'' situation).

    Both issues will are discussed in detail in Section~\ref{sec:self-simulation}.
    
  \item\label{i:comp.write}
    Write the encoded new cell states \( \upsilon_{*}(\va') \) onto the
    \( \Hold[j].\Info \) track of the interior of the base colony-pair.
    Write \( d \) into the \( \Hold[j].\Drift \) field of \emph{each cell} of
    the left base colony.

    A field called \( \Replace \) is used.
    Its value can be \( * \) (undefined), or an element of the set \( \New \).
    If one of the new states of the simulated cell pair belongs to \( \New \)
    (that is, the rules dictate a replacement,
    as in Definition~\ref{def:dictated}), then write it onto the \( \Hold[j].\Replace \) track
    everywhere; else the values on the track will be undefined.
    There is enough capacity in a cell to record this value of a simulated cell
    (which can have many more states),
    since the set \( \New \) has only two possible elements in \( M^{*} \) as well as \( M \).
  \end{enumerate}

\item\label{i:vote}
  Sweeping through the base colony-pair,
  at each cell compute the majority of \( \Hold[j].\Info \), \( j=1,\dots,3 \),
  and write it into the field \( \Info \).
  Proceed similarly, and simultaneously, with \( \Drift \) and \( \Replace \).
  
\item\label{i:trickle-down} If the \( \Output \) field of the simulated cell is defined, write it
  into the output field of the left end-cell of each colony.
   
\end{enumerate}

Part~\ref{i:trickle-down} achieves that when the computation finishes on some
simulated machine \( M_{k} \),
its output value in cell 0 of \( M_{k} \) will ``trickle down'' to the output field of  cell 0 of \( M_{1} \),
as needed in Theorem~\ref{thm:main}.

The transfer phase (see Section~\ref{sec:transfer})
will use the information in the \( \Replace \) track to carry out, in simulation, the
replacement action of Definition~\ref{def:dictated}.



\subsection{TM simulation and self-simulation}\label{sec:self-simulation}

Let us elaborate stage~\ref{i:comp.trans} of Section~\ref{sec:simulation-phase}.

\subsubsection{Handling the payload}\label{sec:payload}

The tape of machine \( M_{k} \) contains a track called \( \Payload \) that represents the tape of
the target simulated Turing machine \( G \) of Theorem~\ref{thm:main}.
In a simulation work period, the \( \Payload \) track on level \( k \)
will come into play only if the head of \( G \) is simulated on level \( k+1 \).
This simulation will be performed only when the procedure
\begin{align*}
   \ProcessPayload(j)
\end{align*}
is called.
Here, \( j\in\{0,1\} \) shows whether the represented head of \( G \) is inside the left or the right
colony of the current colony-pair.
The procedure returns a number \( d\in\{-1,0,1\} \), to be used as an argument
to \( \MoveFront(d) \).
Here are the details.

\begin{sloppypar}
Consider the simulation of machine \( M^{*} \) on machine \( M \).
The track \( \Payload \) (after decoding) consists of \( \Payload.\Tape \) and \( \Payload.\Pos \).
The track \( \Payload.\Tape \) contains the tape segment of \( G \) represented by the colony.
Track \( \Payload.\Pos \) is used only from the left neighbor colony in case
of \( \ProcessPayload(0) \), else from the right one.
For definiteness, assume it is in the left one.
Its value has the form \( \Pos=(a,j') \) with \( j'\in\{0,1\} \).
Here, \( a \) is the address of the cell of \( M \)  in the colony
containing the represented head of \( G \).
The rules ensure that when
\( \ProcessPayload(0) \) is called then we will never have \( a=0 \), that
is \( a \) will never be at the very left end of the left neighbor colony.
This way, the current simulated cell-pair is always inside the current colony-pair.

If the level \( k \) is larger than 1 then 
each cell of \( M \) itself represents a tape segment of \( G \);
in this case the value \( j' \) shows whether left cell or the right cell of the cell-pair of \( M \)
at address \( a \) contains the represented head of \( G \).
\end{sloppypar}

The rule works on the decoded states \(  (\tilde\va_{0},\tilde\va_{1}) \)
of the current cell-pair of \( M^{*} \), as in~\eqref{eq:decoded}.
It copies \( \tilde\va_{0}.\Payload.\Tape \) and \( \tilde\va_{1}.\Payload.\Tape \) as consecutive
strings onto the \( \Payload.\Tape \) track of the current colony-pair.
Then (assuming \( j=0 \)) it recovers \( (a,j')=\tilde\va_{0}.\Payload.\Pos \),
goes to address \( a \) of the left neighbor colony,
and repeats the following action (see below for how many times).
\begin{itemize}
\item[(G)] 
If \( k=1 \) then the cells \( a-1,a \) of the \( \Payload.\Tape \) track
actually represent the content \( \va \) of two neighbor cells, so apply the transition function
computing \( (\va',d')\gets\tau_{G}(\va) \), and set \( \va\gets\va' \).

If \( k>1 \) then call \( d'\gets\ProcessPayload(j') \) on the cell-pair at address \( a \) (recursively).
In both cases, follow this by \( \MoveFront(d') \).
\end{itemize}
How many times to perform action (G)?
At most \( \Q \) times, but stop earlier if the head would 
reach the left end of the left colony or the right end of the right one.

Return \( d=-1 \) if the head arrives in the left half of the left colony, 
\( d=1 \) if it is in the right half of the right colony,
and \( d=0 \) otherwise.

\subsubsection{New primitives}

In what follows describe the tools of self-simulation.
The simulation phase makes use of the track \( \Work \) mentioned above, and the track
\begin{align*}
   \Index
 \end{align*}
that can store a certain address of a colony.
Recall from Section~\ref{sec:language} that the program
of our machine is a list of nested
``\textbf{if} \emph{condition} \textbf{then} \emph{instruction}
\textbf{else} \emph{instruction}''
statements.
As such, it can be represented as a binary string 
 \begin{align*}
   R.
 \end{align*}
If one writes out all details of the construction of the present paper, this string \( R \)
becomes explicit, an absolute constant.
But in the reasoning below, we treat it as a parameter.
Let us provide a couple of \df{extra primitives} to the rules.
First, they have access to the parameter \( k \) of machine \( M=M_{k} \), 
to define the transition function
 \begin{align*}
            \tau_{R,k}(\va).
 \end{align*}
The other, more important, new primitive is a special instruction
 \begin{align*}
   \WriteProgramBit
 \end{align*}
in the rules.
When called, this instruction makes the assignment \( \Work\gets R(\Index) \).
This is the key to self-simulation: \emph{the program has
access to its own bits}.
If \( \Index=i \) then it writes \( R(i) \) onto the current position of the \( \Work \) track.

\subsubsection{Simulating the rules}

The structure of all rules is simple enough that they can be read and
interpreted by a Turing machine in reasonable time:

\begin{theorem}
There is a Turing machine \( \Interpr \) with the property that for
all positive integers \( k \), string \( R \) that is a
sequence of rules, and a pair of bit strings \( \va=(a_{0},a_{1}) \) with \( a_{j}\in\Sigma_{k} \),
 \begin{align*}
  \Interpr(R,0^{k},\va)=\tau_{R,k}(\va).
 \end{align*}
\end{theorem}
The proof parses and implements the rules in the string \( R \); each of these rules
checks and writes a constant number of fields.
Implementing the \( \WriteProgramBit \) instruction is straightforward:
Machine \( \Interpr \) determines the number \( i \)
represented by the simulated \( \Index \) field, 
looks up \( R(i) \) in \( R \), and writes it into the simulated \( \Work \) field.
There is no circularity in these definitions:
\begin{sloppypar}
  \begin{itemize}
  \item 
The instruction \( \WriteProgramBit \) is written \emph{literally}
in \( R \) in the appropriate place, as ``\(\WriteProgramBit \)''.
The string \( R \) is \emph{not part} of the rules (that is of itself).  
  \item On the other hand, the computation in
\( \Interpr(R,0^{k},\va) \) 
has \emph{explicit} access to the string \( R \) as one of the inputs.
  \end{itemize}
\end{sloppypar}

Let us show the computation step invoking the ``self-simulation'' in detail.
In the earlier outline, step~\ref{i:comp.trans} of Section~\ref{sec:simulation-phase}
said to compute \( \tau^{*}(\tilde\va) \)
(for the present discussion, we will just consider computing 
\( \tau^{*}(\va)=\tau_{k+1}(\va) \)), where \( \tau=\tau_{k} \),
and it is assumed that \( \va \) is available on an appropriate auxiliary track.
We give more detail now of how to implement this step:

\begin{enumerate}
\item Onto the \( \Work \) track, write the string \( R \).
To do this, for \( \Index \) running from 1 to \( \abs{R} \), 
execute the instruction \( \WriteProgramBit \) and move right.
Now, on the \( \Work  \) track, add \( 0^{k+1} \) and \( \va \).
String \( 0^{k+1} \) can be written since the parameter \( k \) is available.
String \( \va \) is available on the track where it is stored.
\begin{sloppypar}
 \item Simulate the machine \( \Interpr \) on track \( \Work \), computing
   \( \tau_{R,k+1}(\va) \).  
 \end{sloppypar}
\end{enumerate}

This implements the forced self-simulation.
Note what we achieved:

\begin{itemize}
  \begin{sloppypar}
\item On level 1, the transition function \( \tau_{R,1}(\va) \) is defined completely
when the rule string \( R \) is given.
It has the forced simulation property by definition, and
string \( R \) is \emph{``hard-wired''} into it in the following way.
If \( (\va',d)=\tau_{R,1}(\va) \), then
\begin{align*}
  a'_{0}.\Work\gets R(a_{0}.\Index)
\end{align*}
whenever \( a_{0}.\Index \) represents a number between 1 and \( \abs{R} \),
and the values \( a_{0}.\Sweep \), \( a_{0}.\Addr \) satisfy the conditions
under which the instruction \( \WriteProgramBit \) is 
called in the rules (written in \( R \)).
      \end{sloppypar}

      \begin{sloppypar}
\item The forced simulation property of the \emph{simulated}
transition function \( \tau_{R,k+1}(\cdot) \) is 
achieved by the above defined computation 
step---which relies on the forced simulation property of \( \tau_{R,k}(\cdot) \).
  \end{sloppypar}
\end{itemize}

\begin{remark}
  This construction resembles the proof of Kleene's fixed-point theorem, and even more
  some self-reproducing programs (like a program \( p \) in the language C causing the computer
  to write out the string \( p \)).
\end{remark}

\subsubsection{The length of a work period}\label{sec:length-work-period}

We claim that the number of steps in the work period is \( O(|R|\Q\F\Z^{2}) \),
where \( R \) is the program string, so \( |R| \) is actually just a constant.
Below we will first ignore the extra burden of zigging and feathering responsible for
the factor \( \F\Z \): without this, we get a bound \( O(|R|\Q\Z) \).

\begin{description}
\item[Coding-decoding]
The computing part of the work period performs three repetitions of
coding, the actual simulation work and decoding.
We need a code that corrects at least 3 bursts of  \( \beta \) cells each.
One way to do this is via a Reed-Solomon code.
There are encoding and decoding procedures for an \( n \)-symbol Reed-Solomon code
that take time \( O(n^{2}) \).
Since this is not linear, we can subdivide the \( \Q \) cells of a colony into segments of, say,
\( \Z \) cell each and encode-decode each of these separately.
This way the total number of steps spent on these procedures is \( O(\Q\Z) \).

\item[Payload]
The procedure \( \ProcessPayload(j) \) can take up to \( \Q \) simulation steps.
(It may make fewer steps if the simulated head reaches the edge of the colony-pair,
but then the next work period will have at least \( \Q \) simulation steps.)

\item[The rest]
The rest of the simulation takes only \( O(|R|\Q) \) steps.
Indeed, the number of fields is \( O(|R|) \), and notice that each field other than the one
having to do with \( \Payload \) is a number whose size is \( O(\Q^{2}) \).
We are now processing a colony representing a cell of \( M^{*}=M_{k+1} \), so the numbers
may have size \( (\Q^{*})^{2}=\Q_{k+1}^{2}=2^{2\cdot 1.5^{k+1}} \).
A cell of \( M \) has size \( \prod_{i<k}\Q_{i}=2^{1+1.5+\dots+1.5^{k-1}} \), more than enough
to store such a number.
The program involves comparing fields represented in cells found in the left and the right
colony of the colony pair, so it may need to copy the fields found in the right colony to the left one,
necessitating \( O(|R|\Q) \) steps; after this the comparisons can be done locally before the results
are carried back.
\end{description}

\subsection{Transfer phase}\label{sec:transfer}

Before the transfer phase, members of the base colony-pair \( C_{0},C_{1} \) have
cells of kind \( \Member_{0} \) and \( \Member_{1} \) correspondingly,
with a possible bridge between them.
In the transfer phase, control will be transferred to the
neighbor colony-pair in the direction of the simulated head movement which we called 
the \df{drift}, found on the \( \Drift \) track.
Whenever the \( \Replace \) track holds a defined value,
we will say that this is a \df{replacement} situation.

During the transfer phase, the range of the head includes the base colony-pair and a neighbor colony
called \df{target colony} in the direction of the drift.
At the beginning of the phase, the current cell-pair is the first cell-pair of \( C_{0} \).
Big turns happen in the turn region as in part~\ref{i:turn-region} of the
description of the \( \Compute \) rule.

In this phase, the front will move in sweeps, and the track \( \Sweep \) will be rewritten to
the number of the sweep being currently performed.
The first sweep will bring the head to the end of the turn region beyond the new neighbor colony.
Subsequent turns will therefore all happen closer.

Consider \( \Drift=1 \).

\begin{enumerate}
\item \label{i:target}
  Suppose that we don't have a replacement situation.
  In the first sweep, the head will travel right.
  Turn all elements of \( C_{0} \) into outer cells,
  and turn the elements of \( C_{1} \) into \( \Member_{0} \) cells.
  Then continue to the right, start a bridge (if necessary) towards the right
  (killing all possible non-adjacent stem cells in the way).
  If the right end of the bridge reaches an outer colony \( C_{2} \)
  before \( Q \) bridge cells are created, then pass to the right edge of \( C_{2} \).
  If \( Q \) bridge cells were created, stop at the right edge of this bridge---call \( C_{2} \)
  the bridge just created.
  Then sweep back to the left end of  \( C_{1} \), while turning the cells of \( C_{2} \) to 
  kind \( \Member_{1} \).

  In both cases, actually go to a distance \( 3\F\log \Q \) from the right end of \( C_{2} \)
  before attempting to turn.
  (This way, all later attempted left turns in the next work period
  will happen to the left of this one, and so any possible cause for alarm is encountered
  already now, in the transfer process.)

\item\label{i:replace}
  In the replacement situation, build a new colony \( C'_{1} \) adjacent on the right to \( C_{0} \).
  In the first sweep, perpetuate the value \( r \) found on the \( \Replace \) track.
  Write onto the  \( \Info \) track of \( C'_{1} \) the encoding of the value \( r \).
  (This requires two steps for each created cell \( x_{i} \) of \( C'_{0} \):
  first it has value \( r\in\New \), then it gets address \( i \), and its \( \Info \) field
  gets the \( i \)th symbol of the encoding of \( r \).)
  Then continue to the end of \( C_{1} \) (which should be beyond the end of \( C'_{0} \)).
  On the way back, replace the remaining elements of \( C_{1} \) with stem cells
  and set the kind of elements of \( C'_{1} \) to \( \Member_{1} \).

\end{enumerate}

A similar program is executed when \( \Drift=-1 \).
The values of \( \Drift \), \( \Replace \), \( \Sweep \) and \( \Addr \) always determine what step
to perform.
  
  Fault-checking during zigging will notice
when a burst compromises this process (for example when the end of a bridge would
``bite'' into another colony), by checking whether all boundaries it finds are legal (see the definition of
health in Section~\ref{sec:health}) and trigger healing (see later).

\subsection{Booting}\label{sec:booting}

Ideally, the work of machine \( M \) starts from a single active cell-pair of the Booting kind,
with addresses \( \Q-1 \) and 0, the middle cell-pair of a yet to be built colony-pair.
The \( \Payload \) track of the cell-pair holds a tape segment of the simulated Turing machine \( \G \),
along with the simulated head.
Such a cell-pair will be called a \df{booting pair}.
The segment consisting of this cell-pair will be extended left-right by booting cells,
eventually creating a colony-pair, as follows.

\begin{description}
  \begin{sloppypar}
  \item[Main work]
    Process the payload just as in Section~\ref{sec:payload}, for at most \( \Q \) processing steps
    (using \( \ProcessPayload() \) and \( \MoveFront() \)).
    All new cells encountered must be stem cells (blanks), and all the ones in the segment already created
    must be of the Booting kind; otherwise call alarm.
    In all this, use zigging and feathering.    
  \end{sloppypar}
  
\item[Lifting]
  Create the colony-pair around the original pair of booting cells (and turn its cells into member cells).
 \df{Lift} (copy) the \( \Payload \) track of its cells into
  the \( \Payload \) track of the cell-pair simulated by it.
  Start the booting procedure on the simulated cell-pair.
\end{description}

No decoding-encoding and repetition
mechanism is used to correct computational faults during the booting phase, since
we do not expect faults
during it---see the probability analysis in Section~\ref{sec:fault-estimation}.
(Of course faults could introduce booting cells during other
parts of the computation; this will be caught by the zigging mechanism.)

\section{Healing and rebuilding}\label{sec:healing}

Here we define the part of the simulation program that repairs local inconsistencies.

\subsection{Health}            \label{sec:health}

Structure is maintained with the help of a small number of fields.
The required relations among them 
allow the identification local inconsistency, and its correction provided it was caused locally.

\begin{definition}\label{def:domains}
The tuple \( \Core \) of fields is the tuple
\begin{align*}
  \Core = (\Kind,\Drift,\Replace,\Addr,\Age,\Sweep,\Rebuild, 
  \BigDigression,\FrontAddr).
 \end{align*}

 An interval of non-stem adjacent neighbor cells is a \df{homogenous domain} if
 its core variables with the exception of \( \Addr \) and \( \Age \) have the same value.
 \( \Addr \) increases one step at a time left to right.
 If we are in the simulation phase then \( \Age \) increases one step at time,
 either left to right or right to left.
 We don't require this In the transfer phase, the \( \Sweep \) track serves for health check instead.
The \df{left end} of a domain is the left edge of its first cell, and its \df{right end} is 
the right edge of its last cell.
A \df{left boundary} is the left end of a homogenous domain with either no left neighbor cell
or with a neighbor cell belonging to another homogenous domain.
Right boundaries are defined similarly.
\end{definition}

Health can be defined formally on the basis of the informal descriptions given here,
but the details would be tedious.
Recall Definition~\ref{def:front} of the front.

\begin{varenum}{H}
\item   A configuration consists of intervals of non-stem neighbor cells,
  with possibly stem cells between them.
  The health for each of these intervals is defined locally.
  No cell is marked for rebuilding, that is \( \Rebuild \) is undefined.

  As a non-local condition we will require that exactly one of these intervals contains
  the front: let us call this the \df{principal interval},
  and that the drift in all other intervals is directed towards the principal one.
 
\item  In the principal interval
  there is a base colony-pair, with possibly a bridge going from one member of the pair
  to the other one.
  Let \( I \) denote the interval containing this pair.

\item During the \( \Compute \) phase,
  outside the base colony-pair (both in the principal interval and elsewhere),
  all non-stem cells have their \( \Drift \) value directed towards the base colony-pair.
  There are possibly colonies of type \( \Outer \) adjacent to it and each other.
  They are called \df{left outer} colonies and their cells left outer cells, or
  \df{right outer} colonies depending on whether their drift is \( +1 \) or \( -1 \).
  Stem cells are also called outer cells (both left and right).

  In part~\ref{i:vote} of this phase, the values of \( \Drift \) and \( \Replace \) can change at the front.
  
\item\label{i:health.transfer}  During transfer, the base colony (of the pair)
  that is in the direction of the drift is possibly extended by a bridge towards the target colony.
  At this point there is no other bridge, and the front is at the tip of the new bridge.
  
  In the later part of this sweep, the new bridge already reaches the target colony.
  If the bridge extends to a full colony then this is converted to the appropriate kinds of member
  cells in the backward sweep.
  Domains ahead and behind the front show the changes done.
  Another possible change occurring at the front is replacement, when dictated by the \( \Replace \) track.
  In this case the front has the property that, for example when replacing a colony in the right direction,
  for each member cell created to replace an old member cell, \emph{the address of the new cell is not
    larger than the address of the one it replaces}.

\item\label{i:health.digression} There is a D-zone (see Section~\ref{sec:feathering})
  of length \( \Z/2\pm 1 \)
  cells, either at or ahead of the front, and the head is in or adjacent to it.
  There are possibly footprints of big turns.
  The length of such a footprint is \( \Z/2 \) when the head is not
  in it, possibly less as it is being created or erased.
\end{varenum}

\begin{definition}\label{def:legal-boundary}
  Let us call a bondary \df{legal} if it can occur in a healthy configuration.
  We say that a configuration is \df{pre-healthy} if it is healthy 
  except possibly condition in~\eqref{i:health.digression} above
  on the length of the D-zone and the length of the footprints of big turns,
  and some rebuilding marks.
\end{definition}

The following lemma shows that health of an interval of non-stem neighbor cells
is locally checkable.
 
\begin{lemma}\label{lem:boundaries-health}
  If an interval of neighbor cells in a tape configuration
  has only legal boundaries (including those at its ends) then it is pre-healthy.
\end{lemma}
\begin{proof}
  In what follows we don't repeat it but each statement is forced by the kind of boundaries
  allowed.
  There are several kinds: a colony boundary, the front, the boundaries of any footprint of a big turn,
  those of the big digression zone,
  and within that zone, the place where \( \BigDigression \) changes by 1.
  What matters is the values of the \( \Core \) variables in the cell pair around the boundary.

  \begin{enumerate}
\item Consider an interval \( I \) of neighbor cells.
  If  \( I \) contains cells that are not outer, or it contains both left and right outer cells
  then it also contains the front.

\item Assume that \( I \) contains only left outer cells.
  Then these consist of colonies possibly separated by stem cells, with drift pointing to the right,
  and possibly a front in a right turn region.
 The situation is similar when \( I \) consists of right outer cells.

In all other cases \( I \)
also contains an interval \( K \) consisting of neighbor non-outer cells.
  
\item Let \( s \) be the maximum age found in \( K \).
  If \( s \) is in the computing phase then \( K \) consists
  of a base colony-pair with a possible inner bridge connecting it, and the
  possible boundaries will force this.

\item Suppose now that \( s \) is in the transfer phase: then the 
  the drift over \( K \) is constant.
  Suppose that, for example, \( \Drift=1 \).
  Look at the description of the transfer phase in Section~\ref{sec:transfer}.
  Depending on whether we are in a replacement situation (defined by the \( \Replace \) track)
  and whether the we are in the first or second sweep, the boundary at the front completely
  determines the possibilities.
  The restriction on the addresses mentioned in~\eqref{i:health.transfer}
  above makes sure that there is enough space for the replacement to succeed.
  \end{enumerate}
\end{proof}

\begin{corollary}\label{crl:health-extension}
  Let \( \xi \) be a tape configuration that is micro-healthy on intervals \( A_{1}, A_{2} \) 
where \( A_{1}\cap A_{2} \) contains a whole cell body of \( \xi \).
  \begin{alphenum}
  \item Then \( \xi \) is also micro-healthy on \( A_{1}\cup A_{2} \).
  \item If \( A_{1}\cap A_{2} \) contains at least \( \Z/2 \) cells and \( \xi \) is healthy
    on both \( A_{1}, A_{2} \), then it is also healthy on \( A_{1}\cup A_{2} \).
  \end{alphenum}
\end{corollary}
\begin{proof}
The first part follows from the fact that micro-health is defined by boundaries.
In case of the second one, if a \( D-zone \)
intersects both \( A_{1} \) and \( A_{2} \) then it is contained entirely in one of them.
\end{proof}


\begin{lemma}\label{lem:3-boundaries}
  In a healthy tape configuration, over any interval of size \( <(1/4)\Z\B \) there are at most
  \begin{align}\label{e.eq:boundaries}
   \CBoundaries = 5
 \end{align}
  boundaries between domains.
\end{lemma}
\begin{proof}
  Two colony-ends can be close to each other in case of a big gap between
  two neighbor colonies.
  Add to this the front,
  one end of a digression zone,
  and one point where \( \BigDigression \) changes.
\end{proof}

In a healthy configuration, the possibilities of finding non-adjacent neighbor
cells are limited.

\begin{lemma}\label{lem:two-domains}
  An interval of size \( <\Q \) over which the configuration \( \xi \) is healthy
contains at most two maximal sequences of adjacent non-stem neighbor cells.
\end{lemma}
\begin{proof}
  By definition a healthy configuration consists of intervals covered by
  full colonies connected possibly by bridges, and possibly stem cells between these intervals.
An interval of size \( <\Q \) contains sequences of adjacent cells 
from at most two such intervals.
\end{proof}

\begin{lemma}\label{lem:infer-between}
In a healthy configuration, 
the state of a cell-pair shows whether they are outer, and also their direction towards the front.
The \( \Core \) track of a homogenous domain can be reconstructed from any pair of its cells.
\end{lemma}
\begin{proof}
Whether a cell is outer is computed from its age field.
If the cell is outer then \( \Drift \) shows its direction from the front,
else the increase direction of \( \Age \) shows it.
\end{proof}

\subsection{Stitching}\label{sec:stitching}

We will show that a configuration admissible over an interval of size \( >(1/2)\Q\B \)
can be locally corrected;
moreover, in case the configuration is clean, this correction
can be carried out by the machine \( M \) itself.

\begin{definition}[Substantial domains]\label{def:substantial}
Let \( \xi(A) \) be a tape configuration over an interval \( A \).
A homogenous domain of size at least \( 4\CStain\beta\B \) will be called \df{substantial}.
The area between two neighboring maximal
substantial domains or between an end of \( A \) and the closest substantial domain in \( A \)
will be called \df{ambiguous}.
It is \df{terminal} if it contains an end of \( A \).
Let
 \begin{align}\label{eq:Delta}
     \Delta &= (4 \CBoundaries+9)\CStain\beta.
 \end{align}
\end{definition}

In Section~\ref{sec:annotation}, we introduced the notion of \df{islands}: intervals
of size \( \le\CStain\beta\B \) with the property that if the configuration is changed in the
islands it becomes healthy.
Under normal circumstances, there will be at most 3 islands in any interval of size \( \Q\B \).
The size of a substantial domain assures that at least one of its cells is
outside an island, since even three neighboring islands have a
total size \( \le 3\CStain\B \).

\begin{lemma}\label{lem:ambiguous}
In an admissible configuration, each half of a 
a substantial domain contains at least one cell outside the islands.
If an interval of size \( \le  \Q\B \) of a tape configuration \( \xi \) differs from a  healthy tape 
configuration \( \chi \) in at most three islands, then 
the size of each ambiguous area is at most \( \Delta\B \).
\end{lemma}
\begin{proof}
The first statement is immediate from the definition of substantial domains.
By Lemma~\ref{lem:3-boundaries}, there are at most \( \CBoundaries \) boundaries in \( \chi \).
There are at most 3 islands.
Between islands and boundaries there are at most \( \CBoundaries+3-1 \)
non-substantial domains: of sizes \( < 4\CStain\beta\B \).
The islands have a total size \( <3\CStain\beta\B \) and the space between boundaries may add at
most \( \CBoundaries\B \).
Adding all these up we get the following multiple of \( \B \):
\begin{align*}
  4\cdot(\CBoundaries+2)\CStain\beta + 3\CStain\beta + \CBoundaries
  < (4 \CBoundaries+9)\CStain\beta = \Delta.
\end{align*}

\end{proof}

The following lemma forms the basis of the healing algorithm.

\begin{lemma}[Stitching]\label{lem:stitching}
In an admissible configuration, inside a clean interval,
let \( U,W \) be two substantial domains separated by an ambiguous area \( V \).
It is possible to change the tape on \( U,V,W \) using only information in \( U,W \) in such a 
way that the tape configuration over \( U\cup V\cup W \) becomes
micro-healthy (see Definition~\ref{def:legal-boundary}).
Moreover, it is possible for a Turing machine to do so gradually, changing and/or enlarging
the tape in \( U \) or \( W \) at the expense of \( V \), making a constant number of sweeps
over \( U\cup V\cup W \), with ``small turns'' as defined in Section~\ref{sec:sweep} at the ends of sweeps.
\end{lemma}
After the stitching, the tape configuration may only be \emph{micro-}healthy,
as the length of a digression zone or a footprint of a big turn may change slightly.
\begin{Proof}
  At any step below, if we find that \( U\cup V\cup W \) is micro-healthy then we stop.

\begin{step+}{step:stitching.colonies}
  If \( U \) consists of colony cells then let it be extended towards \( V \) until a colony boundary or \( W \)
  is hit.
  Assign all core variables in a way to keep the domain \( U \) homogenous.

  Then do the same with \( W \).
  If \( V \) gets eliminated then the boundary-pair
  between \( U \) and \( W \) is necessarily a legal one.
\end{step+}
Assume now that the above operations have still left \( V \).

\begin{step+}{step:stitching.1-colony}
  If for example \( U \) consists of colony cells but \( W \) does not, then extend \( W \) towards \( U \)
  until \( V \) is erased: the boundary obtained must be legal.
\end{step+}

\begin{step+}{step:stitching.2-colonies}
Assume that both \( U \) and \( W \) consist of colony cells.
\end{step+}
\begin{prooof}
  If both colonies are outer then we can turn all elements of \( V \) into stem cells.
  This situation will not be actually encountered since the front is never near the boundary between
  two outer colonies.

  If both colonies are inner then turn the cells between them into a bridge from \( U \) to \( W \).

  If for example \( U \) is inner and \( W \) is outer then, if the age is not in a the transfer phase
  towards \( W \) then fill \( V \) with stem cells; else
  fill \( V \) with bridge cells extending \( U \).
\end{prooof} 

\begin{step+}{step:non-colony}
  The case remains when neither \( U \) nor \( V \) consist of colony cells.
\end{step+}
\begin{prooof}
  If one of them consists of stem cells then extend it (does not matter which)
  towards the other until they meet.
  If both consist of bridge cells then extend either one of them towards the other until they meet.
  We will always have to end up with a legal boundary.
\end{prooof} 
\end{Proof}

\subsection{The healing procedure}\label{sec:healing-proc}

The healing procedures \( \rHeal \), \( \rRebuildHeal \)
and the rebuilding procedure \( \Rebuild \) look as if we assumed no noise or disorder.
The rules described here, however (as will be proved later), will also clean an area
locally---under the appropriate conditions.

Healing performs only local repairs of the structure: for a given (locally) admissible configuration,
it will attempt to compute a satisfying (locally) healthy configuration.
If it fails---having encountered an inadmissible configuration---then
the rebuilding procedure is called, which is designed to repair a larger interval.

To protect from noise, any one call of the healing procedure will change only 
a small part of the tape, essentially one cell: so a burst during healing
can only have limited impact.
Every healing operation starts with a survey zig around its starting point:
an illegal boundary or a boundary of a D-zone, called the \df{center}.

If the survey finds some possible healing to do then it performs one step of it, and returns.
Otherwise the ``attempt'' \df{fails}; in this case, it will build a rebuilding base,
an interval is defined with the help of a new field
\begin{align*}
   \Rebuild.\Base\in\{*,1\}.
\end{align*}
Its default value is \( * \).

\begin{definition}\label{def:rebuild-base}[Rebuilding base]
  A \df{rebuilding base} is an interval of \( 4\Delta \) cells with \( \Rebuild.\Base=1 \).
\end{definition}

%

Here are the details of healing.
Suppose that the \( \rHeal \) procedure is called at some position.
In what follows, turns will always be small turns, as defined in Section~\ref{sec:feathering}. 
In any newly created cell, 
\( \Drift \) is set backwards to the creating cell---to make sure that the head does not get
lost on the edge of the infinite vacant space when hit by a burst.
When finding homogenous intervals or boundaries, we ignore the track
\( \Rebuild.\Base \).

\begin{description}
\item[Survey]
  Look over an interval \( I \) consisting of \( 4\Delta \) cells left and as many right from the center.
  (As usual, when a neighbor cell is not found, one is created.)
  Let \( I'\subset I \) be an interval consisting of \( \Delta \) cells to the left and as
  many to the right of the center.

  \begin{itemize}
    \instr
    If the ambiguous areas in \( I \) cannot be covered by 3 intervals of size \( \le\Delta \) (separated by
    substantial homogenous intervals) then go to part \emph{Fail}.

    \instr If \( I' \) is pre-healthy then go to the center and then pass to part~\emph{Move}.
      
    \instr Find the first illegal boundary \( x \) in \( I' \).
    Attempt to find a substantial domain within \( \Delta \) steps on the left of \( x \);
    let \( S' \) be the first one.

    \instr  Attempt to find a substantial domain \( S'' \) within \( \Delta \) steps on the right of \( x \).
    Let \( S'' \) be the first one.
    If \( S' \) and \( S'' \) are adjacent (so there is an illegal boundary between them) then
    go to part~\emph{Fail}.
    Else pass control to part~\emph{Stitch}.    
    
\end{itemize}

\item[Stitch]
  Attempt one stitching operation, go to the center and pass control
  to part~\emph{Survey}.

  (Note that an ambiguous interval (between \( S' \) and \( S'' \)) does not trigger stitching unless
  it contains an illegal boundary.)

\item[Move]
Repeat:
\begin{itemize}
  \instr If you are at an illegal boundary go to part~\emph{Survey}.
  \instr Else if you are not in a D-zone then make a step towards it.
  \instr Else go to part~\emph{Adjust}.
  \end{itemize}

\item[Adjust] You come here only if you are in a D-zone.
  Let us call the \df{target position} the point inside the D-zone where the
  value of \( \BigDigression \) changes, or it center, if it is not changing.  
  
  Repeat:
  \begin{itemize}
    \instr Survey \( \Z \) cells ahead and \( \Z \) cells behind the center.
    \instr If you find an illegal boundary jump to part~\emph{Survey}.
    \instr  Else if the D-zone has \( \Z/2 \) cells then:
    if there are rebuilding marks (for example belonging to a rebuilding base), remove one;
    else go to the target position and finish.
    \instr  If the D-zone has \( \Z/2\pm 3\Delta \) cells then
  make a step bringing its size closer to \( \Z/2 \) cells; else go to part~\emph{Fail}.
  \end{itemize}
  
\item[Fail]
  \begin{itemize}
  \instr Survey \( \Z \) cells ahead and \( \Z \) cells behind the center.
  \instr If you find a rebuilding base then start rebuilding, from a center in its middle.
  \instr Else contribute one cell to a rebuilding base (if there is none, start at an illegal boundary)
  then restart at part~\emph{Survey}.
  \end{itemize}
  
\end{description}

\subsection{Rebuilding}\label{sec:rebuilding}

Just as with healing, we will not mention disorder in describing the rebuilding procedure
(since the program does not see it)---but the analysis will take disorder into account.
Rebuilding will normally start from a rebuild base as in Definition~\ref{def:rebuild-base}
(its middle point will be taken as the center \( \z \) of rebuilding).
This makes sure that, if rebuilding gets triggered by a burst in normal mode, zigging will
notice this and call alarm.
Rebuilding could also start from big turn starvation (see Section~\ref{sec:feathering}).
In order to find the center, a track
\begin{align*}
   \Rebuild.\Half\in\{*,-1,1\}
\end{align*}
will show whether we are in the left or the right half of the rebuilding area.
The center is the place separating the two halves.
We will use the following notion.

\begin{definition}\label{def:valid-colony}
  Suppose that in a configuration \( \xi \), there is an interval \( I \) 
  ending in substantial homogenous domains, with at
  most 3 ambiguous intervals inside, in which it can be changed
  to become healthy, with a colony \( C \) in \( \Int(I,0.2\Q\B) \).
  Then we will say that \( C \) is a \df{valid} colony of \( \xi \) with a \df{neighborhood} \( I \).
\end{definition}

The goal is that
\begin{varenum}{r}
\item\label{i:rebuild.size} we end up with a new decodable area
  extending at least one colony to the left and one colony to the right of \( \z \).
\item\label{i:rebuild.keep-healthy} the process does not destroy any valid colony.
\end{varenum}

Rebuilding happens in a bounded number of sweeps. 
The number of the current sweep is shown on a track \( \Rebuild.\Sweep \).
Zigging is done similarly to the ordinary simulation,
but now a frontier zone of \( \Z/2 \) cells similar to the one defined in Section~\ref{sec:zigging} will
be marked explicitly in the track \( \Rebuild.\Addr \) of the frontier zone, which is undefined outside it.
It will show the (positive or negative) distance of the frontier from the center.
Since the center is the place separating the negative and positive values of \( \Rebuild.\Half \),
returning to it does not need to rely on the distance as shown in the
value of \( \Rebuild.\Addr \) carried in the rebuild frontier zone.
The actual distance can namely deviate from this value somewhat: it can grow
if repairs insert new cells.

In every zigging pass, as the frontier gets advanced, the head shifts the frontier zone
and changes \( \Rebuild.\Addr \) accordingly (increasing by 1 at right shifts, decreasing it
at left shifts).
In more detail: for example at a right shift, the head first moves right from the center of the
zone by \( \Z \), adding a new element to the zone on the right,
then left by \( 2\Z \) while adding 1 to \( \Rebuild.\Addr \) in the zone,
then deletes leftmost element, then moves back to the center.
If during all this operation some inconsistency is
seen in the tracks \( \Rebuild.\Sweep \), \( \Rebuild.\Half \) or \( \Rebuild.\Addr \),
then the healing procedure for rebuilding is called
(see Section~\ref{sec:rebuild-heal}).

Now there is only a bounded number of big turns, at the end of the sweeps.
They are handled similarly but simpler than in normal mode.
The \( \BigDigression \) field is used;
however, now the frontier will simply move along with the D-zone.
Big turn starvation happens when \( |\Rebuild.\Addr| \) grows larger than \( 5\Q \);
this will trigger restart.

Here are the stages.
Recall the stitching operation from Section~\ref{sec:stitching}.
\begin{description}
 \item[Mark] Extend a rebuilding area  over \( 4\Q \) cells to the right
   and \( 4\Q \) cells to the left from the center \( \z \).
   At start, erase the \( \Rebuild.\Base \) track of the rebuilding base.
   So the distance to the edge of rebuilding from the center is \( \le 8\Q\B \),
   which led us to define \( \CRebuild \) in~\eqref{eq:cns.traj}.

   In a clean area and in the absence of noise,
   rebuilding can be interrupted only in the following ways by the content of the area encountered.
\begin{varenum}{i}
\item\label{i:rebuild.restart}
  A \df{restart event}:
  
  If at the beginning or during leftward marking,
  a zig movement found on the left at least \( \Z/4 \) cells of the
  frontier of a rightward marking stage of another rebuilding process,
  then move back there and start a new rebuilding from there.

  Similarly, if at the beginning, or during leftward marking, a zig movement
finds at least \( \Z/4 \) cells of a rightward frontier zone (of normal mode),
then start a new rebuilding from there (after moving back there in a zigging motion).

(In both cases, at the time of restart, zigging motion still shows at least \( 2\Delta \) cells of
the abandoned rebuilding frontier, so for example a normal mode frontier cannot be
triggered into a rebuilding restart by a burst.)
  
\item  Turn starvation as defined in Section~\ref{sec:feathering}.
  Small turn starvation triggers alarm (thus healing).
  If \( |\Rebuild.\Addr| \) grows large then rebuilding will restart due to big turn starvation.
\end{varenum}

 \item[Survey and Create]
More details of this stage will be given below.
It looks for existing valid colonies, and possibly creates some.
As a result, we will have one colony called \( C_{\Left} \)
on the left of \( z \) along with its neighborhood as in Definition~\ref{def:valid-colony},
one called \( C_{\Right} \) on the right of \( z \), and possibly some colonies between them.
Make all newly created colonies represent stem cells.
Direct all the other colonies with drifts and bridges towards \( C_{\Left} \).
The interval covering \( C_{\Left} \) and \( C_{\Right} \)
will be called the \df{output interval} of rebuilding.
The pair of neighbor colonies with \( C_{\Left} \) on its left will be made the current colony-pair.

\item[Mop] Remove the rebuild marks (address and sweep),
  shrinking the rebuilding area \( R \), starting on its left end,
  onto the right end of \( C_{\Left} \).

\end{description}

Marked cells from some interrupted rebuilding may remain even after the mop-up operation.
(These may trigger new healing-rebuilding when the head meets them sometime later.)

\subsubsection*{Details of the Survey and Create stage}

The complexity of this stage is due mainly to guaranteeing
property~\eqref{i:rebuild.keep-healthy} above.

\begin{varenum}{s}
\item\label{i:stitch}
  Going from left to right,
  pass through the marked area, and stitch every pair of consecutive
  substantial domains separated by an ambiguous area of fewer than \( \Delta \) cells,
  just as during healing.
  But don't create new rebuilding cells as in healing:
  if the stitch result leaves some illegal boundary, just leave it there.

\item\label{i:find-colonies}
  Pass through again, and look for whole colonies.
  Mark the cells belonging to whole colonies as such, and mark all other cells as such.
  The marks should go to a special track \( R_{1} \).

\item\label{i:repeat}
  Repeat steps~\eqref{i:stitch} and~\eqref{i:find-colonies}, writing the resulting marks onto a special
  track \( R_{2} \).
  The following steps will rely on the two tracks \( R_{1},R_{2} \) having identical content.
  If it is discovered that this is not the case, alarm is called.
  This is important since the colony creation operations are destructive, they should not be triggered by
  a single burst.
  
\item\label{i:C-left-right} Check if there is a marked colony at least three quarters
  to the left of the center,
  whose whole neighborhood as in Definition~\ref{def:valid-colony} is healthy.
  If yes, find the closest one.
  If not, create one making sure it does not intersect any marked whole colony, and its
  neighborhood is healthy (if necessary overwrite part of the neighborhood with stem cells).
  Call this colony, found or created, \( C_{\Left} \).
  Proceed similarly in finding or creating a colony \( C_{\Right} \) (disjoint from \( C_{\Left} \))
  at least three quarters to the right of the center.
  
\item Fill in the area between \( C_{\Left} \) and \( C_{\Right} \)
  and the other marked whole colonies between them:
  fill these gaps with adjacent stem cells, creating a new colony every time
  an interval of \( \Q \) adjacent stem cells has been created.

\item Make all newly created colonies represent stem cells.
  Let \( C_{0} \) be the first colony towards the left of the center
  with at least half of it to the left of the center.
  Direct all drifts to the left end of \( C_{0} \).
  Make \( C_{0} \) and its right neighbor the new current colony-pair (create a bridge from
  \( C_{0} \) to its right neighbor if needed), and let them
  represent the start of a healing process on the level of \( M^{*} \).
\end{varenum}

\begin{remarks}\label{rem:rebuild-precedence}
  \begin{enumerate}
  \item
    Once rebuilding completes, since the state of the current colony-pair simulates the start of
    healing in \( M^{*} \), the head will continue to the right.
    There, rebuilding may be called again, but it will not rewrite the colony \( C_{\Left} \).
    It may rewrite its right neighbor colony, but not destroy it; so repeated calls to rebuilding
    will result in progress.
    
  \item The precedence given to rebuilding from left to right
    will be used in the proof of Lemma~\ref{lem:pass-clean}.
    In the absence of new noise, after a constant number of passes,
    a certain clean interval \( J \) will extend to the right until it reaches
the end of a colony pair or of a rebuilding area.
If marking was unstoppable in both directions then this might not happen soon, since after disorder
is entered and exited, no assumption can be made of the state of the current cell-pair.
Reaching the left end of  \( J \)
a new rebuilding process may send the head back to the right end, from which a new rebuilding
process can send it back to the left end, and so on.
But because marking to the right has precedence, a new left-directed
marking started by disorder on the right end would not stop it.

\item Giving precedence to rightward rebuilding has the drawback that one can design a initial configuration
  in which even in the absence of noise, higher-level structure will never arise even locally.
  Namely, we can fill the line with short intervals \( \dots,J_{-1},J_{0},J_{1},\dots \)
  each of which is the start (say of size \( 3\Z \)) of a
  right-directed marking process.
  Then the head, after moving left on \( J_{0} \), will be captured by the process on \( J_{-1} \), then
  later captured by the similar process on \( J_{-2} \), and so on.
  But we will not need to consider such pathological configurations.

  
  \end{enumerate}
\end{remarks}

\subsection{Healing during rebuilding}\label{sec:rebuild-heal}

Some healing may be needed even within the rebuilding process in case of a new burst,
(as shown in Example~\ref{xpl:heal-in-rebuild}).
We will call it \( \rRebuildHeal \).
We can define a notion of rebuild health similarly to health, using the tuple
\begin{align*}
 \Rebuild.\Core = (\Rebuild.\Sweep, \Rebuild.\Addr, \Rebuild.\Half,\BigDigression).
\end{align*}
During the marking stage in rebuilding, the area outside what is already marked is not subject
to any requirements on tracks other than \( \BigDigression \).
(The circumstances mentioned in part~\ref{i:rebuild.restart} above can trigger a restart event, but not
a healing process.)

We will say that the configuration is \df{pre-healthy for rebuilding} if it
satisfies all boundary requirements, and \df{healthy for rebuilding}
if also the rebuild frontier zone and the digression zone have \( \Z/2\pm 1 \) cells.
The analogue of Lemma~\ref{lem:boundaries-health}
and Corollary~\ref{crl:health-extension} will hold for rebuild pre-health.
As there, checking for the actual rebuild health, zigging is needed.
The analogue of Corollary~\ref{crl:health-extension} for rebuild health holds:

\begin{lemma}\label{lem:rebuild-health-extension}
Let \( \xi \) be a tape configuration that is healthy for rebuilding on intervals \( A_{1}, A_{2} \) 
where \( A_{1}\cap A_{2} \) contains at least \( \Z/2 \) cells.
Then \( \xi \) is also healthy for rebuilding on \( A_{1}\cup A_{2} \).  
\end{lemma}

Healing for rebuilding is very similar to healing, trying to repair the health of the
rebuilding process.
In its Survey  and Stitch stages, it will repeatedly check, similarly to \( \rHeal \),
whether it is possible to recreate pre-health
for rebuilding by stitching up to three ambiguous areas.
(Stitching is simpler now as there are no rigid colony boundaries to consider.)
If yes, it performs one stitching step, otherwise it adds a cell to a new rebuilding base.
If stitching is successful, it goes to the Move stage to gradually move back to the D-zone.
Once there, it goes to the Adjust stage to remove possible rebuilding base marks,
and to correct the size of the zone before finishing.

\section{Scale-up, isolated bursts}

This section first shows how health is restored in the absence of disorder and noise.
Then it defines the code \( \Phi \) mapping a history \( (\eta,\Noise) \) of a machine \( M \) into a
history \( (\eta^{*},\Noise^{*}) \) of the simulated machine \( M^{*} \).
For this, it introduces the notion islands in the framework of an ``annotated'' history.
Finally, using the introduced terminology,
it shows that the healing procedure indeed deals with isolated bursts.
For the elimination of disorder created by faults
we will rely on the Escape, Spill Bound and the Attack Cleaning
properties of a trajectory in Definition~\ref{def:traj}.

\subsection{Restoring health in the clean, noiseless case}

An interval rewritten by noise can have \( \Pass\ne 0 \) everywhere or many footprints of
a big turn too close to each other even if it is clean, so we define
a property of intervals avoiding this.

\begin{definition}\label{def:safe-for-turns}
  An interval will be called \df{safe for small left turns}
  if it has no more than \( 3\Delta \) consecutive cells with \( \Pass=1 \).
  And if it has more than \( \Delta \) then it is preceded by a footprint of a big left turn.
  It is \df{safe for big left turns} if it has no sequence of more than \( 3\F\log\Q \)
  footprints of a big left turn closer than \( 2\F \) cells to each other.
  And if it has more than \( 2\F\log\Q \), then on the left of this sequence there is a
  colony encoding a big cell with \( \Pass=1 \).
  It is \df{safe for left turns} if it is safe for both small and big left turns.
  It is \df{safe for turns} if it is safe for both left and right turns.

  We say that the interval is \df{weakly} safe for turns if the upper bounds \( 3\Delta \) and
  \( 3\F\log\Q \) are replaced with \( 6\Delta \) and \( 6\F\log\Q \).
\end{definition}

Note that in order to be safe for turns, the interval only has to be safe for left turns on
its part to the right of the head and safe for right turns on the part to the left.

\begin{lemma}\label{lem:safe-for-turns}
  If a clean interval is passed over from left to right by a path \( P \) having at most one burst,
  then it becomes safe for turns.
  If it is passed from right to left then it becomes safe for small turns; but if this pass happens
  after a left-to-right pass then it also stays safe for all turns.
\end{lemma}
\begin{proof}
  We will present the proof for a pass from left to right, and point out the only difference for the case
  when the pass is from right to left, in part~\ref{i:safe-for-turns.right-to-left} below.

  We will consider the case when no burst occurs.
  The path with a possible burst can be divided into three parts \( P_{1},P_{2},P_{3} \):
  before the burst, the part when the
  disorder created by the burst is crossed over, possibly several times,
  and the part when the path leaves this place behind definitively.
  Parts \( P_{1},P_{3} \) are handled below.
  Due to attack cleaning, \( P_{2} \) can cross the disorder at most \( \beta \) times, these can
  modify the estimates in parts \( P_{1},P_{3} \) only by \( \beta \) cells.
  
  Let \( x_{1}<x_{2}<\dots<x_{n} \) be the points of the interval left behind with \( \Pass=1 \), and let
  \( t_{i} \) be the times when this happens at \( x_{i} \).
  \begin{enumerate}
  \item Consider the space-time points \( (x_{i},t_{i}) \)
    where the head makes a big turn in rebuilding mode.
    In rebuilding mode, big turns happen at the end of
    sweeps---there is only a constant number of them---so
    before time \( t_{i} \), the head must have performed at least one complete sweep of rebuilding.
    If this rebuilding succeeds then it leaves a healthy area containing at least one colony on the left and one
    on the right of its center.
    Another rebuilding can only start at the left or right of this interval.
    It cannot be on the left, since \( t_{i} \) was the last time when \( x_{i} \) was passed.
    So a next rebuilding big turn can only happen about \( \Q\B \) cells to the right of \( x_{i} \).

    If the pass is from left to right then 
    this rebuilding can only fail by big left turn starvation, see Section~\ref{sec:feathering}.
    This must happen at a distance at least \( 2\Q\B \) to the intended left turn.
    If any later rebuilding has its left end within \( \Q\B \) of \( x_{i} \)
    then it will already succeed, and we can reason as above.

    \item \label{i:safe-for-turns.right-to-left}
    In case the pass is from right to left then another possible way that the rebuilding can fail
    is when the leftward rebuilding is overridden by a new rightward rebuilding, see the Marking part of the
    rebuild procedure in Section~\ref{sec:rebuilding}.
    As a result, the right-to-left turns may get close to each other,
    so safety for big turns will not be guaranteed.
    However, when the interval was previously passed from left to right then, as seen above,
    interrupted rebuildings can only happen because of turn starvation.
    As seen above,
    they are placed at a distance \( >2\Q\B \) from each other, and big turns will remain safe.
 
  \item Consider now the space-time points \( (x_{i},t_{i}) \)
    where the head makes a big turn in normal mode.
    Suppose first that the simulation program may be interrupted by healings, but these
    healings all succeed, so rebuilding is not triggered.
    We only need to consider big turns made in the turn region at a left end of a colony \( C \)
    which is the left element of a colony-pair, during the starting and ending sweeps:
    the footprints inside will be overwritten by the ending sweeps.
    If the work period finishes normally, and a healthy colony \( C \)
    remains to the right of \( x_{i} \), then 
    big turns in normal mode not belonging to this work period
    will only be made at least \( \approx \Q\B \) to the right.

    In normal mode, \( C \) may be replaced with another colony \( C' \), overlapping it.
    Then \( C' \) will already not be replaced (without going to another colony on its left),
    so the big right turns on its left are necessarily separated on the
    right from others by at least \( \approx\Q\B \), but the turns on the left of \( C' \)
    can be close to the previous ones on the left of \( C \).
    
    It is also possible that rebuilding will be called before the work period over \( C \) finishes.
    (This can only happen if then this rebuilding experiences big left turn frustration on its right,
    since otherwise it would sweep over \( x_{i} \).)
    Such a big left turns in normal mode at the left end of a rebuilt
    colony \( C' \) near \( x_{i} \) can occur only once, since rebuilding
    created a colony-pair \( (C',C'') \).
    
    \item Consider the points \( (x_{i},t_{i}) \) where the head turns during healing.
  If \( i<n \) this healing cannot fail, since then the subsequent rebuilding would bring the head to the left
  of \( x_{i} \), contradicting the assumption that \( t_{i} \) was the last time when it was there.
  It could, though, experience small turn starvation.
  But if healing eventually restarts near \( x_{i} \) then it will not experience small turn
  starvation again and succeeds, so any following turn points \( (x_{j},t_{j}) \) due to healing will be
  at least \( \Delta \) away.

  The analysis is similar for points \( (x_{i},t_{i}) \) where the head turns in the part of rebuilding when it is
  attempting to stitch an ambiguous area.

\item What remains is space-time points \( (x_{i},t_{i}) \)
  where the head makes a small turn in normal or rebuilding mode.
  In these modes the head makes a zig only in every second step, so normally these places also don't occur
  consecutively; the violations of this may happen only during healing and are limited as discussed above.
    \end{enumerate}
\end{proof}

\begin{lemma}[Combined heals]\label{lem:combined-heals}
 Let \( d = 3\Z\B \).
 Assume that the head moves in a noise-free and clean space-time rectangle
  \( J\times K \) with with \( 2 d < |J| \), \( |K|\ge \Tu^{*} \),
  touching every cell of \( J \) at least once, and never in rebuilding mode.
  Assume also that at the beginning, \( J \) is safe for small turns.
  Then during \( K \), the area \( \Int(J, d) \) becomes healthy.
\end{lemma}
\begin{proof}
  Since \( J \) is safe for small turns, in what follows we will not see small turn starvation.
  Since big turn starvation would trigger rebuilding and we don't see rebuilding,
  we will not encounter big turn starvation either.
  
  If healing was not called then after the head touched every cell of \( J \) the
  health of the area is proved.
  If the head entered \( J \) during healing then before finishing healing,
  it touches over an area of size \( \le 2\Z\B \).
  The upper bound \( 3\Delta \) on the number of consecutive cells with \( \Pass\ne 0 \)
  makes sure that the turns happen within \( 3\Delta \) cells of both ends of this area, increasing
  its size to at most \( d = 3\Z\B \).
  
  Consider some time when healing is started while the head is in \( J \),
  and let \( I_{1} \) be the interval \( I \) defined in the healing
  procedure for this point.
  Since rebuilding is not started, the stitching part of healing succeeds, with the interval \( I'_{1} \)
  becoming pre-healthy, while staying in an area of size \( d \).
  After this, the Move part of healing is trying to move the head towards a D-zone.
  New healing starts only at some illegal boundary \( x \) outside \( I'_{1} \),
  and the interval \( A \) containing both \( x \) and \( I'_{1} \) is pre-healthy at this time.
  If the head does not leave \( J \) during this procedure (which is the case if \( x\in\Int(J,d) \)),
  this stitching also succeeds, creating a new pre-healthy interval \( I'_{2} \) that
  intersects \( A \) in an interval of size \( \ge\Delta\B \).
  Hence \( A\gets A\cup I'_{2} \) becomes pre-healthy.

  This process continues until the head arrives at a D-zone, when it starts
  the Adjust part of healing in a pre-healthy interval \( A \);
  this can be interrupted by an illegal boundary again,
  triggering new stitching, but since rebuilding is not started, 
  it eventually succeeds, creating a healthy interval \( A \) that
  contains all areas surveyed until now.
  Now normal mode resumes, moving the D-zone and eventually moving the front,
  until a new illegal boundary is found.
  This way the healthy interval in which the head is moving is getting extended,
  and the only parts of \( J \) not becoming healthy are confined to the borders
  of size \( d \).

  The healing procedure can only create at most \( \Delta \) consecutive cells with
  \( \Pass\ne 0 \).
  Zigging occurs only every 2 steps, so it does not create solid intervals with \( \Pass\ne 0 \).

  The number of steps in normal mode is at most as much as our bound on the number of
  steps in a work period.
  This changes by at most a factor of 2 due to the delays in healings, so the bound \( |K|\le\Tu^{*} \)
  is sufficient.
\end{proof}

\subsection{Annotation, scale-up}\label{sec:annotation}

Let us define the notion of ``almost healthy'' (admissible) for histories.

Informally, an admissible configuration may differ from a healthy one in a small number
of intervals we will call ``islands''.
Even a healthy configuration may contain some intervals called ``stains'':
places in which the \( \Info \) track differs from a codeword.
These pose no obstacle to the simulation, and if they are small and few then
will be eliminated by it, via the error-correcting code.

Annotation will interpret parts of a history, ``covering up''
small segments that are not quite healthy.
It will leave other parts uninterpreted.


\begin{definition}[Annotation]\label{def:annotation}
  Recall the definition of \( \beta' \) and \( \CStain \) in~\eqref{eq:stain}.
  Let
  \begin{align*}
   \Q' = \Q-2\PadLen
  \end{align*}
  be the number of cells in the interior of a colony.
  Then any healthy interval of size \( \Q'\B \) intersects with at most one turn region at the end
  of some colony.
  
  A \df{annotation} for a history \( (\eta,\Noise) \) over a time interval \( \clint{0}{u} \) is a tuple
  \begin{align*}
    (\Rg, \chi, \cI,\cS)
  \end{align*}
  with the following properties, for a certain constant
  \begin{align*}
   \CRelief.
 \end{align*}
  \begin{varenum}[series=annotated]{a}

  \item \( \Rg \) is a subset of \( \bbZ\times\clint{0}{t} \) called the \df{range} of the annotation;
    denote \( \Rg(t)=\setOf{x}{(x,t)\in \Rg} \).
    
  \item \( \chi \) is a history over \( \Rg \),
    there are no bursts of \( \Noise^{*} \) over \( \Rg \),
    and \( \chi(\cdot,t) \) is healthy over \( \Rg(t) \). 

  \item\label{i:annotated.islands}
    The disorder in \( \Rg \) is covered by a set \( \cI \) of connected space-time
    regions called \df{islands}.
    Elements of the set \( \cS \) are connected space-time regions in \( J \) called \df{stains}
    whose space projection has size \( \le\CStain\B \).
    Each island is contained in a stain.
    We will write \( \cI(t) \) for the set of islands at time \( t \); in other words, 
    \( L \) is in \( \cI(t) \) iff 
    \( L\times\{t\} \) belongs to \( \setOf{K\cap(\bbZ\times\{t\})}{K\in\cI} \).
    Similarly for stains.

    At any time, \( \cI(t) \) and \( \cS(t) \) consist of intervals.


    \( \eta \) differs from \( \chi \) over \( \Rg \) only in the islands.

  \item\label{i:annotated.island-escape}
    In any one island, the head does not spend more total time (even when
    entering and exiting possibly several times) than
    \( \CRelief(\beta^{2}/\gamma)\passno\escno\Tu \).

  \item\label{i:annotated.num-islands} At any time \( t \), any segment of size \( \Q\B \)
    contains at most 3 islands, and at most 1 can be in the interior of a colony.
    At most 2 stains can be in the interior of a colony.

  \item If at some time no cell of a colony of \( \chi \) belongs to the update phase then
the \( \Info \) track of the interior can be changed in the stains in such a way that it becomes
a codeword of the code \( \upsilon \) as in Definition~\ref{def:colony-interior}.

  \item\label{i:annotated.turned}
    Every interval of size \( \Q'\B \) is weakly safe for turns.
\end{varenum}



  



In an annotated history, we will say that the head is \df{free} if it is at a distance of at least
  \( \Z \) steps from any island.
  \end{definition}

A configuration may allow several possible annotations;
however, since the code defined in Section~\ref{sec:simulation-phase}
is \( (\CStain\beta,2) \)-error-correcting, 
the codewords recoverable from it do not depend on the choice of the annotation.

Formally, the proof of error-correction would proceed by proving that
annotation can essentially be extended forward in time; however, we will retain an informal language
whenever it is clear how to translate it to annotation.
Example~\ref{xpl:3-islands} shows how three islands may arise, along with
an informal argument that local correction
does not have to deal with more than three islands in any area of size \( \Q'\B \). 

Let us now define formally the codes \( \varphi_{*k},\Phi_{k}^{*} \) needed
for the simulation of history \( (\eta^{k+1},\Noise^{(k+1)}) \) by history \( (\eta^{k},\Noise^{(k)}) \).
Omitting the index \( k \) we will write \( \varphi_{*},\Phi^{*} \).
To compute the configuration encoding \( \varphi_{*} \) we proceed first as
done in Section~\ref{sec:hier-codes}, using the code \( \psi_{*} \) there,
and then initialize the kind, sweep, drift and address fields appropriately.
The value \( \Noise^{*} \) is obtained by a residue operation
as in Definition~\ref{def:sparsity}; it remains to define \( \eta^{*} \).
In the parts of the history that can be locally annotated, and which we will call \df{clean},
if no colony has its starting point at \( x \) at time \( t \), set \( \eta^{*}(x,t)=\Vacant \).
Otherwise \( \eta^{*}(x,t) \) will be decoded from
the \( \Info \) track of this colony, in its work period containing time \( t \).
More precisely:

\begin{definition}[Scale-up]\label{def:scale-up}
Let \( (\eta,\Noise) \) be a history of \( M \).
We define \( (\eta^{*},\Noise^{*}) \) \( =\Phi^{*}(\eta,\Noise) \) as follows.
Consider position \( x \) at time \( t \), let \( I=\lint{x-\Q\B}{x+2\Q\B} \),
\( J =\rint{t-\Tus}{t} \).
If \( \eta(\cdot,t) \) cannot be annotated in \( I\times J \)
then \( \eta^{*}(x,t)=\Bad^{*} \).
If \( x \) is not the start of some colony \( C \) in this annotation
then let \( \eta^{*}(x,t)=\Vacant \); assume now that it is.
Then let \( t' \) be the last time when the head is not in an island and its age
is not in the update phase, and let
\( \eta^{*}(x,t) \) be the value decoded from \( \eta(C,t') \).
In more detail, as said at the end of Section~\ref{sec:coding}, we apply the decoding
\( \psi^{*} \) to the interior of \( C \) to obtain \( \eta(x,t) \).
\end{definition}

\subsection{Dealing with isolated bursts}\label{sec:1-level-noise}

Definition~\ref{def:scale-up} decodes trajectories \( (\eta,\Noise) \) into
histories \( (\eta^{*},\Noise^{*}) \).
We don't know yet whether trajectories of \( M \) are decoded into trajectories
of \( M^{*} \).
Let us give an informal argument first.

Isolated bursts don't create disorder larger than \( \beta \).
The head escapes a disorder interval \( I \) via the Escape property; while it is inside, the
spreading of this interval is limited by the Spill Bound property.
Every subsequent time when the head enters and exits \( I \) this gets decreased
via the Attack Cleaning property, so it disappears after \( O(\beta) \) such interactions---see
Lemma~\ref{lem:healing} below.

Let us first show that annotation implies that the job of simulation proceeds as required.

\begin{lemma}\label{lem:transition}
  Consider an annotation \( (\Rg,\chi, \cI,\cS) \) of a trajectory \( (\eta,\Noise) \)
  over a rectangle in \( \Rg'\subseteq\Rg \)
  which the head never leaves.
  \begin{alphenum}
    \item The decoded history \( (\eta^{*},\Noise^{*}) \) over \( \Rg' \)
      satisfies the Transition Function property of trajectories (Definition~\ref{def:traj}).
    \item\label{i:transition.no-burst} Assume that, in addition, no bursts occur over \( \Rg' \),
      let \( t \) be the end of a complete work period of some colony-pair, and
      \( C=\lint{x}{x+\Q\B} \) one of its colonies.
      Then 
\begin{align*}
 \xi=\eta(\cdot,t)\on C =\phi_{*}(\eta^{*}(x,t)) .
\end{align*}
      In other words, the colony \( C  \) will have no stains at time \( t \); each of its cells has
      the value assigned to it by the code \( \phi_{*} \).
  \end{alphenum}
\end{lemma}
\begin{proof}
  Consider a sequence of configuration in \( \Rg' \) starting from one corresponding to the
  beginning of a working period of a colony-pair, at a time when the head is free. 
  Properties~\eqref{i:annotated.islands} and~\eqref{i:annotated.island-escape} limit
  the space occupied by each island as well as the time that a head can spend in it.
  While the head is free it is carrying out the simulation program.
  Property~\eqref{i:annotated.turned} makes sure that no turn starvation slows it down.
  The error-correcting decoding will recover the state of the simulated cell from each colony
  since it needs to recover from at most two stains.

  Any island in the interior of the colonies will be eliminated during the initial
  sweeps of the work period, as this is the only way the head can be freed from it.
  The computation part of the simulation is repeated three times.
  A new burst can occur in at most one of these repetitions; if it appeared in the
  \( i \)th repetition, the result on the \( \Hold_{i} \) track may be worthless.
  But the two other \( \Hold_{j} \) tracks will contain the correct results, and the majority
  voting will recover it.
  If no new bursts occur then the majority voting will recover everywhere the correct result,
  showing~\eqref{i:transition.no-burst}.
  
  The length of the work period is upper-bounded by \( \Tu \) times the
  number \( O(\Q\F\Z^{2}) \) of computation
  steps (see Section~\ref{sec:length-work-period})
  while the head is free plus the time while the head was not free.
  Given that the total time spent in any one island is \( \le\CRelief \beta^{3}\escno\Tu \)
  and there are at most 3 islands per colony, the total time in islands
  is \( O((\beta^{2}/\gamma)\passno\escno\Tu ) \): adding this to \( O(\Q\F\Z^{2}\Tu) \),
  we still stay below the upper bound \( \Tu^{*} \).

\end{proof}

Let us proceed to proving that in the absence of \( \Noise^{*} \), annotation can be extended
on  the range \( \Rg \).
More precisely, we will introduce a game.

\begin{definition}\label{def:annotation-game}[Annotation game]
  The game is played over a trajectory \( (\eta,\Noise) \) of a machine \( M \).
  Its two players are Range Extender and Annotator.
  Range Extender is in charge of a sequence of times \( 0 = t_{0} < t_{1} < t_{2}<\dots \),
  and in extending the range.
  At any stage of the game, it assumed that the annotation is defined up to time \( t_{i} \),
  that is part of the range \( \Rg\subseteq\bbZ\times\lint{0}{t_{i}} \) is defined as well as
  the annotation over it.
  Now the Range Extender extends the range over the time interval \( \lint{t_{i}}{t_{i+1}} \)
  in two possible ways.
  \begin{djenum}
  \item It deletes an interval \( I \) containing the head from \( \Rg \).
    (We still talk about ``extending'' since the range is extended in time.)
    More precisely, it defines
    \begin{align*}
   \Rg\gets \Rg\cup(\Rg(t_{i}-1)\setminus I)\times\lint{t_{i}}{t_{i+1}}.
    \end{align*}
    It also trims all islands and stains, deleting their parts in \( I \), and keeps
    \( \cI(t) \), \( \cS(t) \) constant over \( \lint{t_{i}}{t_{i+1}} \).

  \item It adds an interval \( I \) containing the head to \( \Rg \).
    More precisely, it defines
    \begin{align*}
   \Rg\gets \Rg\cup(\Rg(t_{i}-1)\cup I)\times\lint{t_{i}}{t_{i+1}},
    \end{align*}
    where \( I \) has the property that by defining \( \chi(x,t_{i})=\eta(x,t_{i}) \) over \( I \)
    and setting \( \cI(t_{i})=\cI(t_{i}-1) \), \( \cS(t_{i})=\cS(t_{i-1}) \) we still get an annotation.
    Also \( \eta \) is safe for turns over \( \Rg(t_{i}) \).
  \end{djenum}
  In both cases, the Range Extender can extend the range only in such a way that
  \( \Noise^{*} \) remains empty over the new range.

  The Annotator must respond by extending the definition of \( \chi,\cI,\cS \) over the
  new range, that is up to \( t_{i+1} \), in such a way that the resulting structure is still an
  annotation.
\end{definition}

Two remarks:
\begin{itemize}
\item The players are ``clairvoyant'': they see the whole trajectory
\( (\eta,\Noise) \) ahead as well, not just the parts that has been annotated already.
\item The players should not be considered adversaries: rather, they cooperate in creating
  the annotation of the whole trajectory.
\end{itemize}

Below, we will show that under the conditions, Annotator can always respond; this way,
the effect of sparse bursts will be corrected, leading to the Transition Function property
of the simulated trajectory \( (\eta^{*},\Noise^{*}) \).
In later sections we will show the other properties, and also how they can lead to appropriate
choices of the Range Extender player.
Note in particular that in this game, as Lemma~\ref{lem:no-new-island} shows,
every island is the result of a burst that occurred during \( \Rg \).
 
In a clean configuration, whenever healing started with an alarm, the procedure
will be brought to its conclusion as long as no new fault occurs.
However, every time the head emerges from disorder, we
cannot assume anything about the state of the cell-pair to which it arrives.
This complicates the reasoning, having to consider several restarts of the healing procedure.
By design, this procedure can change the \( \Core \) track only in one cell.
The following lemma limits even this kind of possible damage.

\begin{lemma}\label{lem:no-new-island}
In the absence of noise, no new island will arise.
\end{lemma}
\begin{proof}
  The islands are defined only by the tracks in the \( \Core \) group.  
In normal mode, these tracks change only at a boundary point.

The healing procedure changes the \( \Core \) as part of a stitching operation,
or removing or adding a rebuild mark.
The proof of Lemma~\ref{lem:stitching} shows that 
inside a healthy area, healing can only change the \( \Core \) track by
moving the boundaries around.
These operations don't affect pre-health; they may temporarily change the
size of the frontier zone or the D-zone by up to \( 3\Delta \), but then the
adjust stage restores these sizes.
\end{proof}

The following lemmas are central to the analysis under the condition that bursts are isolated.
Let us first show how the Escape and Pass Cleaning property of trajectories
helps bringing the head into clean intervals, even if these are rather small.

\begin{lemma}\label{lem:create-holes}
 Let \( \gamma' = \gamma/2 \).
 Consider a noise-free path \( P \) starting at some time \( t_{0} \) at a point \( b_{0} \).
Create points \( b_{j}=b_{0}+j\gamma'\B \) for (positive and negative) values of \( j \).
Let us call the intervals \( \lint{b_{j}}{b_{j+1}} \) \df{blocks}.
Assume that the disorder that the path may encounter is covered by \( n \) blocks.
Then there is a sequence of times \( t_{0}<t_{1}<\dots \) 
with \( t_{i+1}-t_{i}\le\escno\Tu \), such that during the time 
intervals \( \rint{t_{i}}{t_{i+1}} \) called \df{skips} the head passes over some block \( H_{i} \)
either leftwards or rightwards, further
except for \( 2 n\passno \) skips, the interval
\( H'_{i}=\Int(H_{i},(\CMarg+\CSpill)\B) \) is clean.
Consequently, the first such skip happens for some 
time \( t_{i}\le (2 n\passno + 1)\escno\Tu \).
\end{lemma}
By~\eqref{eq:beta-lb}, \( |H'_{i}|>0 \).
\begin{proof}
  Suppose that time \( t_{i} \) has been defined and the head is at point \( b_{j} \) at this time.
  The interval \( \lint{b_{j-1}}{b_{j+1}} \) has length \( 2\gamma'\B=\gamma\B \), so 
  by the Escape property of trajectories and~\eqref{eq:beta-lb}, the head
  will escape it within time \( \escno\Tu \).
  Let \( t_{i+1} \) be the time when it arrives at \( b_{j-1} \) or at \( b_{j+1} \).
  
  We claim that the number of skips \( i \) in which \( H'_{i} \)
  is not clean is at most \(  2 n\passno \).
  Indeed, the total number of possible blocks containing disorder originally is \( \le n \).
  If \( H = H_{i} \) for a right skip \( i \) for \( \passno \) times,
  then the Pass Cleaning property implies that \( \Int(H,\CMarg\B) \)
  becomes clean; then by the Spill Bound property, \( H' \)  stays clean.
  Similarly for left skips, so \( H = H_{i} \) without this for at most \( 2\passno \) values of \( i \).
\end{proof}

\begin{lemma}[Healing]\label{lem:healing}
  In the annotation game, Annotator can always respond.

\end{lemma}
\begin{Proof}
  A space-time point is a \df{distress event} if either a
  fault occurs there or the head steps onto an island.
The extension of annotation is straightforward as long as no distress event is encountered.
If after a distress event the head becomes free (according to Definition~\ref{def:annotation}),
then we will say that \df{relief} occurred.
Let us see what can occur between a distress and relief event.

\begin{step+}{step:heal.disorder-leave}
  Every time that the head enters disorder, it will leave within time
  \(  21(\beta/\gamma)\escno\passno\Tu \).
\end{step+}
\begin{pproof}
  The disorder is covered by 3 intervals of size \( \beta' \),
  so we can set in Lemma~\ref{lem:create-holes}, using~\eqref{eq:stain} and~\eqref{eq:beta-lb}:
\begin{align*}
  n= 3(\beta'/\gamma' + 2) 
  \le 10\beta/\gamma.
\end{align*}
According to the lemma, the head will leave within time
\( (2 n\passno + 1)\escno\Tu \).
\end{pproof} 

\begin{step+}{step:heal.disorder-visits}
  The head can leave disorder in any one island at most \( \beta'+1 \) times.
  Hence in any interval of size \( \Q'\B \), if there are \( \le 3 \) islands
  then the head can leave the disorder at most \( 3(\beta'+1)\le 4\beta' \) times.
\end{step+}
\begin{pproof}
The size of the smallest interval covering disorder in any island is at most 
\begin{align*}
 d=\beta'\B .
\end{align*}
By the Attack Cleaning property, every time the head leaves it on a side of the disorder
interval \( D \) where it had entered before, it shrinks \( D \) by \( \B \).
So every such visit does this, except possibly the one time
when it exited without entering before.
\end{pproof}

\begin{step+}{step:heal.island-size}
  The size of stains will never grow beyond  \( \CStain\B \).
Also the total change of the front and the \( \Age \) variable between distress and
relief events connected with one island is not more than \( 3\Delta \).
\end{step+}
\begin{pproof}
  The disorder created by a burst can always be covered by an island identical to it.
   We can always define islands to end at illegal boundaries.  
  Suppose that the boundary cell-pair is clean.
  If the head entered it in normal mode it would start healing.
  If it entered as part of a healing procedure started in the island,
  then this procedure would not change the boundary in the direction
  of increasing the island.
  Rebuilding would only start in the island by a disorder.
These rebuilding cells are counted as part of the island, so by this bound on their number,
a rebuilding base will not be created, and
rebuilding will not be started by the healing procedure.
  
  So an island can increase only as a consequence of the head leaving disorder.
  It increases by at most \( \B \) in each such stage, and by part~\ref{step:heal.disorder-visits}
  there are at most \( \beta'+1 \) such stages.
  We started from an island size \( \le\beta'\B \), therefore an island never grows larger than
  \( (2\beta'+1)\B \).


Since between distress and relief the \( \Age \) and address of the front would only change
at the boundaries of the substantial domains, the estimate on the size of change on them follows.
\end{pproof}

\begin{step+}{step:heal.stages}
  Between two visits to disorder, if the head does not become free, it can spend
  only \( O(\beta\Z) \) steps.
\end{step+} 
\begin{pproof}
  Indeed, there can only be \( O(\Delta)=O(\beta) \) steps involving a (Survey) followed by a (Stitch).
  Each (Move) steps takes one closer to the D-zone, and at the start of healing, the D-zone
  was within a distance of \( (3/4)\Z \) cells.
  The number of calls to (Adjust) is at most \( 4\Delta \), each is a sweep of \( \le 4\Z \).  
  The same counting holds for (Fail).
\end{pproof}
  
\begin{step+}{step:heal.total-time}
  The total time between distress and relief is \( \le 64(\beta^{2}/\gamma)\passno\escno\Tu \).      
  Outside distress, the head is at the front in normal mode.
 \end{step+}
 \begin{pproof}
   Initially there are no islands.
By the Spill Bound property, the disorder 
can grow to at most \( 3\beta'\B \).

By~\ref{step:heal.disorder-visits}, between a distress and relief,
there are at most \( \beta'+1 \) visits to disorder.

By~\ref{step:heal.disorder-leave}, every time the head enters disorder,
it leaves within time \(  21(\beta/\gamma)\passno\escno\Tu \).

By~\ref{step:heal.stages} 
the time between two such visits to the same disorder is \( O(\beta\Z\Tu) \).
So the total time between visits is \( O(\beta^{2}\Z\Tu) \), which is
dominated by our bound
\(  3\beta'\cdot 21(\beta/\gamma)\passno\escno\Tu \) on the total time
the head needs to leave disorder.
\end{pproof}

\begin{step+}{step:heal.transition.conditional}
  By Lemma~\ref{lem:transition}
  if a rectangle \( R' \) was annotated until time \( t \) when the head is at the beginning
  of a work period of some colony-pair and then the annotation is extended to the
  end of the work period, the resulting configuration
  will be according to the transition function of the simulated machine \( M^{*} \).
  After the decoding and the computation, the earlier stains get eliminated;
  new stains only arise in new islands, so they remain bounded again.
\end{step+}




\begin{step+}{step:heal.num-islands}
  Condition~\eqref{i:annotated.num-islands} on the number of islands is met.
\end{step+}
\begin{pproof}
  A crucial observation is that in the normal course of simulation, if the head encounters an
  island and must pass over it then the island will be eliminated,
  since repeated zigging will notice it again and again.
  
  We started with a clean configuration.
  A burst can leave an island \( I_{1} \) in a colony \( C \), if it happens in the last
  sweep of the work period.
  Much later the head may return, say from the left.
  Then it may not pass over \( I_{1} \) only if \( I_{1} \) was in the right turn region of \( C \)
  and after a work period it the head moves left again.
  In the last sweep it may leave another island \( I_{2} \), and this is the only
  way for two islands to arise.
  But in this case the cell simulated by \( C \) has \( \Pass=1 \).
  So when returning a third time, 
  the simulated head will continue right and will eliminate necessarily both islands.
  Before doing this it may create a new island, this way temporarily increasing
  the number of islands to three; however, after it leaves \( C \) there will be at most
  one island left.
\end{pproof}

\begin{step+}{step:heal.turned}
  The extended annotation satisfies property~\eqref{i:annotated.turned} of annotation.
\end{step+}
\begin{pproof}
  An interval of \( \Pass\ne 0 \) can be created by the healing procedure as it makes several sweeps
  needed for its stitching operations.
  Given the bound \( \Delta \) on the size of ambiguous areas, only \( \Delta \) stitches should be needed,
  so \( 3\Delta \) is a generous upper bound on the size of an interval of cells where new turns were made.
  Once an island is eliminated and its place is passed over, these traces of small turns are also erased,
  leaving only the ones at the bottom of zigs.
  But zigs are made only after every two steps of progress,
  so these new \( \Pass \) signs will (almost) never be consecutive.
  By the restriction on the player Range Extender, each extension starts with a configuration that
  is safe for turns, so has a
  bound of \( 3\Delta \) on the number of
  consecutive cells with \( \Pass\ne 0 \); the additional \( 3\Delta \) may increase this bound temporarily to
  \( 6\Delta \), but any new islands will give rise to at most \( 3\Delta \) consecutive cells with \( \Pass\ne 0 \).

  A similar argument applies to footprints of a big turn.
  The number of big turns in a work period is less than \( \Q^{2} \), so by the argument seen
  in Section~\ref{sec:feathering}, they give rise to no more than \( 2\log\Q \) new
  consecutive footprints of a big turn, say in the right turn region of a colony.
  Safety for turns imposed on the Range Extender allows \( 3\log\Q \),
  but even when counting \( 2\Q^{2} \) new ones in
  two new work periods, the bound \( 6\log\Q \) of~\eqref{i:annotated.turned} is satisfied.
  After two work periods, these footprints will necessarily be erased
  due to feathering in the simulated machine \( M^{*} \). 
\end{pproof} 
\end{Proof}

\section{Cleaning}\label{sec:cleaning}

This section will scale up the Spill Bound, Escape, Attack Cleaning
and Pass Cleaning properties of trajectories, proving them for the history \( (\eta^{*},\Noise^{*}) \)
decoded from a trajectory \( (\eta,\Noise) \).

\subsection{Escape}\label{sec:escape}

We will scale up the Escape property in Lemma~\ref{lem:escape} below;
here is an outline of the argument.
Consider some fault-free path during a time interval \( J \) (later we will allow a single
burst) over some space interval \( G \) of size \( |G|=\gamma\Q\B \).
For the times \( t\in J \), let \( K(t) \) denote the set of those clean points in \( G \) that
the head passed at least once since they were clean.
Then by Lemma~\ref{lem:safe-for-turns},
this set consists of intervals that are safe for small turns.
The goal is to show that the path will not stay too long in \( G \).
This will be since if it stays long then it enlarges \( K(t) \), and then builds up colonies in it.
These simulate the machine \( M^{*} \), which commands its head to swing wide according
to the program (in zigging or healing), and thus leave \( G \).
Initially, the clean intervals of \( K(t) \) can be created using the 
Pass Cleaning property of trajectories.
Every time the head leaves such an interval, the latter grows via the Attack Cleaning property;
so we will mainly be concerned with longer stays.
The notion of ``long'' will be chosen here to make sure that most of it has to be spent in
simulating \( M^{*} \), since both healing and rebuilding finish relatively fast.

\begin{definition}\label{def:directed-colony-pair}
  In the work period of a colony-pair, let us call the phase during which
  the \( \Info \) track and the \( \Drift \) track is updated, the \df{update phase}.
  We will say that the pair is \df{right-directed} if the following holds.

  If the age of all cells is before the update phase, and 
  the colony-pair represents a pair of cells \( a,b \) of the machine \( M^{*} \), then the
  transition function of \( M^{*} \) applied to \( a,b \) will direct the head right.

  If the age is after the update phase then the \( \Drift \) track contains 1.

  Left-directedness is defined similarly.
\end{definition}



Time intervals of length \( \Tu \) we may consider as \df{steps}, since under clean and
noiseless conditions,
the machine \( M \) will perform at least one step of computation during each.
Recall  \( \U_{k} = \U =c_{\U}\Q\passno^{9} \) in Definition~\ref{def:hier-params}.
This is an upper bound on the number of computation steps in one work period, even
allowing some calls for healing.
The lemmas below
follow the development of a maximal interval \( I(t) \) of \( K(t) \).
We will need some new, temporary concepts.
Recall the notion of a \df{restart event} from Section~\ref{sec:rebuilding}, when the leftward
marking process of rebuilding encounters a right-directed frontier and this causes
the rebuilding to restart.
In the present context, we will consider such restart events within \( K(t) \) only if
it was preceded by leftward marking that covered an area of at least \( \Z \) cells.

\begin{lemma}\label{lem:escape.restarts}
  The \( \ge\Z/4 \) cells involved in a restart event (or their descendants if they are shifted)
  will not be involved in any other restart event later.
\end{lemma}
\begin{proof}
  The restart event overwrites the frontier cells by the new frontier of rebuilding.
  Any new frontier created later can be left behind only due to small turn starvation.
  (For big turns, the rebuilding process does not leave behind the frontier zone.)  
  But when the head arrived from the right over an area of at least \( \Z \) cells, this area
  became safe for turns, and in the absence of bursts, it will remain so.
\end{proof}
\begin{corollary}\label{crl:escape.restarts}
  The total number of restart events is at most \( 4\gamma\Q/\Z \).
\end{corollary}

\begin{lemma}\label{lem:escape.inside-hole}
 \begin{alphenum}
 \item\label{i:escape.no-spill} No maximal subinterval of \( K(t) \) ever
   decreases by more than \( \CSpill\B \) on either side. 
\item\label{i:escape.long-stay}
  In any subinterval of  \( K(t) \) of size \( \le n\Q\B \),
  the amount of time the head can spend is at most the time spent on
  rebuildings (possibly interrupted but only if done so by restarts), plus \( 2 n\U\Tu \).
  \end{alphenum}
\end{lemma}
\begin{proof}
  \begin{enumerate}
  \item 
  \eqref{i:escape.no-spill} follows from the No Spill property of trajectories.
\item\label{i:escape.inside-hole.long-stay} On~\eqref{i:escape.long-stay}:  
  While no rebuilding starts, we can apply Lemma~\ref{lem:combined-heals}.
  Thus, eventually the computation in normal (or, similarly, booting)
  mode leads to the simulation of cells in \( M^{*} \).
  The program of \( M^{*} \), just like that of \( M \), proceeds by sweeps (zigging or healing).
  Even the shortest of these sweeps
  has size at least \( 2\beta\Q\B >\gamma\Q\B \), therefore a full sweep
  will exit \( I(t) \) in \( \le n\U\Tu \) steps.

  If any rebuilding has been started, it will finish in \( O(\Q\Z) \) steps unless
  interrupted by a restart event (see above) or big turn starvation, see Section~\ref{sec:feathering}.
  If it succeeds, it creates a colony-pair that simulates a cell-pair at the start of
  healing, hence starting a sweep of size \( >\beta\Q\B \) to the right.
  It may trigger rebuilding again, but this rebuilding
  results in a new colony-pair of the same kind, in \( O(\Q\Z) \) steps, at the place
  where it started, at least \( \Q\B \) to the right of the last working colony-pair.
  So exit happens again in \( \le n\U\Tu \) steps.

  Consider now rebuildings that don't succeed.
  We are counting the time spent on rebuilding interrupted by a restart separately, so
  consider the ones interrupted by big turn starvation.
  These can occur at most \( \gamma n/4 \) times,
  since they occur only after the head moved to a distance \( \ge 4\Q\B \) from
  the rebuilding center---removing the center of a new rebuilding at least this far
  from the previous one.
    \end{enumerate}
  \end{proof}

 An interval \( I \) of size \( \Q\B \) in \( \Int(K(t), 2\B) \) will be called a \df{manifest colony} if
it is healthy with the possible exception of having some rebuild marks (only for survey, not
for decision), and has undergone a complete simulation work period as part of a colony-pair in a
clean subinterval of \( K(t) \).
In a manifest colony, 
unless it is at a distance \( \le 2\Q\B \) from the head in the same interval of \( K(t) \),
the \( \Drift \) track points towards the head.

\begin{lemma}\label{lem:escape.non-decr}
  The number of manifest colonies does not decrease.
\end{lemma}
\begin{proof}
  Simulation or healing does not destroy any part of a colony.
  It may shift a colony, if the simulation work period of a colony-pair encounters a
  replacement situation, see Section~\ref{sec:transfer}.
  Let us see that rebuilding does not destroy them either.

  In the definition of manifest colonies we did allow some (possible leftover)
  rebuild survey marks, but not decision marks.
  In order to destroy a colony, the rebuilding process needs to create two decision tracks.
  One has to consider also the case when the rebuilding process is at one end of \( I(t) \),
  hence is fed some uncontrollable information.

  The head can exit during a rebuilding process only if it is in its starting
  stage, marking its interval of operation.
  Zigging along with attack cleaning
  implies that even if the head enters and exits \( K(t) \) multiply,
  by decision time the whole
  rebuilding interval will have to be incorporated into \( K(t) \), therefore the decision will
  be a correct one, not destroying a manifest colony.
\end{proof}

    Suppose that during a stay in \( K(t) \) a rebuilding process is started and completed.
  Its result is a pair of neighbor manifest colonies, on the left and right of the
  center from which rebuilding started.
  Let us call these the \df{left result} and \df{right result} of rebuilding.

  \begin{lemma}\label{lem:escape.results}
  A manifest colony can only become a left result once, and a right result once.
\end{lemma}
\begin{proof}
  The smallest interval \( D \) containing the resulting colony-pair will be healthy and safe
  for turns at the time when rebuilding finishes.
  The following development will never introduce inconsistency into \( D \),
  other than rebuild marks resulting from some rebuilding process started outside it.

  The only way in which a rebuilding process can start even in a healthy area is
  big turn starvation, as defined in Section~\ref{sec:feathering}.
  But in the present case, the rebuild process looking for a big turn must have started
  looking for a turn outside \( D \), and since \( D \) is safe for turns, it would have
  found a turning point close to an end of \( D \), so the left colony of \( D \) could not become
  its left result, nor the right colony of \( D \) its right result.
\end{proof}

  We will say that the stay of the head in some maximal interval of \( K(t) \) is \df{short}
  if the part of it remaining \emph{after subtracting the time
  spent on rebuild procedures interrupted by a restart}
is smaller than for \( 2\U\Tu \) for \( \U \) as in Definition~\ref{def:hier-params},
otherwise it is \df{long}.

\begin{lemma}\label{lem:escape.long-stays}
Each long stay either adds a new manifest colony, or
  joins two subintervals of \( K(t) \) of size \( \ge\Q\B \), or creates 
  a new left result and a new right result as defined in Lemma~\ref{lem:escape.results}.
\end{lemma}
\begin{proof}
  Consider some maximal interval \( I(t) \) of \( K(t) \), and a long stay in it.
  \begin{enumerate}
  \item\label{i:escape.create}
  Suppose first that no rebuilding process 
  is triggered or continued during the stay; then
  only healing and computation steps are possible.
  
  Suppose that there was no manifest colony in \( I(t) \) before entry.
  The stay is long enough that at least one complete work period will be performed on
  a neighbor colony-pair.
  So by the time the head leaves, there will be at least one manifest colony; in fact the
  exit will happen during a transfer process from a manifest colony (possibly slowed down
  by healing).
  
  In general, whenever the exit happens after a long stay
  then it either happens this way or during the marking stage of
  a rebuilding process.

  Suppose there were manifest colonies at entry time, and the head enters, say, on the left.
  Let \( C_{0} \) be the leftmost manifest colony of \( I(t) \).
  The head can pass to \( C_{0} \) only as a consequence of a
  transfer process from some colony \( C_{-1} \).
  So the long stay either adds \( C_{-1} \) as a new
  manifest colony, or joins \( I(t) \) with another subinterval of \( K(t) \) on its left, which
  contains a manifest colony \( C_{-1} \).
  
\item\label{i:escape.result}
  Suppose now that a rebuilding process starts before a manifest colony could have
  been added as described above.
  This process could be restarted repeatedly,
  but we don't count the time spent on the rebuilding processes interrupted by a restart.
  Since the stay is long, one of them has to succeed; on termination,
  it creates a left result and a right result.
\end{enumerate}  
\end{proof}

\begin{corollary}\label{crl:escape.long-stays}
  The number of long stays is at most \( 3\gamma \).
\end{corollary}
  \begin{proof}
  There are at most \( \gamma \) manifest colonies in \( K(t) \), so
  there can be at most \( \gamma \) creation events.
  There are at most \( \gamma-1 \) events of
  joining two disjoint subintervals of \( K(t) \) of size \( \ge\Q\B \).
  Hence the total number of long stays of kind~\ref{i:escape.create} is at most \( 2\gamma-1 \),
  and the total number of long stays of kind~\ref{i:escape.result}
  is also at most \( \gamma \).
  \end{proof}

The following lemma is the scale-up of the Escape condition.

\begin{lemma}[Escape]\label{lem:escape}
  Let  \( \escno \) be as introduced in~\eqref{eq:cns.traj}.
  In the absence of \( \Noise^{*} \), the
  head will leave any interval \( G \) of size \( \gamma\Q\B \),
  within time \( \escno^{*}\Tu^{*} \).
\end{lemma}
\begin{Proof}
  Consider a time interval of length \( \escno^{*}\Tu^{*} \) that the head
  spends in \( G \) in the absence of \( \Noise^{*} \).
  Because of the absence of \( \Noise^{*} \), at most one burst can happen during it.
  If it does then we will consider the larger part \( J \) of the time interval
  before or after the burst (or the whole interval if there is none).

  Let us apply Lemma~\ref{lem:create-holes} as well as its notation
  to the current situation, with \( t_{0} \) our starting time.
  The interval \( G \) is covered by 
\begin{align*}
 n=\gamma\Q/\gamma' =2\Q  
\end{align*}
blocks of size \( \gamma' \).
Assume that our noise-free path \( P \) starts
at some time \( t_{0} \) at a point \( b_{0} \).
Then there is a sequence of times \( t_{0}<t_{1}<\dots \) 
with \( t_{i+1}-t_{i}\le\escno\Tu \), such that during the skips
\( \rint{t_{i}}{t_{i+1}} \) the head passes over some block \( H_{i} \)
either leftwards or rightwards, further
except for \( 2 n\passno \) skips, the interval
\( H'_{i}=\Int(H_{i},(\CMarg+\CSpill)\B) \) is clean.

\begin{step+}{step:escape.mixed}
  The number of skips during which the head touches disorder is at most
  \( n(\passno + 4(\CMarg+\CSpill)) \), and this is also an upper bound on the
  number of (short or long) stays in \( K(t) \).
\end{step+}
\begin{pproof}
  We already estimated the number of skips \( i \) for which \( H'_{i} \)
  is not clean.
  The remaining skips pass over a clean \( H'_{i} \), but may touch disorder before
  or after it.
  If they do this then they will have to either enter a clean \( H'_{i} \) from disorder,
  or leave it.
  By the Attack Cleaning property, each leaving skip increases by \( \ge\B \) the clean interval
  it leaves.
  it follows that disorder in \( H \) will be eliminated after \( 2(\CMarg+\CSpill) \)
  leaving skips over some block \( H \) (ignoring integer parts).
  The entering skips will have to be balanced by leaving skips,
  so the total number of skips touching disorder in \( H \) is \( \le 4(\CMarg+\CSpill) \).

  Since each (long or short) stay ends with a skip that touches disorder, this is also the bound
  on the number of stays.
\end{pproof} 

\begin{step+}{step:escape.sum}
Let us add up all the estimates, using the notation \( \passno'=\passno+4(\CMarg+\CSpill) \).
\end{step+}
\begin{prooof}
  Part~\ref {step:escape.mixed} shows that the number of skips
  that are not clean is at most   \( n\passno' \), with \( n=2\Q \),
  for a total time of \(   2\Q\escno\Tu\passno' \).
  
  Corollary~\ref{crl:escape.long-stays} bounds the number of long stays by \( 3\gamma \),
  and Lemma~\ref{lem:escape.inside-hole} bounds the length
  of each long stay except for the time spent on interrupted rebuildings 
  by \( 2 \gamma\U\Tu \), so this gives a total time at most \(  6\gamma^{2}\U\Tu \).

 Corollary~\ref{crl:escape.restarts} bounds the number of restart events by \( 4\gamma\Q/\Z \).
Each rebuilding interrupted by restart has at most \( 12\gamma\Q\Z \) steps for the first
sweeps over the rebuilding area, so \( 48\gamma^{2}\Q^{2}\Tu \) bounds
the total time spent on rebuilding procedures interrupted by restarts.

Part~\ref{step:escape.mixed} bounds then number of stays, so the total
time spent in short stays is at most \(   2n\U\Tu\passno' = 4\Q\U\Tu\passno' \).

This gives the bound on the total time as \( \Tu \) multiplied with
\begin{align*}
  2\Q\escno\passno' +  6\gamma^{2}\U + 48\gamma^{2}\Q^{2} + 4\Q\U\passno'.
 \end{align*}
For \( \U \) as in Definition~\ref{def:hier-params},
the last term dominates the previous ones, so for large \( \Q \) this will be bounded by
\( \CEsc\Q\U\passno \) for an appropriate constant \( \CEsc \).
As by definition \( \escno^{*}=\CEsc\Q\passno \), this completes the proof of the lemma.
\end{prooof} 
\end{Proof}

\subsection{Weak attack cleaning}

This section will scale up the Attack Cleaning
property of trajectories (Definition~\ref{def:traj})
to machine \( M^{*} \), but first only in 
a weaker version, restricting the number of bursts in the relevant interval.

The Attack Cleaning property says the following for the present case.
Let \( P \) be a path that is free of \( \Noise^{*} \).
For current colony-pair \( \pair{x}{x'} \) (where \( x'< x+2\Q\B \)), suppose that the interval
\( I=\lint{x-(\CSpill+1)\Q\B}{x'+\Q\B} \) is clean for \( M^{*} \).
Suppose further that \( t \) is at the end of a work period in which
the transition function, applied to \( \eta^{*}(x,t) \), directs the head right.
Then by the time the head returns to \( x'-\CSpill\Q\B \),
the right end of the interval clean in \( M^{*} \)
containing \( x \) advances to the right by at least \( \Q\B \).

We will use the constant
\begin{align}\label{eq:E}
   \E =16\Delta\B
\end{align}
which bounds the size of the whole range in which a call to healing operates.
The constant \( \Delta \) was defined in~\eqref{eq:Delta}.

\begin{definition}\label{def:tame}
  Recall \( \passno^{*}=\passno + 5 \) from Definition~\ref{def:hier-params}.
   A path \( P \) is called \df{tame} over the interval \( I \) if
   during every time interval that it spends in \( I \) has at most one burst,
   with at most
\begin{align}\label{eq:s-def}
  \s=  2\passno^{*}(\passno^{*}+2^{\gamma/5+\CMarg + 2})
\end{align}
   bursts altogether (this is less than \( 3\passno^{2} \) for large \( \passno \)).
 \end{definition}

 \begin{lemma}[Weak attack cleaning]\label{lem:weak-attack-clean}
   \begin{alphenum}
     \item
   In addition of the above condition of attack cleaning,
   assume that the trajectory \( (\eta,\Noise) \) is tame over the interval \( I \).
   Then the conclusion holds.
   The analogous statement is also true when switching left and right.
   
 \item\label{i:weak-attack-clean.game}
   Assume that the annotation game has been played to the beginning of the attack.
   Then the Range Extender player can extend the range in \( I \) to the end of the attack,
   satisfying the conditions of the game.     
   \end{alphenum}
\end{lemma}
\begin{proof}
When the transfer phase of the simulation on the colony-pair \( \pair{x}{x'} \) 
begins, it may enter disorder to the right of \( x'+\Q\B \).



  
\begin{enumerate}
\item Assume an attack to the right.
  In the transfer process of the simulation, or if a rebuilding process is triggered later,
  the frontier zone is moving right.
  When head enters and later exits the disorder then it may create some new inconsistency.
  Moreover, every time the head exits disorder, since this may happen after a long-time
  absence, a burst may occur.
  Since the path is tame, the total number of bursts is limited to \(  \s \ll \Z \).
The length of the frontier zone at the end of \( J \) is \( \Z/2 \), so, with \( \E \) defined
in~\eqref{eq:E}, there are at least
  \( \Z/2-2\s\E \) cells of the frontier zone that are at a distance at least \( \E \) from
  all bursts.
  No healing will change the sweep or the \( \BigDigression \) field in any of these cells.
  So they can be overwritten only in normal or rebuilding mode.
  If rebuilding is started then it will move right.
The only way that the head can move much left if the frontier zone itself turns back.
This will only happen either at the end of transfer or at the end of a first sweep of rebuilding.
In the first case the simulation creates a new colony.
Its \( \Info \) track may be unusable, being damaged by many bursts, but 
during the rest of the simulation the head does not exit the colony pair,
(the transfer sweep went out to the end of the turn region), so at most one
new burst occurs, and the \( \rul{ComplianceCheck} \) part of the simulation forces
a compliant codeword by the end of the work period.

In the rebuilding case, the content of the new colonies is created from scratch anyway. 

\item Assume that the attack is to the left.
  Then for the same reason as above, new bursts cannot turn back
  the leftward moving front.
  It can be replaced repeatedly by a rebuilding front started on its left,
  which may also be overridden similarly.
  But the only way to arrive to \( x' + \CSpill\Q\B \) is to finish a started rebuilding,
  creating a new colony on the left, and also to erase the rebuilding marks
  in every cell marked for rebuilding in this process, making the whole passed-over area 
  healthy.

\item In order to satisfy~\eqref{i:weak-attack-clean.game}, let the Range Extender remove from \( \Rg \)
  at the beginning the whole area in which rebuilding happened.
  Then at the end, add back the whole area including the newly created colony; however,
  if a burst would occur in the last sweep then do this before the last sweep.
  This way, the condition will be satisfied that the addition to the range is clean.
  As the added area has been passed over several times in normal mode, the other
  conditions of cleanness and turn safety are also satisfied.

\end{enumerate}
\end{proof}

The following lemma draws a consequence of repeated applications of weak attack cleaning.

  \begin{lemma}\label{lem:weak-repeated-attack}
  Let \( I_{0} \) be an interval of size \(  \ge (2\CSpill+1)\Q\B \) and \( J \) an adjacent
  interval of size \( n\Q\B \) on its right.
  Consider a path \( P \) at whose beginning the interval
  \( I_{0} \) is clean for \( M^{*} \), and that is tame over \( I_{0}\cup J \).
  Assume that
  \( P \) passes \( I_{0} \) at least \( 2^{n+1} \) times from left to right and back.
  Then at some time during \( P \), the whole interval \( J \) becomes clean for \( M^{*} \).
  The analogous statement holds if we switch left and right.

  The statement analogous to part~\eqref{i:weak-attack-clean.game}
  of Lemma~\ref{lem:weak-attack-clean} also holds.
\end{lemma}
\begin{proof}
  Let \( I_{j} = \lint{a_{j}}{b_{j}} \) be
  the largest interval containing \( I_{0} \) after \( j \) pairs of (left-right, right-left) passes.
  We will concentrate on the growth of \( b_{j} \), though a similar analysis can show
  a simultaneous decrease of \( a_{j} \).
  We know from Lemma~\ref{lem:weak-attack-clean} that after every pair of passes
  over \( I_{0} \) (not necessarily on the larger interval \( I_{j} \)),
  the interval \( I_{0}=\lint{a_{0}}{b_{0}} \) will be clean for \( M^{*} \).
  Also \( b_{j}-b_{0}\ge \Q\B \), and  \( b_{j} \) is nondecreasing.

  Suppose now that \( b_{j}-b_{0}\le i\Q\B \);
  we claim that then \( b_{j+2^{i}}-b_{j}\ge\Q\B \).
  Indeed, applying Lemma~\ref{lem:feathering-lb} to the machine \( M^{*} \),
  after some \( j'\le 2^{i} \) left-right passes, the head must make an attack from the rightmost
  colony of \( I_{j+j'} \) allowing to apply Lemma~\ref{lem:weak-attack-clean}, and as a result,
  increasing \( b_{j'}\ge b_{j} \) by at least \( \Q\B \).

  Repeating the argument, we get \( b_{2^{n+1}}-b_{0}\ge n\Q\B \).
\end{proof}

\subsection{Pass cleaning}\label{sec:pass-cleaning}

The scaled-up version of the Pass Cleaning property
considers a path \( P \) with no \( \Noise^{*} \), as it makes
\( \passno^{*} \) pairs of passes over an the interval \( I \),
 and claims that they make \( \Int(I,\CMarg\Q\B) \) clean for \( \eta^{*} \).
 From now on, until further notice, consider a tame path \( P \) over an interval \( I \).
 Then \( I \) will be made up of subintervals of
 size \( \ge 4\Z\B \) that never gets bursts,
 separated from each other and the ends by distances \( \le 4\s\Z\B \).
 We will call these \df{basic holes}.
 
The pass cleaning property of \( (\eta,\Noise) \) cleans the basic
holes (except for margins of size \( \le\CMarg\B \)) in the first \( \passno \) pairs of passes.
The Spill Bound property allows them to erode on the edges by the amount \( \CSpill\B \).
We will call these somewhat smaller intervals still basic holes.
One more pair of passes
will make the basic holes, according to Lemma~\ref{lem:safe-for-turns}, safe for turns.
The following two lemmas will show how some order will be established on them in a constant
number of more passes.
Recall that the maximum number of cells in a healing area, \( \E  = O(\beta) \) in~\eqref{eq:E}
is a constant, much smaller than the zigging distance (in cell widths) defined in~\eqref{eq:FDef}.

\begin{definition}\label{def:directed}
  An interval will be called \df{almost clean} if
  it is clean except for a single island of size \( \le \beta'\B \) where \( \beta' \) was
  defined in~\eqref{eq:stain}.
  Let us call this island the \df{blemish}.

  An interval \( J \) of size \( >3\Z\B \) in the left direction from the head is
  called \df{right-directed} if
  \begin{itemize}
  \item it is almost clean;
  \item it is safe for turns;
  \item outside the blemish, its cells are all right-directed as seen by their sweep values (of
    simulation or rebuilding);
  \item its right end contains a frontier zone (of normal or rebuilding mode).
  \item this is the only frontier zone in \( J \).
  \end{itemize}
  Left-directedness is defined similarly.
\end{definition}

\begin{lemma}\label{lem:make-directed}
  Consider an almost clean interval \( J=\lint{a}{b} \) of size \( \ge 4\Z\B \) that
  is safe for turns.
  If a path with at most one burst passes it
  from left to right then it will leave a right-directed interval \( J'=\lint{a}{b'} \) with 
  \( b'\ge b-\CSpill\B \).
  The same is true when interchanging left and right.
\end{lemma}
\begin{proof}
  \begin{sloppypar}
  If rebuilding never starts then every possible healing that is triggered succeeds just as
  in Lemma~\ref{lem:combined-heals}, extending the healthy area.
  The disorder in the blemish will be corrected as the head passes it, but
  a new burst may leave a new blemish behind (if it happens at the bottom of a zig).
  At the time of exit, \( J' \) naturally becomes directed.    
  \end{sloppypar}

  Suppose that rebuilding gets triggered.
  If it exits on the right then it leaves \( J' \) directed.
  Here as well as in similar later situations,
  we use part~\eqref{i:rebuild.restart} of the definition of rebuilding,
  achieving that if a rebuilding is triggered within \( (1/4)\Z\B \) on the right of a frontier zone
  of normal mode then it moves back to start immediately from this frontier zone.
  A blemish or burst or a restart (by some frontier on the left)
  may trigger healing, but eventually a rebuilding will either exit or succeed.
  If it succeeds then a normal mode starts with a direction to the right,
  and/or possibly new rebuilding with a center at least
  \( \Q\B \) to the right of the old one.
  Eventually the head will exit leaving \( J' \) directed.
\end{proof}


Recall the definition of the feathering parameter \( \F \) in~\eqref{eq:FDef}.

\begin{lemma}\label{lem:keep-directed}
  Consider a tame path starting on the right of a right-directed interval
  \( J = \lint{a}{b} \), and eventually crossing it to the left.
  Then we end up with left-directed interval \( J' = \lint{a'}{b'} \),
  \( a'\le a+\CSpill\B \) and \( b'\ge b \) having within \( 2\Z \) cells of the right end
  a footprint of a big left turn (as defined in Section~\ref{sec:feathering}).

   If \( J \) already had such a footprint at position \( x \) then \( b'\ge x+\F\B \).
\end{lemma}
\begin{proof}
  The path can enter and exit \( J \) repeatedly, and is allowed at most one new burst every time,
  with a bound on the total number of bursts given in Definition~\ref{def:tame}.
  Now the same reasoning applies to the front as in the proof of Lemma~\ref{lem:weak-attack-clean}.
  If the head leaves \( J \) on the right then we end up with an interval \( J'\supseteq J \)
  that differs from a right-directed one only in possibly an area of size \( \le 2\E\s\B \) due to
  uncorrected islands caused by bursts, at most one in each entrance.
  The head can only leave on the left end
  if the frontier zone, either in normal or rebuilding mode,
  turns left on the right end of \( J' \), leaving the footprint of a big left turn.
  In all inconsistencies trigger healing or rebuilding; but as the head leaves on the left,
  these all must succeed with the exception of one possible blemish caused by a burst
  that occurred during the last right-to-left pass; we end up with left-directed interval.

  If rebuilding started then the frontier can move left only after it moved right
  by \( >\Q \) cells.
  In case of a footprint of a big left turn at the end of \( J \), if rebuilding does not start then
  the feathering property will force the front to move by at least \( \F \) cells to the right before
  passing \( J \) to the left.
\end{proof}

\begin{lemma}[Weak pass cleaning]\label{lem:weak-pass-clean}
  \begin{alphenum}
    \item Suppose that a tame path \( P \)
  makes \( \passno^{*} \) pairs of passes over the interval \( I \) starting from the left.
Then before the end of the \( \passno ^{*} \)th pair of passes the interior
\( \Int(I, \CMarg\Q\B) \) becomes clean for \( \eta^{*} \).
 \item\label{i:weak-pass-clean.game}
   Assume that the annotation game has been played to the beginning of \( P \).
   Then the Range Extender player can extend the range in \( I \) to the end,
   satisfying the conditions of the game.     
  \end{alphenum}
\end{lemma}

\begin{Proof}
We will number the right and left passes after the first \( \passno \) pairs of passes
  as r1, l1, r2, l2, \dots.
  \begin{step+}{step:weak-pass-clean.pass-clean}
    r1, l1 make the basic holes safe for turns, except
    for margins of size \( \le(\CMarg+\CSpill)\B \).
  \end{step+}
  \begin{pproof}
 As shown above, the basic holes, of size \( \ge 4\Z\B \) and
 separated from each other and the ends by distances \( \le 4\s\Z\B \),
 become and stay clean for \( \eta \) in the first \( \passno \) passes, except for margins
 of size \( \le(\CMarg+\CSpill)\B \).
 By Lemma~\ref{lem:safe-for-turns},  passes r1, l1 will make the basic holes safe for turns.
\end{pproof} 

Let us call at any time a subinterval of \( I \) 
an \df{essentially right-directed hole} if it can be turned into a right-directed
one by changing it in \( \s \) islands of size \( \le\E\B \),
and an \df{essentially foot-printed hole} if it has a footprint of a big left turn
within \( 2\Z\B \) of its right end, again except for these islands.

\begin{step+}{step:weak-pass-clean.traps}
  Lemma~\ref{lem:make-directed} shows that
  pass r2 turns each basic hole \( J \) into a right-directed hole, with possibly a small
  decrease on the left.

  Lemma~\ref{lem:keep-directed} shows that
  pass l2 turns each of these holes into a foot-printed one, and that possible subsequent intrusions
  from the left will leave it essentially foot-printed.

  The first time the head passes right over any one of these footprints, this will be conserved.
  During all the later parts of the path, feathering requires that the head can pass left over it
  only by first shifting it to the right by at least \( \F\B \).
  This will happen no later than during pass l3, thus by this time the right end of
  each hole will move by at least this much to the right of the right end of the
  basic hole it originates from.
  Since basic holes are separated by distances of \( \le 4\s\Z\B\ll\F \),
  by the end of pass l3 all holes will overlap,
  leaving the whole interval \( I \) left-directed.
  
  The intrusions between pass l3 and r4 may introduce some isolated islands and decrease
  \( I \) by the non-isolated ones, but pass r4 make it right-directed again, and also safe for turns.
  Pass l4 makes it left-directed, and leaves it still safe for turns.
\end{step+}

\begin{step+}{step:weak-pass-clean.finish}
  Pass r5 will clean \( \Int(I, \CRebuild\Q\B) \) for \( \eta^{*} \).
\end{step+}
\begin{pproof}
  The intrusions between pass l4 and r5 may again introduce some isolated islands and decrease
  \( I \).
  But during pass r5 all islands must be healed, except a single blemish left behind due to a new burst.
  As long as healing succeeds then, just as in Lemma~\ref{lem:combined-heals}
  after it the head returns to the front, and simulation continues.
  Suppose new rebuilding begins.
  This being a rightward pass, this does not happen
  so close to the left boundary that the left marking stage of the process would
  leave \( I \).
  Hence it could have been triggered only on the right side of a healthy left segment \( I' \) of \( I \),
  and after completion, it extends this segment to the right.
  It may encounter more islands just ahead of the rightward marking front.
  Then as specified in part~\eqref{i:rebuild.restart} of the description of rebuilding, after
  any triggered healing, whether successful or not, the rebuilding process will just continue.
  Eventually, the only unfinished rebuilding process can be one that exits on the right before
  finishing, so its area is outside \( \Int(I,\CRebuild\Q\B) \).
\end{pproof} 

\begin{step+}{step:weak-pass-clean.game}
  Part~\eqref{i:weak-pass-clean.game} of the lemma can also be satisfied.
\end{step+}
\begin{prooof}
  The proof is similar to the corresponding part of Lemma~\ref{lem:weak-attack-clean}.
  At the beginning, remove the whole interval \( I \) from the range \( \Rg \).
  Then in the last pass \( r5 \), add back \( \lint{a}{b}=\Int(I,\CRebuild\Q\B) \) to the range, but
  do it possibly in two steps.
  It there is no burst in this pass that creates an uncorrected island then add it all back at the end.
  Suppose there is such a burst, then it must be at the left end of a zig that started at a position
  \( x \) some \( \Z \) cells to the right of the burst.
  Say, the pass r5 is between times \( t_{1},t_{3} \), and the burst happens just after time \( t_{2} \)
  where \( t_{1}<t_{2}<t_{3} \).  
  Then at time \( t_{2} \) (when it is still healthy),
  add back interval \( \lint{a}{x} \) to the range \( \Rg \),
  and at time \( t_{3} \) add back the rest, \( \lint{x}{b} \).
\end{prooof} 
\end{Proof}


\begin{lemma}\label{lem:repeated-attack}
  Consider the statement of Lemma~\ref{lem:weak-repeated-attack} with the interval
  \( I_{0} \) having the same length \( n\Q\B \) with \( n=\gamma/5 \) as the interval \( J \).
  The conclusion holds also for non-tame paths.
\end{lemma}
\begin{proof}
  Let \( J_{1}= I_{0} \), \( J_{0}=J \).
    We will show that if the conclusion does not hold then the path passes over all the
  infinite sequence of consecutive adjacent intervals  \( J_{2},J_{3},\dots \) 
  on the left of \( J_{1} \), of size \( |J_{1}| \).
  Since the path is finite, this leads to a contradiction.
  Let \( i=1 \).
 \begin{enumerate}
 \item\label{i:repeated-attack.2}
   Suppose the conclusion of  Lemma~\ref{lem:weak-repeated-attack} does not hold.
   Then there are more than \( \s  \) bursts over 
\begin{align*}
 \lint{a}{b}=J_{i+1}\cup J_{i}  
\end{align*}
during this time,
   consequently at least \( \passno^{*}+2^{n+2}<\s/2^{n+2} \) bursts happened during
   some consecutive pair of the  \( 2^{k+2} \) rightward passes over \( J_{i+1} \).

\item\label{i:repeated-attack.choice}
  By the Escape property, the path cannot stay long in an interval of size
  \( \gamma\Q\B > 3n\Q\B \),
  so each burst is contained in a segment of the path covering an interval \( >3n\Q\B \)
  with no bursts in it.
  Since these segments don't pass over \( J_{i+1} \),  each contains a fault-free pass over
  \begin{align*}
    I = \lint{a - \gamma/5-\CMarg}{b}.
  \end{align*}  
  Suppose that \( \Int(I,\CMarg\Q\B)\supseteq J_{i+2} \) becomes clean
  at some time during the first \( \passno^{*} \) of these passes.
  There are still \( 2^{n+2} \) fault-free passes over \( J_{i+2} \), so we are
  back at the situation of part~\ref{i:repeated-attack.2} with \( i\gets i+1 \).

\item\label{i:repeated-attack.1}
  Suppose that \( J_{i+2} \) does not become clean during the first \( \passno^{*} \) passes.
  Then by Lemma~\ref{lem:weak-pass-clean}, the number of bursts 
  in \( J_{i+2} \) exceeds  \( \s \).
  Then at least \( 2(\passno^{*}+2^{n+2}) \) bursts happen over \( J_{i+2} \)
  between some consecutive pair of the left-right passes over \( J_{i+2} \).
  Using the Escape property similarly to the above,
  each burst belongs to a segment containing a fault-free pass over \( J_{i+3} \).
  This brings us back to the situation of part~\ref{i:repeated-attack.choice} with \( i\gets i+1 \).
\end{enumerate}
\end{proof}

Let us remove the bound on the number of bursts
in the Pass Cleaning property of \( \eta^{*} \)

\begin{lemma}[Pass cleaning]\label{lem:pass-clean}
  Let \( P \) be a space-time path without \( \Noise^{*} \) that makes
  at least \( \passno^{*} \) passes over an interval \( I \).
  Then there is a time during \( P \) when \( \Int(I,\CMarg\Q\B) \) becomes clean.
\end{lemma}
\begin{proof}
  We will prove the statement for \( |I|=(\gamma/5)\Q\B \).
  If \( |I| \) is larger we can cover it by intervals of size \( (\gamma/5)\Q\B \)
  overlapping by \( \CMarg\Q\B \): applying the statement simultaneously to each,
  it follows for \( I \).
  
  So assume \( |I|=(\gamma/5)\Q\B \) and let \( J_{1}= I \).
  We will show that if the conclusion does not hold then the path passes over all the
  infinite sequence of consecutive adjacent intervals  \( J_{2},J_{3},\dots \) 
  on the left of \( J_{1} \), of size \( |J_{1}| \).
  Since the path is finite, this leads to a contradiction.
Let \( i=1 \), \( n=\gamma/5+\CMarg \).

   \begin{enumerate}
  \item\label{i:pass-clean.1}
  By weak pass cleaning (Lemma~\ref{lem:weak-pass-clean}),
  if \( \Int(J_{i},\CMarg\Q\B) \) did not become clean for \( \eta^{*} \),
  the number of bursts
  in \( J_{i} \) is more than \( \s \) as in~\eqref{eq:s-def}.
  Then there is a time interval between two consecutive left-right passes over \( J_{i} \)
  with at least \( 2(\passno^{*}+2^{n+2}) \) bursts over \( J_{i} \).

  Consider one of the bursts and an interval of size \( \gamma\Q\B \) containing it in the middle.
  Using the Escape property similarly to the proof of Lemma~\ref{lem:repeated-attack},
  we conclude that the head will escape it without other bursts.
  So the path contains a burst-free segment of size 
\begin{align*}
 (\gamma/2-\gamma/5)\Q\B > (\gamma/5+\CMarg)\Q\B
\end{align*}
  either on the left or on the right of \( J_{i}=\lint{a}{b} \).
  Without loss of generality we can assume that at least half of them are on the left, giving
  \( \passno^{*}+2^{k+2} \)  noise-free passes over
  \( \lint{a-(\gamma/5-\CMarg)\Q\B}{b} \) during this time.
  Let \( J_{i+1}=\lint{a-(\gamma/5)\Q\B}{a} \).
  \item\label{i:first-choice} If  \( J_{i+1} \)
    does not become clean during the first \( \passno^{*} \) of these passes
  then restart the reasoning, going back to part~\ref{i:pass-clean.1}, setting \( i\gets i+1 \).
  Otherwise by Lemma~\ref{lem:repeated-attack},
  interval \( J_{i} \) becomes clean during the next \( 2^{n+2} \)
  noise-free passes over \( J_{i+1} \), contrary to the assumption.
\end{enumerate}
\end{proof}

\subsection{Attack cleaning and spill bound}

Let us remove the bound on the number of bursts in the scale-up of the Attack Cleaning property.

\begin{lemma}[Attack cleaning]\label{lem:attack-clean}
  Consider the situation of Lemma~\ref{lem:weak-attack-clean}.
  The conclusion holds also if the path is not tame.
\end{lemma}
\begin{proof}
  Consider the path \( P'\subseteq P \) containing the first \( \passno^{*}+2^{\CMarg+4} \) bursts.
  The Escape property, used similarly to the proof of Lemma~\ref{lem:repeated-attack}
  implies that \( P' \) passes the
  interval \( J \) of length \( (\gamma/2)\Q\B \) on the right of \( I \) this many times.
  The Pass Cleaning property then implies that \( \Int(J,\CMarg\Q\B) \) becomes
  clean for \( \eta^{*} \) during the first \( \passno^{*} \) passes of \( P' \).
  Then Lemma~\ref{lem:repeated-attack} (applied in the left direction)
  implies that within the next \( 2^{\CMarg+4} \) right-left passes,
  the disorder of \( \eta^{*} \) of length \( \le (1+\CMarg)\Q\B \) between the old clean interval
  ending at \( x'+\Q\B \) and the new one beginning at \( x'+(\CMarg+1)\Q\B \) will be erased.
    \end{proof}

Here is the scaled-up version of the spill bound property.

\begin{lemma}[Spill bound]\label{lem:spill-bound}
Suppose that an interval \( I \) of size \( > 2\CSpill\Q\B \) is clean for \( \eta^{*} \), and
let \( P \) be a path with no faults of \( \eta^{*} \).
Then \( \Int(I,\CSpill\Q\B) \) stays clean for \( \eta^{*} \).
\end{lemma}
\begin{proof}
  Without loss of generality, consider exits and entries of the path on the left of \( I \).
  Let \( C_{0},C_{1}\) be the two leftmost colonies in \( I \), where
  by definition \( C_{0} \) is at the very end of \( I \).
  The Spill Bound property of \( (\eta,\Noise) \) allows a spill of size \( \CSpill\B \) into \( I \).

  \begin{enumerate}
  \item 
    Assume first that the path is tame according to Definition~\ref{def:tame}.
  Let \( J \) be the largest interval on the left end of \( I \) such that every subinterval
  of size \( \E\B \) contains a burst; then \( |J|\le \s\E\B\ll\Z\B \), where \( \s \) is defined
  in~\eqref{eq:s-def}.
  Let \( I'=I\setminus J \).

  As long as no rebuilding is triggered 
  the islands created by bursts in \( I' \) do not affect admissibility.
  Indeed, without rebuilding the path just continues the simulation in \( I' \).
  During every entrance of the path in \( I \) at most one burst can happen within
  \( \gamma\Q\B \) of the end.
  If the island left by a burst is not corrected during the present intrusion then
  it will be corrected at a next one if the path passes over it.
  If the next intrusion just deposits an island next to it without correcting then this
  has to be at a big leftward turn.
  Then due to feathering, if the head ever gets next time,
  it will pass over these islands by at least \( \F\B \), and thus correct them.

\item Assume now that rebuilding is triggered: this can only happen in \( J \).
  It may be interrupted by exit on the left or a burst.
  The interruption by a burst would be only temporary, since the heal rebuilding procedure
  of Section~\ref{sec:rebuild-heal} deals with it.
  
  Still, a rebuilding may leave an island at its right end, due to a burst, and exit on the
  left end without finishing.
  Other rebuildings may leave other islands due to new bursts,
  so later rebuildings may encounter more, but total number of bursts is bounded by \( \s \),
  not enough to prevent the frontier zone of some rebuilding from eventually continuing and 
  returning.
  As all rebuilding processes to be considered here must be triggered in \( J \) or to the left of it,
  they can affect the health of an area of size at most \( \CRebuild\Q\B \) on the left.

\item If the path is not tame,
  then we can finish just as in the proof of Lemma~\ref{lem:attack-clean}.
  \end{enumerate}
\end{proof}




Part~\eqref{i:weak-attack-clean.game} of Lemma~\ref{lem:weak-attack-clean}
and~\eqref{i:weak-pass-clean.game} of Lemma~\ref{lem:weak-pass-clean} hold, of course,
also for the corresponding lemmas~\ref{lem:attack-clean} and~\ref{lem:pass-clean}.
One can conclude from them and the other lemmas of this section the following.

\begin{lemma}\label{lem:game}
  Given a trajectory \( (\eta,\Noise) \) of machine \( M \),
  the scaled-up history \( (\eta^{*},\Noise^{*})\)
  is a trajectory of machine \( M^{*} \).
  Moreover, as the annotation game of \( (\eta,\Noise) \) is played, whenever the
  Attack Cleaning or Pass Cleaning property is applied to the
  scaled-up trajectory \( (\eta^{*},\Noise^{*}) \), the annotation can be extended to the
  range cleaned up by these properties.
  
\end{lemma}

\section{Proof of the theorem}\label{sec:computation}

Above, we constructed a sequence of generalized Turing machines \( M_{1},M_{2}\dots \)
with cell sizes \( \B_{1},\B_{2},\dots \) where \( M_{k} \) simulates \( M_{k+1} \).
The sequences and dwell periods were also specified in Definition~\ref{def:hier-params}.
Here, we will use this construction to prove Theorem~\ref{thm:main}.

\subsection{Fault estimation}\label{sec:fault-estimation}

The theorem says that there is a Turing machine \( M_{1} \) that can reliably (in the defined sense)
simulate any other Turing machine \( \G \).
Before the simulation starts, the input \( x \) of \( \G \) must be encoded by a code depending on
its length \( |x| \).
We will choose a code that represents the input \( x \) as the information content of a
pair of cells of \( M_{r} \) for an appropriate \( r=r(x) \), and set their kind to \( \Booting \).
The code does not depend on the length of the computation to be performed, only on the input length.
At any stage of the computation there will be a highest level \( K \) such that a generalized Turing
machine \( M_{K} \) will be simulated, with its cells of the Booting kind.
We will denote the history of the computation by \( (\eta^{1},\Noise^{1})=(\eta,\Noise) \),
and its decodings by the recursion, as defined in Section~\ref{sec:annotation} by
\( (\eta^{k},\Noise^{k}) \) where
\( (\eta^{k+1},\Noise^{k+1})=((\eta^{k})^{*},(\Noise^{k})^{*}) \).

Recall the values of \( \Q_{k},\U_{k},\B_{k},\Tu_{k},\S_{k}\) in
Definition~\ref{def:hier-params}.
Let \(  \cH_{k}  \) be the event that no burst of level \( k \) occurs in the space-time region
\begin{align*}
 W_{k} = \bB(\pair{0}{0},\gamma\pair{\B_{k+1}}{\S_{k+1}}) .
\end{align*}
Lemma~\ref{lem:sparsity} bounds the probability of burst of level \( k \)
in any rectangle of type \( \bB(\x,\pair{\B_{k}}{\S_{k}}) \) by
\( p_{k} =\eps \cdot 2^{-1.5^{k-1}} \), giving
\begin{align*}
   \Prob(\neg\cH_{k}) = O(\U_{k}\Q_{k}p_{k}),
\end{align*}
with \( \Q_{k},\U_{k} \) in Definition~\ref{def:hier-params}.
This shows \( \sum_{k}\Prob(\neg\cH_{k}) = O(\eps) \).
From now on we assume that the event \( \bigcap_{k}\cH_{k} \) holds, since it holds
with probability \( 1-O(\eps) \).

As the computation continues (and the probability of some fault occurring
over the longer time increases),
the encoding level will be raised again and again, by the lifting
mechanism of Section~\ref{sec:booting}.
The configurations \( \eta^{k}(\cdot,0) \) are clean by definition for all levels \( k \).
Let \( \sigma_{k} \) be the (random) time when lifting to level \( k \) succeeded.
By definition \( \sigma_{r}=0 \).
All trajectory properties are lifted by the lemmas in the preceding sections.
Since \( \cH_{r+1} \) holds, the Transition Function property of trajectories
applies to \( \eta^{r} \) over the rectangle \( W_{r+1} \).
The booting and lifting steps of \( \eta^{r} \) will leave the head within the
space-time rectangle \( W_{r+1} \), so \( \sigma_{r+1}<\S_{r+1} \).
Also the lifted configuration \( \eta^{r+1}(\cdot,\sigma_{r+1}) \) is
clean and healthy as no \( r \)-level noise disturbed its creation.
By the same argument we get that the booting and lifting steps of \( \eta^{r+1} \)
will leave the head within \( W_{r+2} \), with \( \sigma_{r+2}<\S_{r+2} \),
the lifted configuration \( \eta^{r+2}(\cdot,\sigma_{r+2}) \) is clean and healthy.
And so on, this holds for all \( k \).

Suppose that the original simulated Turing machine \( \G \)
produces output \( y \) at its step \( t \) (there is no halting,
but the output in cell 0 will not change further).
There will be a smallest level \( s = s(t) \), depending only on the structure of the simulation,
such that in our history \( \eta \), for all \( k>s \) there is a time \( u \) between
\( \sigma_{k} \) and \( \sigma_{k+1} \) with \( \eta^{k}(0,u).\Output = y \),
that is the \( k \)th level simulation also outputs \( y \).
The times \( \sigma_{k} \) are random, but we will compute below an upper
bound \( f(t) \) on \( \sigma_{s(t)} \) that follows from the earlier assumptions.
Take an arbitrary \( t'>f(t) \).
For each \( k \), let
\( \cH'_{k}(t') \) be the event that no burst of level \( k \) appears in
\begin{align*}
 W'_{k} = \bB(\pair{0}{t'},\gamma\pair{\B_{k+1}}{\S_{k+1}}) .
\end{align*}
Just as above for \( \cH_{k} \), we can
assume that the event \( \bigcap_{k}\cH'_{k} \) holds, since it holds
with probability \( 1-O(\eps) \).
For each \( k \) let \( \sigma'_{k} \) be the (random) last time before \( t' \)
when the head of the simulated machine \( M_{k} \) reaches position 0.
Let \( s' \) be the largest \( k\ge s \) with \( \sigma_{k}<\sigma'_{k} \).

Then \( \cH'_{s'} \) implies   \( \eta^{s'}(0,\sigma'_{s'}).\Output=y \), that is
the output of the simulated computation at time \( \sigma'_{s'} \) on level \( s' \) is \( y \).
Now we will use Lemma~\ref{lem:game}, saying that the areas known to be clean
can also be annotated.
Then by part~\eqref{i:transition.no-burst} of  Lemma~\ref{lem:transition} and
by part~\ref{i:trickle-down} (the trickle-down) of the simulation procedure as described
in Section~\ref{sec:simulation-phase}, the absence of faults of level \( s'-1 \) while this
procedure operates, as implied by condition \( \cH_{s'-1} \),
implies \( \eta^{s'-1}(0,\sigma'_{s'-1}).\Output=y \).
Repeating the argument for all \( k<s' \)
we find \( \eta^{k}(0,\sigma'_{k}).\Output=y \),
so finally \( \eta(0,t').\Output=y \), with probability \( 1-O(\eps) \).

 \subsection{Space- and time-redundancy}\label{sec:redundancy}

 Even with the simple tripling error-correcting code, there is a constant \( \lambda>1 \) such that
 a colony of level \( k \) uses at most
 \( \lambda \) times more space than the amount of information contained in the
 cell of level \( k+1 \) that it simulates.
 Therefore if \( k \) is the level that needs to be simulated before an output of \( G \) can be reached
 then the space used at that time is at most \( \lambda^{k} \)
 times the space needed to just store the information.
 If \( G \) produces output at time \( t \) then its space need is bounded by \( t \),
 so the space need of the reliable simulation is at most \( \lambda^{k}t \).
 Suppose this is within a pair of \( k \)-level cells just created by booting.
 The size of cells of level \( k \) is, according to Definition~\ref{def:hier-params},
 \begin{align*}
   \Q_{1}\Q_{2}\dotsm\Q_{k-1}=c_{\Q}^{k}2^{1 + 1.2 + \dots + 1.2^{k-1}}=c_{\Q}^{k}2^{5\cdot 1.2^{k}},
 \end{align*}
 so they can simulate \( t \) steps of \( G \) if
 \( \lambda^{k}t = c_{\Q}^{k}2^{5\cdot 1.2^{k}} \).
 So \( k \) is about \( d\log\log t \) with \( d\approx 1/\log 1.2 \).
 This gives a bound 
 \begin{align*}
 \lambda^{k}\approx \lambda^{d\log\log t}=(\log t)^{\alpha}  
\end{align*}
on the space redundancy factor, for some \( \alpha>0 \).

The time redundancy can be estimated using the conclusions of Section~\ref{sec:length-work-period}.
It shows that the simulation on a given level of the Turing machine \( G \) incurs a redundancy that
is a multiplier
\begin{align*}
 O(\F\Z^{2}) = O(\passno^{8 + 4\rho})=O(\passno^{9})
\end{align*}
if \( \rho \) is small.
Recall \( \passno=5 k + O(1) \),
Multiplying these on all levels we get, for some \( \mu \), that the time redundancy on level \( k \) is
\begin{align*}
   \mu^{k}(k!)^{9}=2^{9k\log k+O(k)}.
\end{align*}
Again, if \( k=d\log\log t \) then this is less than
\begin{align*}
 (\log t)^{10\log\log\log t}.
 \end{align*}

\begin{remark}
  There is a mechanism more economical on storage, used in~\cite{GacsSorg01},
  with narrow \( \Work \) and \( \Hold[j] \) tracks but with some added time complexity.
  This allows a space redundancy factor \( 1+\delta_{k} \) with \( \prod_{k}(1+\delta_{k})<\infty \),
  yielding a constant space redundancy factor for the whole hierarchy.
\end{remark}

 \section{Discussion}

\paragraph{A weaker but much simpler solution}
If our Turing machine could just simulate a 1-dimensional
fault-tolerant cellular automaton, it would become
fault-tolerant, though compared to a fault-free Turing machine computation of length \( t \),
the slow-down could be quadratic.
(Such a solution would be only \emph{relatively} simpler, being a reduction to a complex, existing one.)
We did not find an easy reduction by just having the simulating
Turing machine sweep larger and larger
areas of the tape, due to the possibility of the head being trapped too long in some large disorder created by
the group of faults.
Trapping can be avoided, however, 
\emph{provided that the length \( t \) of the computation is known in advance}.
The cellular automaton \( C \) can have length \( t \) , and we could define
a ``kind of'' Turing machine \( T \) with a \emph{circular tape} of size \( t \) simulating \( C \).
The transition function of \( T \) would move the head to the right in every step
(with any backward movement just due to faults).

 \paragraph{Decreasing the space redundancy}
 We don't know how to reduce the time redundancy significantly, but
 the space redundancy can be apparently reduced to a multiplicative constant.
 Following Example~\ref{xpl:Reed-Solomon}, it is possible to
 choose an error-correcting code with redundancy that is only a factor \( \delta_{k} \)
 with \( 
   \prod_{k=1}^{\infty}(1-\delta_{k})>1/2
 \).
 This also requires a more elaborate organization of the computation phase described in
 Section~\ref{sec:simulation-phase} since
 the total width of all other tracks must be only some \( \delta_{k} \) times the width
 of the \( \Info \) track.
 For cellular automata, such a mechanism was described in~\cite{GacsSorg01}.

 \paragraph{Other models}
 There is probably a number of models worth exploring with more parallelism than Turing machines, but less
 than cellular automata: for example having some kind of restriction on the number of active units.
 On the other hand, a one-tape Turing machine seems to be the simplest computation model for which a reasonable
 reliability question can be posed, in the framework of transient, non-conspiring faults of constant-bounded
 probability.

 A simpler, universal computation model is the so-called \df{counter machine}.
 This has some constant number of nonnegative integer counters (at least two for universality), and an internal state.
 Each transition can change each counter by \( \pm 1 \), depends on both the internal state
 and on the set of those counters with zero value.
 A fault can change the state and can change the value of any counter by \( \pm 1 \).
 It does not seem possible to perform reliable computation on such a machine in any reasonable sense.
 The statement of such a result
 cannot be too simple-minded, since there is \emph{some} nontrivial task that such a machine can
 do: with \( 2n \) counters, it can remember almost \( 2n \) bits of information with large probability forever.
 Indeed, let us start the machine with \( n \) counters having the value 0, and the other \( n \) having some
 large value (depending on the fault probability \( \eps \)).
 The machine will remember forever (with large probability) which set of counters was 0.
 It works as follows (in the absence of a fault):
 at any one time, if exactly \( n \) values have value 0, then increase each nonzero counter by 1.
 Otherwise decrease each nonzero counter by 1.
 
 This sort of computation seems close to the limit of what counter machines can do reliably, but
 how to express and prove this?
 
\bibliographystyle{plain}
\bibliography{reli,gacs-publ}

\begin{thebibliography}{10}

\bibitem{AsarinCollins2005}
Eugene Asarin and Pieter Collins.
\newblock Noisy turing machines.
\newblock In Lu{\'i}s Caires, Giuseppe~F. Italiano, Lu{\'i}s Monteiro, Catuscia
  Palamidessi, and Moti Yung, editors, {\em Automata, Languages and
  Programming}, pages 1031--1042, Berlin, Heidelberg, 2005. Springer Berlin
  Heidelberg.

\bibitem{BennettThermodynComp1982}
Charles~H. Bennett.
\newblock The thermodynamics of computation -- a review.
\newblock {\em Intern. J. of Theor. Physics}, 21:905--940, 1981.

\bibitem{burstyTuring13}
Ilir \c{C}apuni and Peter G\'acs.
\newblock A {T}uring machine resisting isolated bursts of faults.
\newblock {\em Chicago Journal of Theoretical Computer Science}, 2013.
\newblock See also in arXiv:1203.1335. Extended abstract appeared in SOFSEM
  2012.

\bibitem{DurandRomashShenTiling12}
Bruno Durand, Andrei~E. Romashchenko, and Alexander~Kh. Shen.
\newblock Fixed-point tile sets and their applications.
\newblock {\em Journal of Computer and System Sciences}, 78:731--764, 2012.

\bibitem{Gacs1dim86}
Peter G\'acs.
\newblock Reliable computation with cellular automata.
\newblock {\em Journal of Computer System Science}, 32(1):15--78, February
  1986.
\newblock Conference version at STOC' 83.

\bibitem{GacsSorg01}
Peter G\'acs.
\newblock Reliable cellular automata with self-organization.
\newblock {\em Journal of Statistical Physics}, 103(1/2):45--267, April 2001.
\newblock See also arXiv:math/0003117 [math.PR] and the proceedings of STOC
  '97.

\bibitem{GacsReif3dim88}
Peter G\'acs and John Reif.
\newblock A simple three-dimensional real-time reliable cellular array.
\newblock {\em Journal of Computer and System Sciences}, 36(2):125--147, April
  1988.
\newblock Short version in STOC '85.

\bibitem{Kurd78}
G.~L. Kurdyumov.
\newblock An example of a nonergodic homogenous one-dimensional random medium
  with positive transition probabilities.
\newblock {\em Soviet Mathematics Doklady}, 19(1):211--214, 1978.

\bibitem{QianSoloveichikWinfree2011}
Lulu Qian, David Soloveichik, and Erik Winfree.
\newblock Efficient turing-universal computation with dna polymers.
\newblock In Yasubumi Sakakibara and Yongli Mi, editors, {\em DNA Computing and
  Molecular Programming}, pages 123--140, Berlin, Heidelberg, 2011. Springer
  Berlin Heidelberg.

\bibitem{Toom80}
Andrei~L. Toom.
\newblock Stable and attractive trajectories in multicomponent systems.
\newblock In R.~L. Dobrushin, editor, {\em Multicomponent Systems}, volume~6 of
  {\em Advances in Probability}, pages 549--575. Dekker, New York, 1980.
\newblock Translation from Russian.

\bibitem{VonNeum56}
John von Neumann.
\newblock Probabilistic logics and the synthesis of reliable organisms from
  unreliable components.
\newblock In C.~Shannon and McCarthy, editors, {\em Automata Studies}.
  Princeton University Press, Princeton, NJ., 1956.

\end{thebibliography}

\section{Appendix}\label{sec:appendix}

The examples below serve to motivate some complexities of the construction.


\begin{example}[Need for feathering]\label{xpl:need-feather}
  Some big noise can create a number of intervals \( I_{1},I_{2},\dots,I_{n} \)
  consisting of colonies of machine \( M_{1} \), each interval with its own simulated head,
  where the neighboring intervals are in no relation to each other.
  When the head is about to return from the end of \( I_{k} \)
  (never even to zig beyond it),
  a burst can carry it over to \( I_{k+1} \) where
  the situation may be symmetric: it will continue the simulation that \( I_{k+1} \) is performing.
  (The rightmost colony of \( I_{k} \) and the leftmost colony of \( I_{k+1} \) need not be complete:
  what matters is only that the simulation in \( I_{k} \) would not bring the head beyond its right end,
  and the simulation in \( I_{k+1} \) would not bring the head beyond its left end.)

  The head can be similarly captured to \( I_{k+2} \), then much later back from \( I_{k+1} \) to \( I_{k} \),
  and so on.
  This way the restoration of structure in \( M_{2} \) may be delayed too long.
\end{example}

\begin{example}[Two slides over disorder]\label{xpl:two-slides}
  This example shows the possibility for the head to slide twice over disorder without cleaning it.
  
Consider two levels of simulation as outlined in Section~\ref{sec:hier}: 
machine \( M_{1} \) simulates \( M_{2} \) which simulates \( M_{3} \).
The tape of \( M_{1} \) is subdivided into colonies of size \( \Q_{1} \).
A burst on level 1 has size \( O(1) \), while a burst on level 2 has size \( O(\Q_{1}) \).

Suppose that \( M_{1} \) is performing a simulation in colony \( C_{0} \).
An earlier higher-level burst may have created a large interval \( D \) of disorder
on the right of \( C_{0} \), even reaching into \( C_{0} \).
For the moment, let \( C_{0} \) be called a \df{victim} colony.
Assume that the left edge of \( D \) represents the last stage of a transfer operation to the right neighbor 
colony \( C_{0}+Q_{1} \).
When the head, while performing its work in \( C_{0} \), moves close to its right end, a
lower-level burst may carry it over into \( D \).
There it will be ``captured'', and continue the (unintended) right transfer operation.
This can carry the head, over several successful colony simulations in \( D \), to some
victim colony \( C_{1} \) on the right from which it will be captured to the right similarly.
This can continue over new and new victim colonies \( C_{i} \) (with enough space between them to
allow for new faults to occur), all the way inside the disorder \( D \).
So the \( M_{2} \) cells in \( D \) will fail to simulate \( M_{3} \).

After a while the head may return to the left in \( D \)
(performing the simulations in its colonies).
When it gets at the right end of a victim colony \( C_{i} \), a burst might move it back there.
There is a case when \( C_{i} \) now can just continue its simulation and then send the head
further left: when before the head was captured on its right,
it was in the last stage of simulating a left turn of the head of machine \( M_{2} \).

In summary, a high-level burst
can create a disordered area \( D \) which can capture the head and on which the head can slide
forward and back without recreating any level of organization beyond the second one.
\end{example}

The following example extends the above, showing the possibility of many
levels of malicious (dis-)organization.

\begin{example}[Many slides over disorder]\label{xpl:unbounded}
  Let us describe a certain ``organization'' of a disordered area in which an unbounded number of passes
  may be required to restore order.
For some \( n<0 \), let the cells of \( M_{1} \) at positions
\( x_{-\Q_{1}},\dots,x_{n} \), where \( x_{i+1}=x_{i}+\B_{1} \),
represent part of a healthy colony \( C(x_{-\Q_{1}}) \) starting at \( x_{-\Q_{1}} \), where \( x_{n} \)
is the rightmost cell of \( C(x_{-\Q_{1}}) \)
to which the head would come in the last sweep before
the simulation will move to the \emph{left} neighbor colony \( C(x_{-2\Q_{1}}) \).
Let them be followed by cells \( x_{n+1},\dots, x_{\Q_{1}-1},\dots\)
which represent the last sweep of a transfer operation to the \emph{right} neighbor colony \( C(x_{0}) \).
If the head is in cell \( x_{n} \), a burst can transfer it to \( x_{n+1} \).
The cell state of \( M_{2} \) simulated by \( C(x_{-\Q_{1}}) \) need to be in \emph{no relation} to 
the cell state of \( M_{2} \) simulated by \( C(x_{0}) \).
This was a capture of the head by a burst of \( M_{1} \) across the point 0, to the right.

We can repeat the capture scenario, say around points \( i \Q_{1}\Q_{2} \) for \( i=1,2,\dots \),
and this way cells of \( M_{3} \) simulated by \( M_{2} \) (simulated by \( M_{1} \))
can be defined arbitrarily, with no consistency needed between any two neighbors.
(We did not write \( i \Q_{1} \) just in case bursts are not allowed in neighboring colonies.)
In particular, we can define them to implement a \emph{leftward} capture scenario
via level 3 bursts at points \( i \Q_{1}\Q_{2}\Q_{3}\Q_{4} \), allowing to simulate arbitrary cells of \( M_{5} \)
with no consistency requirement between neighbors.
So \( M_{5} \) could again implement a rightward capture scenario, and so on.
In summary, a malicious arrangement of disorder and noise allows \( k \) passes
after which the level of organization is still limited to level \( 2 k + 1 \).
\end{example}

\begin{example}[Three islands]\label{xpl:3-islands}
  Suppose that the head has arrived at some colony-pair \( C_{0},C_{1} \) from the left,
goes through a work period and then passes to the right.
In this case, if no new noise occurs then we expect that
all islands found in \( C_{0},C_{1} \) will be eliminated by the healing procedure.
A new island \( I_{1} \) can be deposited in the last sweep.

Consider the next time (possibly much later), when the head arrives (from the right).
If it later continues to the left, then the situation is similar to the above.
Island \( I_{1} \) will be eliminated, but a new one may be deposited.
But what if the head arrived to the colony-pair \( C_{1},C_{2} \) and
turns back right at the end of the work period?
If \( I_{1} \) is not near the right end of \( C_{0} \), then
the head may never reach it to eliminate it; moreover, by the
feathering way of making turns,
it may add a new island \( I_{2} \) near on the right end of \( C_{0} \).

When the head returns a third time (possibly much later), 
from the right, feathering on the level of the simulated machine will
cause it to leave on the left.
Islands \( I_{1},I_{2} \) will be eliminated but a new island
\( I_{3} \) may be created by a new burst before, after or during the elimination.
So the healing procedure must count with possibly three islands
possibly in close vicinity to each other.
But at least one of these, namely \( I_{2} \), is near the end of \( C_{0} \), not in
the extended interior \( \Int(C_{0},\PadLen-\F\B) \).
\end{example}

\begin{example}[No healing in rebuilding]\label{xpl:heal-in-rebuild}
  This example shows the need for some healing of the rebuilding process itself.
  In it, restarting a rebuilding process on the occasion of every alarm
  prevents the scale-up of the Spill Bound property from \( \eta \) to \( \eta^{*} \).
  This property supposes that 
  an interval \( I=\lint{a}{b} \) of size \( > 2\CSpill\Q\B \) is clean for \( \eta^{*} \) and considers
  a path \( P \) be a path that has no faults of \( \eta^{*} \).
  It concludes that \( \Int(I,\CSpill\Q\B) \) stays clean for \( \eta^{*} \).
  We can exploit the fact that \( P \) has no faults of \( \eta^{*} \)
  only by the implication that during every time interval
  that the path spends in \( I \), it can have at most one burst.
  Let \( l = \CRebuild\Q\B \).

  Suppose that the path enters \( I \) on the right during a rebuilding process that
  (seems) started just on the outside of \( I \).  
  The process marks the interval \( \lint{a_{1}}{b} \) were \( b-a_{1}\approx l \).
  (We don't see what it does outside \( I \), on the right of \( b \)).
  Somewhere near \( a_{1} \), a burst causes alarm, restarting a rebuilding process which
  marks the interval \( \lint{a_{2}}{b} \) where \( a_{1}-a_{2}\approx l \),
  but the head leaves on the right of \( I \) before rebuilding finishes.

  Later, the head returns to continue the rebuild process, but a burst at position
  \( a'_{1}=a_{1}+l/2 \) causes alarm and triggers a new rebuilding process.
  This finishes, making the interval \( \lint{a'_{2}}{b} \) healthy,
  where \( a'_{2}\approx a_{2}+l/2 \).
  The interval \( \lint{a_{2}}{a'_{2}} \), of size \( \approx l/2 \), is still marked for rebuilding.

  Now an iterative process starts, creating marked intervals \( \lint{a_{i}}{a'_{i}} \),
  of size \( \approx l/2 \), where \( a'_{i}\approx a_{i-1} \).
  This way the disorder of \( \eta^{*} \) in \( I \) will not be confined to a subinterval
  at the right end of \( I \) as the Spill Bound requires.
  
  Suppose that we have a marked interval \( \lint{a_{i}}{a'_{i}} \) of size \( \approx l/2 \)
  such that \( \lint{a'_{i}}{b} \) is healthy.
  The head enters in normal mode, continuing a simulation until it reaches \( a'_{i} \).
  Then the marked cells it encounters trigger new rebuilding, which marks an interval
  \( \lint{a_{i+1}}{b'} \) where \( a_{i+1}\approx a_{i}-l/2 \), \( b'\approx a_{i+1}+2 l \).
  The new rebuilding process is interrupted by a burst at \( a'_{i}+l/2 \), starting
  a new rebuilding.
  This rebuilding finishes, leaving the interval \( \lint{a_{i+1}}{a'_{i+1}} \) marked
  where \( a'_{i+1}-a_{i}\approx l/2 \), and the the interval \( \lint{a_{i+1}}{b} \) healthy.
\end{example}

\end{document}